\documentclass[a4paper, reqno, 11pt]{amsart}        

\usepackage{amssymb}

\usepackage{epsfig}  		% For postscript
\usepackage{epic,eepic}       % For epic and eepic output from xfig

\let\oldtocsection=\tocsection
 
\let\oldtocsubsection=\tocsubsection
 
\let\oldtocsubsubsection=\tocsubsubsection
 
\renewcommand{\tocsection}[2]{\hspace{0em}\oldtocsection{#1}{#2}}
\renewcommand{\tocsubsection}[2]{\hspace{1em}\oldtocsubsection{#1}{#2}}
\renewcommand{\tocsubsubsection}[2]{\hspace{2em}\oldtocsubsubsection{#1}{#2}}

\usepackage{amsfonts}
\usepackage{amsthm}
\usepackage{amsmath}
\usepackage{amscd}
\usepackage[latin2]{inputenc}
\usepackage{t1enc}
\usepackage[mathscr]{eucal}
\usepackage{indentfirst}
\usepackage{graphicx}
\usepackage{graphics}
\usepackage{pict2e}
\usepackage{epic}
\numberwithin{equation}{section}
\usepackage[margin=2.9cm]{geometry}
\usepackage{hyperref}
\usepackage{dsfont}
\usepackage{csquotes}
\usepackage{stmaryrd}
\usepackage{blindtext}
\usepackage{tikz-cd}
\usepackage{mathrsfs}
\usepackage{cite}
\usepackage{mathtools}
\usepackage{extarrows}
\usepackage{epstopdf}
\usepackage[T1]{fontenc}
\usepackage{braket}

\theoremstyle{definition}
\newtheorem{definition}[equation]{Definition}
\newtheorem{example}[equation]{Example}
\newtheorem{proposition}[equation]{Proposition}
\newtheorem{theorem}[equation]{Theorem}
\newtheorem{remark}[equation]{Remark}

\newtheorem{corollary}[equation]{Corollary}
\newtheorem{lemma}[equation]{Lemma}

\numberwithin{equation}{section}

\newcommand{\midwedge}{\text{\Large$\wedge$}}
\newcommand{\midodot}{\text{\Large$\odot$}}

\newcommand{\dorf}{{\tt D}}
\newcommand{\cour}{{\tt C}}
\newcommand{\dual}{{\textrm{\tiny$\vee$}}}
\newcommand{\LC}{{\tt LC}}

\newcommand{\DFTLie}{{\boldsymbol\pounds}}

\newcommand{\be}{\begin{equation}}
\newcommand{\ee}{\end{equation}}
\def\beqa{\begin{eqnarray}}
\def\eeqa{\end{eqnarray}}
\def\bean{\begin{eqnarray*}}
\def\eean{\end{eqnarray*}}

\newcommand{\R}{\mathbb{R}}

\newcommand{\de}{\mathrm{d}}

\newcommand{\cdo}{\mathrm{\mathsf{CDO}}}

\newcommand{\IZ}{\mathbb{Z}}

\newcommand{\IN}{\mathbb{N}}
\newcommand{\IR}{\mathbb{R}}

\newcommand{\IT}{\mathbb{T}}

\newcommand{\IA}{\mathsf{At}}

\newcommand{\frX}{\mathfrak{X}}

\newcommand{\frso}{\mathfrak{so}}

\renewcommand{\Im}{\ensuremath{\mathsf{Im}}}
\renewcommand{\ker}{\ensuremath{\mathsf{Ker}}}

\newcommand{\cI}{{\mathcal I}}

\newcommand{\cN}{{\mathcal N}}
\newcommand{\cM}{{\mathcal M}}
\newcommand{\cS}{{\mathcal S}}

\newcommand{\cH}{{\mathcal H}}
\newcommand{\cA}{{\mathcal A}}

\newcommand{\cQ}{{\mathcal Q}}

\newcommand{\cF}{{\mathcal F}}
\newcommand{\cD}{{\mathcal D}}

\newcommand{\ccT}{{\mathscr T}}

\newcommand{\sfa}{{\mathtt{a}}}
\newcommand{\sfi}{{\mathtt{i}}}

\newcommand{\sfp}{{\mathtt{p}}}

\newcommand{\sfq}{{\mathtt{q}}}

\newcommand{\unit}{\mathds{1}}   			% identity map/matrix

\begin{document}

\title[Algebroids, AKSZ Constructions and Doubled Geometry]{Algebroids, AKSZ Constructions and Doubled Geometry}

\author[V.~E.~ Marotta]{Vincenzo Emilio Marotta}
\address[Vincenzo Emilio Marotta]
{Department of Mathematics and Maxwell Institute for Mathematical
  Sciences\\ Heriot-Watt
  University\\ Edinburgh EH14 4AS\\ United Kingdom}
\email{vm34@hw.ac.uk}

\author[R.~J. Szabo]{Richard J.~Szabo}
  \address[Richard J.~Szabo]
  {Department of Mathematics, Maxwell Institute for Mathematical Sciences and Higgs Centre for Theoretical Physics\\
  Heriot-Watt University\\
  Edinburgh EH14 4AS \\
  United Kingdom}
  \email{R.J.Szabo@hw.ac.uk}

\vfill

\begin{flushright}
\footnotesize
{\sf EMPG--21--05}
\normalsize
\end{flushright}

\vspace{1cm}

\begin{abstract}
We give a self-contained survey of some approaches aimed at a global description of the geometry underlying double field theory. After reviewing the geometry of Courant algebroids and their incarnations in the AKSZ construction, we develop the theory of metric algebroids including their graded geometry. We use metric algebroids to give a global description of doubled geometry, incorporating the section constraint, as well as an AKSZ-type construction of topological doubled sigma-models. When these notions are combined with ingredients of para-Hermitian geometry, we demonstrate how they reproduce kinematical features of double field theory from a global perspective, including solutions of the section constraint for Riemannian foliated doubled manifolds, as well as a natural notion of generalized T-duality for polarized doubled manifolds. We describe the $L_\infty$-algebras of symmetries of a doubled geometry, and briefly discuss other proposals for global doubled geometry in the literature.
\end{abstract}

\maketitle

\begin{center}
{\sl\small Contribution to the Special Issue of Complex Manifolds on `Generalized Geometry'}
\end{center}

\medskip

{%\baselineskip=12pt
\tableofcontents
}

%\bigskip

\section{Introduction}
This contribution is a relatively self-contained survey of some mathematical approaches to a rigorous global formulation of the geometry underlying double field theory, that we will colloquially call `doubled geometry', following standard terminology from string theory (more precise definitions will be given in Section~\ref{sec:doubledgeom}). Double field theory is an extension of supergravity in which stringy T-duality becomes a manifest symmetry. The basic example of a doubled geometry in this context comes from considering toroidal compactifications of string theory, which we shall now briefly review.

\medskip

\subsection{T-Duality and Doubled Geometry} ~\\[5pt]
\label{subsec:introTduality}
Let $V$ be a $d$-dimensional real vector space, and let $\Lambda$ be a lattice of $V$. The symmetry group of string theory with target space the $d$-dimensional affine torus $\cQ=V/\Lambda$ is isomorphic to the integer split orthogonal group ${\sf O}(d,d;\IZ)$; it preserves a flat split signature metric $\eta$ induced by the canonical pairing between the lattice $\Lambda\subset V$ and its dual lattice $\Lambda^*\subset V^*$. This contains the geometric subgroup ${\sf GL}(d,\IZ)\subset{\sf O}(d,d;\IZ)$ generated by large diffeomorphisms of the torus $\cQ$, while the rest of the group is generated by T-dualities combined with integer shifts of the Kalb-Ramond $B$-field which are not geometric symmetries of $\cQ$. However, T-duality does act geometrically on the doubled torus $M:=(V\oplus V^*)/(\Lambda\oplus\Lambda^*)\simeq\cQ\times\tilde\cQ$, where $\tilde\cQ=V^*/\Lambda^*$ is the dual torus: ${\sf O}(d,d;\IZ)$ is a subgroup of the group of large diffeomorphisms ${\sf GL}(2d,\IZ)$ of $M$. In this sense string theory ``sees'' a doubled geometry.

Let $\sfq:M\to\cQ$ and $\tilde\sfq:M\to\tilde\cQ$ be the canonical projections. The doubled torus $M$ has a canonical symplectic form $\omega$ when viewed as the dual torus bundle $\sfq:M\to\cQ$, and a pair of involutive Lagrangian distributions $L_+=\ker(\de\tilde\sfq)$ and $L_-=\ker(\de\sfq)$ (i.e. real polarizations of $(M,\omega)$) such that $TM\simeq L_+\oplus  L_-$. As we will discuss in Section~\ref{sec:parahermgeom}, this is a simple example of a `para-K\"ahler manifold'. Then there is a pair of Lagrangian fibrations
\begin{equation}\label{eq:correspondence}
\begin{tikzcd}
 & M \arrow[dl,"\sfq",swap] \arrow[dr,"\tilde\sfq"] & \\
 \cQ &  & \tilde\cQ
\end{tikzcd}
\end{equation}
which yields a Lagrangian correspondence between the torus $\cQ$ and its dual torus $\tilde\cQ$; this defines a T-duality which swaps $\cQ$ with $\tilde\cQ$. Clearly there are different polarizations, corresponding to different choices of splitting $V\oplus V^*$, and in general factorized T-dualities swap only some of the fibre directions.

More generally, if $\pi:\cQ\to W$ is a principal torus bundle whose typical fiber is a $d$-dimensional torus, endowed with a torus-equivariant gerbe with connection on $\cQ$ of curvature $H\in\Omega^3(\cQ)$ (which models the NS--NS $3$-form flux in string theory), then the fibrewise T-duality group acts geometrically on a doubled torus bundle $M\to W$ with fibres of dimension~$2d$~\cite{Hull2005,Belov:2007qj}. If the T-dual is another principal torus bundle $\tilde\pi:\tilde\cQ\to W$, with an equivariant gerbe with connection on $\tilde\cQ$ of curvature $\tilde H\in\Omega^3(\tilde\cQ)$, then the correspondence space of \eqref{eq:correspondence} is homeomorphic to the fibred product $M\simeq\cQ\times_W\tilde\cQ$ with the principal doubled torus fibration $\pi\circ\sfq = \tilde\pi\circ\tilde\sfq:M\to W$. It has a fibrewise non-degenerate $2$-form $\omega\in\Omega^2(M)$ which is invariant under both torus actions on $\cQ$ and $\tilde\cQ$, and which obeys~\cite{Svoboda:2020msh}
\begin{align*}
\de \omega=\sfq^*H-\tilde\sfq^*\tilde H \ . 
\end{align*}
This is an example of a `para-Hermitian fibration' (see Section~\ref{sec:parahermgeom}), and it defines a topological T-duality between the principal torus bundles $\pi:\cQ\to W$ and $\tilde\pi:\tilde\cQ\to W$~\cite{Bouwknegt:2003vb,Cavalcanti2010}. These correspondence spaces were extended to doubled twisted tori in~\cite{Hull:2007jy}, which further double the base $W$, giving examples of `almost para-Hermitian manifolds' (see Section~\ref{sec:parahermgeom}), and provide a geometrization of the non-geometric T-duals that may arise (such as the `T-folds' of~\cite{Hull2005}); see~\cite{Aschieri:2020uqp} for an alternative viewpoint on these constructions in the language of $C^*$-algebra bundles and noncommutative correspondences.

\medskip

\subsection{Supergravity and Courant Algebroids} ~\\[5pt]
Supergravity is the low-energy approximation to string theory. It has long been appreciated that the geometry underlying type~II supergravity is generalized geometry on Courant algebroids~\cite{Hitchin2011,gualtieri:tesi}: the complete bosonic field content (in the NS--NS sector) can be encoded in a generalized metric on an exact Courant algebroid~\cite{Grana2008,Coimbra:2011nw}. Exact Courant algebroids over a manifold $\cQ$ have underlying vector bundle $E\simeq T\cQ\oplus T^*\cQ$ and are classified by the class of the $3$-form $H$-flux in ${\sf H}^3(\cQ,\IR)$~\cite[Letter~1]{Severa-letters} (see Sections~\ref{sec:AKSZ} and~\ref{sec:metricalg}). In this sense generalized geometry ``doubles'' the tangent bundle $T\cQ$, which captures diffeomorphisms and $B$-field gauge transformations as transition functions, and hence are manifest symmetries of supergravity. 

However, factorized T-dualities relate supergravity in different duality frames. This is reflected mathematically in the feature that topological T-duality between principal torus bundles can be implemented, using the correspondence \eqref{eq:correspondence}, as an isomorphism between exact Courant algebroids~\cite{Cavalcanti2010}, but not generally as a symmetry of a single exact Courant algebroid. Hence supergravity is not manifestly T-duality invariant.

\medskip

\subsection{Double Field Theory and Para-Hermitian Geometry} ~\\[5pt]
In double field theory, one instead ``doubles'' the underlying $d$-dimensional manifold $\cQ$ to a manifold $M$ of dimension $2d$, and considers geometry on the tangent bundle $TM$ (see Section~\ref{sec:applications}). What this doubling means exactly will be defined precisely in this paper, but the rough idea is as follows. Double field theory is a constrained theory, whose constraint follows from the level matching condition in string theory. At present this constrained theory is not very well understood, but its reduction under a stronger constraint, called the `section constraint', has been extensively studied. Solving the section constraint amounts to selecting a `polarization' which reduces the geometry on $TM$ to generalized geometry on an exact Courant algebroid. What the doubled geometry of $M$ accomplishes is that its group of (large) diffeomorphisms contains the T-duality group in $d$-dimensions, and in this way T-duality becomes a manifest symmetry of the unconstrained double field theory. In the example of the doubled tori or doubled torus bundles $M\to W$ from Section~\ref{subsec:introTduality}, double field theory on $M$ can be reduced in this way to string theory on a torus or a T-fold~\cite{Hull2005}.

Such a duality covariantization of supergravity, with manifest ${\sf O}(d,d)$ symmetry, was suggested some time ago by Siegel~\cite{Siegel1993a,Siegel1993b}. A theory with manifest ${\sf O}(d,d;\IZ)$ symmetry was later shown to arise naturally as a consequence of string field theory on a $d$-dimensional torus by Hull and Zwiebach~\cite{HullZw2009}. One of the goals of the programme that we outline in this contribution is to write double field theory on more general doubled manifolds $M$, and to understand the meaning of the doubling for general string target spaces $\cQ$. This can be achieved by using the symmetries and geometry of double field theory to define a particular type of metric algebroid~\cite{Vaisman2012}, which we describe in Section~\ref{sec:doubledgeom} and call a `DFT algebroid' following the terminology of~\cite{Jonke2018}, and encoding the bosonic fields in a generalized metric on a DFT algebroid and their dynamics by the vanishing of a suitable Ricci tensor~\cite{hohmhz,hullzw}. 

In this contribution we aim to describe the geometric origin of the ingredients of double field theory and its section constraint, as well as its precise geometric relation with generalized geometry, in the language of algebroids, which allows us to import techniques and ideas known from the more thoroughly studied Courant algebroids. We will discuss other approaches to global double field theory, and compare them to our perspectives, at appropriate places throughout the paper, together with many more references to the pertinent literature. We focus only on the kinematical aspects of the theory in the present paper.

As alluded to in Section~\ref{subsec:introTduality}, a prominent ingredient in our treatment of doubled geometry is the notion of a para-Hermitian structure, which we discuss in Section~\ref{sec:parahermgeom}, and in particular the formulation of double field theory on almost para-Hermitian manifolds, which we discuss in Section~\ref{sec:applications}. Para-Hermitian geometry can be roughly thought of as a ``real version'' of complex Hermitian geometry. It has proven to be a suitable framework for addressing global issues of doubled geometry, while providing a simple and elegant description of generalized flux compactifications and non-geometric backgrounds in string theory. The relevance of para-Hermitian structures in doubled geometry was originally noticed by Hull~\cite{Hull2005} (who called them `pseudo-Hermitian structures'), and was later put forward in a rigorous framework by Vaisman~\cite{Vaisman2012}. Interest in the formalism was rekindled by Freidel, Leigh and Svoboda~\cite{Freidel2014} which led to some flurry of activity in the literature, see e.g.~\cite{Svoboda2018,SzMar,Mori2019,Hassler:2019wvn}. 

From this modern perspective, para-Hermitian geometry involves developing the interplay between the well-studied geometry on exact Courant algebroids and the less understood geometry on the tangent bundle of an almost para-Hermitian manifold, equipped with the structure of a DFT algebroid. The most prominent examples of almost para-Hermitian manifolds in the literature are total spaces of fibre bundles, such as the cotangent bundle $T^*\cQ$ and the tangent bundle $T\cQ$ of a manifold $\cQ$, group manifolds of doubled Lie groups and Drinfel'd doubles, and the quotients of all these by discrete group actions, which includes the basic doubled torus and doubled twisted torus examples discussed in Section~\ref{subsec:introTduality}.

We mention that para-Hermitian geometry also has a brief history of other applications to physics. Para-K\"ahler structures appear in the special geometry of $\mathcal{N}=2$ vector multiplets in Euclidean spacetimes~\cite{Cortes2004,Cortes:2009cs}. In~\cite{SzMar} it was shown that para-Hermitian geometry offers an alternative geometrical formulation of both Lagrangian and non-Lagrangian dynamical systems which is more natural than the commonly employed Finsler geometry. Generalized para-K\"ahler structures and Born structures also appear respectively in target space geometries for doubled sigma-models with $\cN=(2,2)$ twisted supersymmetry and $\cN=(1,1)$ supersymmetry in~\cite{AbouZeid:1999em,Stojevic:2009ub,Hu:2019zro}.

\medskip

\subsection{Graded Geometry and AKSZ Theory} ~\\[5pt]
In our development of geometry on certain classes of algebroids, we shall consider their incarnations in graded geometry which leads to generalizations of the AKSZ construction of topological field theories. AKSZ sigma-models capture the topological sectors of physical string theory sigma-models for target spaces with background NS--NS fields, such as the $B$-field or the $H$-flux. They are based on the structure maps of algebroids and allow for a quantization of the underlying algebroid through the BV formalism; this is explained in Section~\ref{sec:AKSZ}. They also allow for a better systematic description of the symmetries of algebroids, through their reformulations in terms of dg-manifolds and $L_\infty$-algebras. 

In Section~\ref{sec:AKSZ} we discuss this in some detail for the case of Courant algebroids; in the case of exact Courant algebroids, the corresponding AKSZ sigma-models describe the coupling of closed strings to (geometric and non-geometric) tri-fluxes. In Section~\ref{sec:doubledgeom} we discuss an extension of the AKSZ theory that writes down a topological doubled sigma-model, which unifies geometric and non-geometric fluxes with manifest T-duality invariance~\cite{Jonke2018}. 

Along the way, we present a new version of the correspondence between metric algebroids and graded geometry in Section~\ref{sec:metricalg} (see Theorem~\ref{thm:1-1metric}), which is entirely geometric and avoids any explicit coordinate description. It uses more recent developments on the geometrization of degree~$2$ manifolds based on double vector bundles and VB-algebroids. This lends a more detailed understanding of the gauge symmetries underlying metric algeboids, and their counterparts in double field theory, as well as a clearer connection with other approaches to double field theory based on graded geometry~\cite{Samann2018,Heller2016}. In particular, it provides a more concise picture of the various weakenings of the axioms of a Courant algebroid described in~\cite{Jonke2018} and their role in the geometry of double field theory.

\medskip

\subsection{Outline of the Paper} ~\\[5pt]
The organization of the remainder of this paper is as follows. In Section~\ref{sec:AKSZ} we introduce general notions of algebroids, culminating in Lie algebroids and Courant algebroids. We also develop their formulations as symplectic Lie $n$-algebroids in graded geometry and the corresponding AKSZ sigma-models (for $n=0,1,2$), together with their gauge symmetries which can be formulated in terms of flat $L_\infty$-algebras. In Section~\ref{sec:metricalg} we discuss the weakening of the notion of Courant algebroid to that of a metric algebroid, and give a new geometric formulation of a metric algebroid as a symplectic $2$-algebroid in graded geometry. In Section~\ref{sec:parahermgeom} we discuss basic aspects of para-Hermitian geometry, and in particular we introduce the canonical metric algebroid which plays a central role in the applications to double field theory. In Section~\ref{sec:doubledgeom} we give a rigorous account of doubled geometry, introducing the notion of DFT algebroid. This has a broader notion of gauge symmetry that can be formulated in terms of curved $L_\infty$-algebras, and we demonstrate how the AKSZ construction can be extended to define a topological sigma-model for a doubled geometry. Finally, in Section~\ref{sec:applications} we describe how everything fits together to give a rigorous formulation of some of the main ideas of double field theory, and in particular how to solve the section constraint in a completely geometric and coordinate-independent manner. We give a detailed account of how DFT algebroids reduce to Courant algebroids in different polarizations of a foliated doubled manifold, how a conventional string background, including the NS--NS fields, is recovered in the language of Riemannian foliations, and how T-duality is manifested in this framework.

\medskip

\subsection{Glossary of Notation and Conventions} ~\\[5pt]
$\cQ:$  manifold --- all manifolds are smooth second countable para-compact Hausdorff manifolds of finite and non-zero dimension; 

\noindent $M:$ even-dimensional  manifold;

\noindent $\cM=(\cQ, \cA):$ graded manifold $\cM$ with body $\cQ$ and sheaf of functions $\cA\,;$

\noindent $|\,\cdot\,|:$ degree of a homogeneous element;

\noindent $\cA^k:$ sheaf of homogeneous functions of degree $k\,;$

\noindent $C^{\infty}(\,\cdot\,):$ space of smooth functions on a (graded) manifold;

\noindent $\Omega^\bullet (\,\cdot\,):$ space of differential forms on a (graded) manifold;

\noindent $\frX(\,\cdot\,):$ sheaf of vector fields;

\noindent $\frX_{k}(\,\cdot\,):$ sheaf of homogeneous vector fields of degree $k;$

\noindent $E \rightarrow \cQ:$ vector bundle $E$ over $\cQ $ --- all vector spaces and vector bundles are considered over the ground field $\IR\,;$

\noindent $\mathsf{\Gamma}(E):$ $C^{\infty}(\cQ)$-module of sections of $E\to\cQ\,;$

\noindent ${\sf Aut}(E):$ group of automorphisms of a vector bundle $E\to\cQ$ which cover the identity map $\unit_\cQ:\cQ\to\cQ\,;$

\noindent $(\,\cdot\,)^{\rm t}:$ transpose of a vector bundle morphism;

\noindent $(\,\cdot\,)^\sharp:$ vector bundle morphism $E^*\to E$ induced by a $(2,0)$-tensor in $\mathsf{\Gamma}(E\otimes E)\,;$

\noindent $(\,\cdot\,)^\flat:$ vector bundle morphism $E\to E^*$ induced by a $(0,2)$-tensor in $\mathsf{\Gamma}(E^*\otimes E^*)\,;$

\noindent $\Im(\,\cdot\,):$ range of a vector bundle morphism;

\noindent $\ker(\,\cdot\,):$ kernel of a vector bundle morphism;

\noindent $\langle\,\cdot\,,\,\cdot\,\rangle:$ duality pairing for a vector bundle and its dual;

\noindent $E[k]:$ vector bundle $E$ whose fibres are shifted in degree by $k \in \IZ \, ;$

\noindent ${\sf Span}_\IR(\,\cdot\,):$ $\IR$-linear span of a set of vectors;

\noindent $\odot:$ symmetric tensor product;

\noindent $\wedge:$ skew-symmetric tensor product;

\noindent $[\, \cdot \,  ,\, \cdot \, ]_\circ :$ commutator bracket with respect to the composition $\circ \, ;$ 

\noindent $\iota :$ interior multiplication;

\noindent $\pounds :$ Lie derivative;

\noindent ${\sf Der}(A):$ $A$-module of derivations of a commutative algebra $A\,$;

\noindent $X \cdot  f:$ action of a vector field $X \in \mathsf{\Gamma}(T\cQ)$ as a derivation on a function $f \in C^{\infty}(\cQ)\, .$

\medskip

\subsection{Acknowledgments} ~\\[5pt]
We thank Thomas~Strobl and Marco~Zambon for helpful discussions. R.J.S. thanks the organisors Vicente~Cort\'es, Liana~David and Carlos~Shahbazi for the invitation to deliver a talk in the Workshop ``Generalized Geometry and Applications'' at Universit\"at Hamburg in March 2020, and to contribute to this special issue. The work of V.E.M. is funded by the STFC Doctoral Training Partnership Award ST/R504774/1. The work of R.J.S. was supported by the STFC Consolidated Grant ST/P000363/1.

\section{Leibniz-Loday Algebroids and AKSZ Sigma-Models}
\label{sec:AKSZ}
In this section we will review some well-known material concerning algebroids, graded geometry and the AKSZ construction. The main intent is to develop a fairly self-contained bottoms-up approach to the notion of a Courant algebroid as well as its features and applications in some detail, because later on we will be interested in suitable weakenings of this notion, and we will attempt analogous constructions in those instances. We omit several noteworthy properties and examples of Courant algebroids in this section for brevity, as they will follow as special cases of our more general considerations in Sections~\ref{sec:metricalg} and~\ref{sec:parahermgeom}.

\medskip

\subsection{Algebroids and Leibniz-Loday Algebroids} ~\\[5pt]
\label{subsec:algebroids}
In this paper we use a very broad notion of an `algebroid' which is adapted to all applications that we shall consider. 

\begin{definition}\label{def:algebroid}
An \emph{algebroid} over a manifold $\cQ$ is a vector bundle $E\to \cQ$ equipped with an $\IR$-bilinear bracket $[\,\cdot\,,\,\cdot\,]_E:\mathsf{\Gamma}(E)\times\mathsf{\Gamma}(E)\to \mathsf{\Gamma}(E)$ on its sections, and a bundle morphism $\rho:E\to T\cQ$ covering the identity such that the anchored derivation property
\begin{align}\label{eq:anchorLeibniz}
[e,f\,e']_E = f\,[e,e']_E + \big( \rho(e)\cdot f\big)\,e'
\end{align}
holds for all $e,e'\in\mathsf{\Gamma}(E)$ and $f\in C^\infty(\cQ)$. The map $\rho$ to the tangent bundle of $\cQ$ is called the \emph{anchor map}.

An \emph{algebroid morphism} from an algebroid $(E,[\,\cdot\,,\,\cdot\,]_E,\rho)$ to an algebroid $(E',[\,\cdot\,,\,\cdot\,]_{E'},\rho')$ over the same manifold is a bundle morphism $\psi:E\to E'$ covering the identity such that $\rho'\circ\psi=\rho$ and $\psi\circ[\,\cdot\,,\,\cdot\,]_E = [\,\cdot\,,\,\cdot\,]_{E'} \circ(\psi\times\psi)$.
\end{definition}

Note that here the bracket $[\,\cdot\,,\,\cdot\,]_E$ need not be skew-symmetric and it need not obey the Jacobi identity. Moreover, at this primitive level the only role of the anchor map $\rho$ is to implement the anchored derivation property \eqref{eq:anchorLeibniz} whose meaning is that, for each section $e$ of $E$, $[e,\,\cdot\,]_E$ is a first-order differential operator on $\mathsf{\Gamma}(E)$ whose symbol is the vector field $\rho(e)$ on $\cQ$. Indeed, when $\cQ$ is a point, then an algebroid is simply a vector space with a binary operation. 

This level of generality is needed for our considerations of doubled geometry later on. As we shall discuss throughout this paper, they have natural descriptions via the language of graded geometry in terms of vector fields and local coordinates. For the different flavours of AKSZ constructions as we use them in this paper, we will need a further algebraic condition on the bracket operation in Definition~\ref{def:algebroid}.

\begin{definition}\label{def:Leibnizalgebroid}
A \emph{Leibniz-Loday algebroid} over a manifold $\cQ$ is an algebroid $(E,[\,\cdot\,,\,\cdot\,]_E,\rho)$ whose bracket satisfies the Leibniz identity
\begin{align}\label{eq:Leibnizid}
[e,[e_1,e_2]_E]_E = [[e,e_1]_E,e_2]_E + [e_1,[e,e_2]_E]_E \ ,
\end{align}
for all $e,e_1,e_2\in\mathsf{\Gamma}(E)$.
\end{definition}

When $\cQ$ is a point, then a Leibniz-Loday algebroid is a vector space endowed with the structure of a Leibniz-Loday algebra~\cite{Loday1993}. Generally, the anchored derivation property \eqref{eq:anchorLeibniz} and the Leibniz identity \eqref{eq:Leibnizid} together imply that the anchor map $\rho:E\to T\cQ$ of a Leibniz-Loday algebroid becomes a homomorphism of Leibniz-Loday algebras:
\begin{align}\label{eq:anchorbracket}
\rho([e_1,e_2]_E) = [\rho(e_1),\rho(e_2)]_{T\cQ} \ ,
\end{align}
where $[\,\cdot\,,\,\cdot\,]_{T\cQ}$ is the usual Lie bracket of vector fields on $T\cQ$. 

\medskip

\subsection{AKSZ Constructions} ~\\[5pt]
\label{subsec:AKSZ}
Fix an integer $n\geq0$. We use a somewhat simplified form of the AKSZ construction~\cite{Alexandrov:1995kv} as a geometric tool for building BV action functionals~\cite{Batalin:1981jr} for topological sigma-models of maps from an oriented compact $n{+}1$-dimensional manifold $\Sigma_{n+1}$ (the `source') to a symplectic Lie $n$-algebroid $E$ over a manifold $\cQ$ (the `target'). The AKSZ sigma-models that arise in this way are Chern-Simons theories. In this sense they uniquely encode (up to isomorphism) the algebroid $E$ and provide a means for quantization of $E$. A review is found in~\cite{Ikeda:2012pv}.

We recall some definitions. A symplectic Lie $n$-algebroid is most concisely and naturally described in the language of graded geometry as a differential graded symplectic manifold~\cite{Severa2005}. Recall that a \emph{graded manifold} $\cM=(\cQ,\cA)$ is a ringed space together with the structure sheaf $\cA$ of a graded commutative algebra over an ordinary manifold $\cQ$. It can be modelled locally using even and odd coordinates in fixed degrees, and treated concretely in the language of formal differential geometry by identifying smooth functions on $\cM$ with formal power series in globally defined coordinates $w^\alpha$. We write $\cA^k$ for the subsheaf of $\cA$ consisting of functions of degree~$k$.

\begin{definition}\label{def:dgmanifold}
A \emph{differential graded manifold} (\emph{dg-manifold} for short) is a $\IZ$-graded manifold $\cM=(\cQ,\cA)$ equipped with a degree~$1$ vector field $Q$ which is integrable, that is, $[Q,Q]=2\,Q^2=0$. The vector field $Q$ is a \emph{homological vector field}. 
\end{definition}

When $\cM$ is $\IN$-graded, so that $\cA^0= C^\infty(\cQ)$, the $\IN$-grading can be conveniently described by means of the Euler vector field $\boldsymbol{\varepsilon}$: in coordinates $w^\alpha$ with degrees $|w^\alpha|\geq0$, $\boldsymbol{\varepsilon} = \sum_\alpha\,|w^\alpha| \, w^\alpha\,\frac\partial{\partial w^\alpha}$. A tensor field $T$ on $\cM$ is said to be homogeneous of degree $n$ if $\pounds_{\boldsymbol{\varepsilon}}T=n\,T$, where $\pounds_{\boldsymbol{\varepsilon}}$ denotes the Lie derivative along $\boldsymbol{\varepsilon}$.
The following construction, due to Roytenberg~\cite{Roytenberg2002}, will be used extensively in this section, as well as in Section~\ref{sec:metricalg}.

\begin{theorem}\label{thm:sympln}
Let $\cM=(\cQ, \cA)$ be an $\IN$-graded manifold equipped with a symplectic structure $\omega$ of degree $|\omega|=n>0$, and associated graded Poisson bracket $\{\,\cdot\,,\,\cdot\,\}$ of degree $-n$.  Then there is a one-to-one correspondence between integrable functions on $\cM$ of degree~$n+1$ and homological symplectic vector fields.  
\end{theorem}

\begin{proof}
Since
\be\nonumber
\pounds_{\boldsymbol\varepsilon} \omega= n\, \omega \ ,
\ee
the Cartan homotopy formula $\pounds_{\boldsymbol\varepsilon}=\iota_{\boldsymbol\varepsilon}\circ\de+\de\circ\iota_{\boldsymbol\varepsilon}$ implies 
$$
\omega= \tfrac{1}{n}\, \de \iota_{\boldsymbol\varepsilon} \omega \ ,
$$
where $\iota_{\boldsymbol\varepsilon}$ denotes contraction with the Euler vector field $\boldsymbol\varepsilon$.

Let $Q \in \frX_1(\cM)$ be a symplectic vector field of degree~$1$. Then $[\boldsymbol\varepsilon, Q]=Q$ and $\pounds_Q\omega=\de \iota_Q \omega=0,$
which gives
\be \nonumber
\iota_Q \omega=\iota_{[\boldsymbol\varepsilon, Q]} \omega= [\pounds_{\boldsymbol\varepsilon} , \iota_Q]_\circ (\omega) =\pounds_{\boldsymbol\varepsilon} \iota_Q \omega- \iota_Q (n \, \omega) \ ,
\ee
where we used the Cartan structure equations in the second equality. This yields
\be \nonumber
\iota_Q \omega= \tfrac{1}{n+1}\, \de \iota_{\boldsymbol\varepsilon} \iota_Q \omega \ .
\ee
Hence $Q$ is a Hamiltonian vector field with Hamiltonian $\gamma = \frac{1}{n+1}\, \iota_{\boldsymbol\varepsilon} \iota_Q \omega \in \cA^{n+1}.$
The graded Jacobi identity for the Poisson bracket implies
\be \nonumber
\{\gamma, \{\gamma, f\}\}= \{\{\gamma, \gamma\}, f\}- \{\gamma, \{\gamma, f\}\}
\ee
for all $f \in C^\infty(\cQ).$ Thus 
\be \nonumber
Q^2=X_{\frac12\, \{\gamma, \gamma\}} \ ,
\ee
where \smash{$X_{\frac12\, \{\gamma, \gamma\}}$} is the Hamiltonian vector field corresponding to $\frac12\, \{\gamma, \gamma\} \in \cA^{n+2}.$
Hence $Q^2 =0 $ if and only if $\{\gamma, \gamma\}=0$ because $\{\,\cdot\,,\, \cdot\,\}$ is non-degenerate.

Conversely, given any integrable function $\gamma \in \cA^{n+1},$ we use a derived bracket to set $Q \coloneqq \{ \gamma, \,\cdot\, \}.$ These two constructions are inverse to each other.
\end{proof}

\begin{remark}\label{rem:sympln}
The proof of Theorem~\ref{thm:sympln} shows that every symplectic vector field of degree~$1$ on an $\IN$-graded symplectic manifold $(\cM, \omega)$ with $|\omega|=n>0$ is given by a Hamiltonian of degree $n+1.$ 
\end{remark}

For our AKSZ constructions we assume that $\cM$ is $n$-graded, that is, its coordinates are concentrated in degrees $0,1,\dots,n$. In this case we also call $\cM$ a degree~$n$ manifold.

\begin{definition}\label{def:symplnalg}
A \emph{symplectic Lie $n$-algebroid} is a degree~$n$ dg-manifold $(\cM,Q)$ with a symplectic structure $\omega$ of degree~$n$ for which $Q$ is a symplectic vector field.
\end{definition}

\begin{remark}
Generally, dg-manifolds are sometimes also refered to as `$Q$-manifolds'. Other terminology for symplectic Lie $n$-algebroids appearing in the literature are `symplectic N$Q$-manifolds of degree $n$' or `QP$n$-manifolds'.
\end{remark}

Symplectic Lie $n$-algebroids $(\cM,Q,\omega)$ arise from $n$-graded vector bundles over the degree~$0$ body $\cQ$ of $\cM$, and are generally characterized by the following result, due to Kotov and Strobl~\cite{Kotov:2010wr}. 

\begin{theorem}\label{thm:QPLeibniz}
Let $(\cM,Q,\omega)$ be a symplectic Lie $n$-algebroid of degree $n>1$. Then functions of degree $n-1$ on $\cM=(\cQ,\cA)$ can be identified with sections of a vector bundle $E\to \cQ$ equipped with the structure of a Leibniz-Loday algebroid.
\end{theorem}

\begin{proof}
Let $f\in\cA^0$ be a function of degree~$0$ on $\cQ$, and let $e,e'\in\cA^{n-1}$ be functions of degree~$n-1$ on $\cM$, identified as functions on $E[n-1]$ for a vector bundle $E\to\cQ$. Define a bracket and anchor map on sections $\mathsf{\Gamma}(E)$ by the derived brackets
\begin{align*}
[e,e']_E := -\{\{e,\gamma\},e'\} \qquad \mbox{and} \qquad \rho(e)\cdot f := (-1)^n \, \{\{e,\gamma\},f\} \ .
\end{align*}
Then the anchored derivation property \eqref{eq:anchorLeibniz} follows from the derivation rule for the Poisson bracket $\{\,\cdot\,,\,\cdot\,\}$ induced by the symplectic structure $\omega$, while the Maurer-Cartan equation $\{\gamma,\gamma\}=0$ implies the Leibniz identity \eqref{eq:Leibnizid} for the bracket $[\,\cdot\,,\,\cdot\,]_E$, and hence that the anchor $\rho$ is a bracket homomorphism.
\end{proof}

With this data, the AKSZ construction proceeds as follows. Let $T[1]\Sigma_{n+1}$ be the tangent bundle of the source manifold $\Sigma_{n+1}$ with degree of its fibres shifted by~$1$, which is isomorphic to the exterior algebra of differential forms on $\Sigma_{n+1}$; under this identification it has a canonical homological vector field induced by the de~Rham differential $\de$. Let $\boldsymbol\cM={\sf Map}(T[1]\Sigma_{n+1},\cM)$ be the mapping space of degree~$0$ smooth maps $\hat X:T[1]\Sigma_{n+1}\to\cM$ which intertwine the homological vector fields, that is, $\hat X_*(\de) = Q$. Given an $n$-form $\alpha\in\Omega^n(\cM)$, we can lift it to an $n$-form $\boldsymbol\alpha\in\Omega^n(\boldsymbol\cM)$ by trangression to the mapping space as
\begin{align*}
\boldsymbol\alpha = \int_{T[1]\Sigma_{n+1}} \, \mu \ {\rm ev}^*(\alpha) \ ,
\end{align*}
where $\mu$ is the natural volume measure on $T[1]\Sigma_{n+1}$ and ${\rm ev}:T[1]\Sigma_{n+1}\times\boldsymbol\cM\to\cM$ is the evaluation map.
Choose a local $1$-form $\vartheta$ on $\cM$ such that $\omega=\de\vartheta$. In this paper we will only explicitly write the `classical' or `bosonic' part of the action functional underlying the AKSZ sigma-model. It is constructed by: (i) transgressing the form $-\iota_{\de}\vartheta+\gamma$ to the mapping space $\boldsymbol\cM$; (ii)~integrating over the odd coordinates of $T[1]\Sigma_{n+1}$; and (iii)~restricting to degree~$0$ fields. The $1$-form $\vartheta$ defines the `kinetic term' and the Hamiltonian $\gamma$ defines the `interaction term' of the AKSZ field theory.

\begin{remark}\label{rem:BVAKSZ}
The AKSZ construction is a geometric realization of the BV formalism~\cite{Batalin:1981jr} for topological field theories with generalized gauge symmetries. The full BV master action functional is obtained by allowing the fields $\hat X:T[1]\Sigma_{n+1}\to\cM$ to be maps of arbitrary $\IZ$-degree, which yields all auxiliary fields and anti-fields of the BV formalism. Then the graded Poisson bracket $\{\,\cdot\,,\,\cdot\,\}$ associated to the symplectic structure $\omega$ implements the BV antibracket, while the Maurer-Cartan equation $\{\gamma,\gamma\}=0$ implements the classical master equation which guarantees gauge invariance of the BV action functional, as well as closure of the gauge algebra.
\end{remark}

\begin{remark}\label{rem:AKSZfurther}
Beyond Theorem~\ref{thm:QPLeibniz}, the further algebraic conditions and structures on symplectic Lie $n$-algebroids, as classical geometric objects, are not known generally and must be unravelled on a case by case basis. We consider below the first three degrees $n=0,1,2$ in some detail. For $n=3$ the algebroids were characterized in~\cite{Ikeda:2010vz,Grutzmann2011} and their AKSZ sigma-models applied to ${\sf SL}(5,\IR)$ exceptional field theory in~\cite{Kokenyesi:2018ynq,Chatzistavrakidis:2019seu}. AKSZ constructions for higher-dimensional exceptional field theory are considered in~\cite{Arvanitakis:2018cyo}.
\end{remark}

\medskip

\subsection{Topological Quantum Mechanics} ~\\[5pt]
\label{subsec:topQM}
The simplest instance of the AKSZ construction is when the target is a symplectic dg-manifold of degree~$0$. In this case $Q=0$, and thus a symplectic Lie $0$-algebroid is just a symplectic manifold $(\cQ,\omega)$~\cite{Severa2005,Roytenberg2002}; the degree~$1$ Hamiltonian $\gamma$ is then locally constant on $\cQ$. The corresponding AKSZ sigma-model is the one-dimensional Chern-Simons theory whose Chern-Simons form is a local symplectic potential $\vartheta$ for the symplectic structure: $\omega=\de\vartheta$. In local Darboux coordinates, where $\omega=\de p_i\wedge\de q^i$, we take $\vartheta=p_i\,\de q^i$; here and in the following we use the Einstein summation convention over repeated upper and lower indices.

For a cotangent bundle $\cQ=T^*W$, with $\omega$ the canonical symplectic structure and $\vartheta$ the Liouville $1$-form, if $\Sigma_1$ is an oriented compact $1$-manifold, then the AKSZ construction produces a one-dimensional topological sigma-model
of smooth maps $X:\Sigma_1\to\cQ$ with action functional
\begin{align*}
S(X) = \int_{\Sigma_1} \, X^*\vartheta \ .
\end{align*}
BV quantization of this action functional defines a topological quantum mechanics, which quantizes the symplectic manifold $(\cQ,\omega)$. When $\Sigma_1$ is an interval this computes the $\widehat{A}$-genus of $W$~\cite{Grady:2011jc,Grady:2015ica}.

\medskip

\subsection{Lie Algebroids and Poisson Sigma-Models} ~\\[5pt]
\label{subsec:Poisson}
Let us turn to the lowest non-trivial rung $n=1$ on the AKSZ ladder, firstly by adding a further algebraic condition on the bracket operation in Definition~\ref{def:Leibnizalgebroid}~\cite{Mackenzie}.

\begin{definition}\label{def:Liealgebroid}
A \emph{Lie algebroid} $(E,[\,\cdot\,,\,\cdot\,]_E,\rho)$ over a manifold $\cQ$ is a Leibniz-Loday algebroid whose bracket is skew-symmetric:
\begin{align*}
[e_1,e_2]_E = -[e_2,e_1]_E \ ,
\end{align*}
for all $e_1,e_2\in\mathsf{\Gamma}(E)$.
\end{definition}

It follows that the bracket operation of a Lie algebroid defines a Lie bracket on the sections of the vector bundle $E\to\cQ$. When $\cQ$ is a point, a Lie algebroid is simply a Lie algebra. At the opposite extreme, the tangent bundle $T\cQ$ of any manifold $\cQ$ is always a Lie algebroid with the Lie bracket of vector fields $[\,\cdot\,,\,\cdot\,]_{T\cQ}$ and the identity anchor map $\unit_{T\cQ}$.

Lie algebroids are canonically associated to dg-manifolds of degree~$1$, by the following construction due originally to Vaintrob~\cite{Vaintrob1997}.

\begin{proposition}\label{prop:LiealgebroidQ}
There is a one-to-one correspondence between Lie algebroids and dg-manifolds of degree~$1$.
\end{proposition}

\begin{proof} 
Given a Lie algebroid $(E, [ \, \cdot \, , \, \cdot \, ]_E, \rho)$ over a manifold $\cQ$, its corresponding degree 1 manifold is $\cM=E[1].$ Since $C^\infty(E[1]) = \mathsf{\Gamma}(\midwedge^\bullet E^*),$ the homological vector field $Q=\de_E$ of degree 1 is the Lie algebroid differential defined by
\begin{align*}
\de_E \, \varepsilon (e_1, \dots , e_{k+1}) & \coloneqq  \sum_{i=1}^{k+1}\, \rho(e_i) \cdot \varepsilon(e_1 , \dots , \widehat{e_i}, \dots, e_{k+1}) \\
& \quad \ + \sum_{i < j}\, (-1)^{i+j} \, \varepsilon ([e_i,e_j]_E,e_1, \dots,\widehat{e_i},\dots,\widehat{e_j} ,  \dots,e_{k+1}) \ , 
\end{align*}
for all $\varepsilon \in \mathsf{\Gamma}(\midwedge^k E^*)$ and $e_1,\dots,e_{k+1}\in\mathsf{\Gamma}(E)$, where the hat denotes omission of the corresponding entry. This is a derivation of degree 1 that squares to zero.

Conversely, given a degree 1 manifold $\cM=(\cQ,\cA)$ endowed with a homological vector field $Q,$ its corresponding Lie algebroid is constructed as follows. Since the categories of vector bundles and degree 1 manifolds are equivalent, $\cM \simeq E[1]$ for some vector bundle $E \rightarrow \cQ.$  Then the Lie algebroid structure on $E$ is given by the $C^\infty(\cQ)$-linear anchor map
\be\nonumber
\rho(e) \cdot f \coloneqq \langle Q \cdot f , e \rangle \ ,
\ee
for all $f \in C^\infty(\cQ)$ and $e \in \mathsf{\Gamma}(E),$ where $\langle\,\cdot\,,\,\cdot\,\rangle$ is the canonical dual pairing between $\cA^1\simeq\mathsf{\Gamma}(E^*)$ and $\mathsf{\Gamma}(E)$. The Lie bracket is given by
\be\nonumber
\langle [e_1, e_2]_E, \varepsilon \rangle \coloneqq \rho(e_1) \cdot \langle e_2, \varepsilon \rangle - \rho(e_2) \cdot \langle e_1 , \varepsilon \rangle - \langle Q \cdot \varepsilon , e_1 \wedge e_2 \rangle \ ,
\ee
for all $e_1, e_2 \in \mathsf{\Gamma}(E)$ and $\varepsilon \in \mathsf{\Gamma}(E^*).$ This bracket is skew-symmetric and satisfies the anchored derivation property \eqref{eq:anchorLeibniz}. The Leibniz identity is equivalent to the condition that $Q$ is homological: $Q^2=0.$
\end{proof}

\begin{remark}\label{rem:ChevalleyEilenberg}
When endowed with the action of the homological vector field $Q$, the space of smooth functions $C^\infty(E[1])$ becomes a cochain complex which computes the cohomology of the Lie algebroid $(E,[\,\cdot\,,\,\cdot\,]_E,\rho)$, that is, its Chevalley-Eilenberg algebra $\big(\mathsf{\Gamma}(\midwedge^\bullet E^*),\de_E\big)$. 
\end{remark}

We can now extend Proposition~\ref{prop:LiealgebroidQ} to the case of symplectic dg-manifolds of degree~$1$ to infer that a symplectic Lie algebroid is the same thing as a Poisson manifold~\cite{Roytenberg2002}.

\begin{theorem}\label{thm:QP1Poisson}
There is a one-to-one correspondence between symplectic Lie algebroids and Poisson manifolds.
\end{theorem}

\begin{proof}
Any symplectic $1$-graded manifold $(\cM,\omega)$ is canonically isomorphic to the shifted cotangent bundle $\cM=T^*[1]\cQ$ over some manifold $\cQ$, with the canonical symplectic $2$-form
$\omega = \de\xi_i\wedge\de x^i $,
where $x^i$ are local coordinates on $\cQ$ and $\xi_i$ are odd coordinates on $T^*[1]\cQ$ corresponding to the holonomic vector fields $\frac\partial{\partial x^i}$ on $\cQ$.
The corresponding graded Poisson bracket $\{\,\cdot\,,\,\cdot\,\}$ can be identified with the Schouten-Nijenhuis bracket $[\,\cdot\,,\,\cdot\,]_{\textrm{\tiny S}}$ of multivector fields on $\cQ$. The most general degree~$2$ function $\gamma$ on $T^*[1]\cQ$ is of the form
\begin{align}\label{eq:gammaPoisson}
\gamma = \tfrac12\, \pi^{ij}(x) \, \xi_i \, \xi_j \ ,
\end{align}
where $\pi=\pi^{ij} \, \frac\partial{\partial x^i}\otimes\frac\partial{\partial x^j}$ is any $(0,2)$-tensor on $\cQ$. Then invariance of $\omega$ under the associated vector field $Q=\{\gamma,\,\cdot\,\}$ on $T^*[1]\cQ$ 
 implies $\pi\in\mathsf{\Gamma}(\midwedge^2T\cQ)$, and the Maurer-Cartan equation $\{\gamma,\gamma\}=0$ implies that $\pi$ is a Poisson bivector on $\cQ$, that is, $[\pi,\pi]_{\textrm{\tiny S}}=0$.

Conversely, if $(\cQ,\pi)$ is a Poisson manifold, then the Poisson bracket $\{f,g\}_\pi=\pi(\de f,\de g)$ defines a Lie algebra structure on the space of smooth functions $C^\infty(\cQ)$ and is reproduced as a derived bracket with the Hamiltonian \eqref{eq:gammaPoisson} through
\begin{align*}
\{f,g\}_\pi = -\{\{f,\gamma\},g\} \ .
\end{align*}
The homological vector field $Q$ is the dg-structure of Proposition~\ref{prop:LiealgebroidQ} corresponding to the cotangent Lie algebroid $(T^*\cQ,[\,\cdot\,,\,\cdot\,]_{T^*\cQ},\pi^\sharp)$: its anchor map is the natural bundle morphism $\pi^\sharp:T^*\cQ\to T\cQ$ induced by the bivector $\pi$, $\pi^\sharp\alpha:=\iota_\alpha\pi$, and the Lie bracket is the Koszul bracket on $1$-forms which is defined  by
\begin{align*}
[\alpha,\beta]_{T^*\cQ} = \pounds_{\pi^\sharp\alpha}\beta - \pounds_{\pi^\sharp\beta}\alpha - \de\iota_\alpha\iota_\beta\pi
\end{align*}
for $\alpha,\beta\in\Omega^1(\cQ)$. The differential $Q=\{\gamma,\,\cdot\,\}$ sends a function $f\in C^\infty(\cQ)$ to $\pi^\sharp\de f$.
\end{proof}

We now choose the Liouville $1$-form $\vartheta=\xi_i\,\de x^i$ on $T^*[1]\cQ$ and apply the AKSZ construction with the Hamiltonian \eqref{eq:gammaPoisson}. The AKSZ action functional is defined on the space of vector bundle morphisms $\hat X:T\Sigma_2\to T^*\cQ$ from the tangent bundle of an oriented compact $2$-manifold $\Sigma_2$; such a map is given by its base map $X:\Sigma_2\to\cQ$ and a section $A\in\mathsf{\Gamma}(T^*\Sigma_2\otimes X^*T^*\cQ)$. Then the action functional reads
\begin{align*}
S(X,A) = \int_{\Sigma_2} \, \langle A,\de X\rangle + \frac12\,\langle A,(\pi^\sharp\circ X)A\rangle \ ,
\end{align*}
where we view the fields as $1$-forms $A\in\Omega^1(\Sigma_2,X^*T^*\cQ)$ and $\de X\in\Omega^1(\Sigma_2,X^*T\cQ)$, and $\langle\,\cdot\,,\,\cdot\,\rangle$ denotes the natural pairing defined by pairing the dual values in the pullback bundles $X^*T^*\cQ$ and $X^*T\cQ$ together with the exterior product of differential forms. This is the action functional of the \emph{Poisson sigma-model}~\cite{Ikeda:1993fh,Schaller:1994es}, which is the most general two-dimensional topological field theory that can be obtained through the AKSZ construction~\cite{Cattaneo:2001ys}. When $\Sigma_2$ is a disk, the BV quantization of this sigma-model gives a string theory derivation of Kontesevich's deformation quantization of the Poisson manifold $(\cQ,\pi)$~\cite{Kontsevich:1997vb,Cattaneo:1999fm}.

\medskip

\subsection{Courant Algebroids and Courant Sigma-Models} ~\\[5pt]
\label{subsec:Courant}
We move to the next rung at $n=2$. This time, instead of further constraining the bracket operation in Definition~\ref{def:Leibnizalgebroid}, we complement it with further algebraic structure generalizing the notion of a quadratic Lie alegbra~\cite{Courant1990,Weinstein1997,Uchino2002}.

\begin{definition}\label{def:CourantD}
A \emph{Courant algebroid} on a manifold $\cQ$ is a Leibniz-Loday algebroid $(E,[\,\cdot\,,\,\cdot\,]_{\dorf},\rho)$ with a symmetric non-degenerate bilinear form $\langle\,\cdot\,,\,\cdot\,\rangle_E$ on the fibres of $E$ which is preserved by the bracket operation:
\begin{align}
\rho(e)\cdot\langle e_1,e_2\rangle_E &= \langle[e,e_1]_{\dorf},e_2\rangle_E + \langle e_1,[e,e_2]_{\dorf}\rangle_E \ , \label{eq:CourantD1} \\[4pt]
\langle[e,e]_{\dorf},e_1\rangle_E &= \tfrac12\,\rho(e_1)\cdot\langle e,e\rangle_E \ , \label{eq:CourantD2}
\end{align} 
for all $e,e_1,e_2\in\mathsf{\Gamma}(E)$. The bracket of sections $[\,\cdot\,,\,\cdot\,]_{\dorf}$ is called a \emph{Dorfman bracket}.
\end{definition}

\begin{remark}\label{rem:Courantcond}
The anchor map $\rho$ and the pairing $\langle\,\cdot\,,\,\cdot\,\rangle_E$ from Definition~\ref{def:CourantD} induce a map $\rho^*:T^*\cQ\to E$ given by
\begin{align*}
\langle\rho^*(\alpha),e\rangle_E = \langle\rho^{\rm t}(\alpha),e\rangle \ ,
\end{align*}
for all $\alpha\in\Omega^1(\cQ)$ and for all $e\in\mathsf{\Gamma}(E)$, where $\rho^{\rm t}:T^*\cQ\to E^*$ is the transpose of $\rho$; as before, the bilinear form $\langle\,\cdot\,,\,\cdot\,\rangle$ (without subscript) is the canonical dual pairing between the bundle $E$ and its dual $E^*$. The map $\rho^*$ induces a map $\cD:C^\infty(\cQ)\to\mathsf{\Gamma}(E)$ defined by $\cD f=\rho^*\de f$, for all $f\in C^\infty(\cQ)$, which obeys a derivation-like rule and is the natural generalization of the exterior derivative in the algebroid $E$. This allows us to recast the condition \eqref{eq:CourantD2} as
\begin{align}\label{eq:CourantD2d}
[e,e]_{\dorf} = \tfrac12\,\cD\langle e,e\rangle_E \ ,
\end{align}
for all $e\in\mathsf{\Gamma}(E)$. In particular, the symmetric part of the Dorfman bracket $[\,\cdot\,,\,\cdot\,]_{\dorf}$ is given by
\begin{align}\label{eq:Dorfsym}
[e_1,e_2]_{\dorf} + [e_2,e_1]_{\dorf} = \cD\langle e_1,e_2\rangle_E \ ,
\end{align}
for all $e_1,e_2\in\mathsf{\Gamma}(E)$.
\end{remark}

\begin{remark}\label{rem:Courantminaxioms}
The anchored derivation property \eqref{eq:anchorLeibniz} can be removed from the list of defining properties of a Courant algebroid, as it now follows from the condition \eqref{eq:CourantD1}. Then the Leibniz identity \eqref{eq:Leibnizid} together with \eqref{eq:CourantD1} and \eqref{eq:CourantD2} (or \eqref{eq:CourantD2d}) are a minimal set of three axioms needed to specify a Courant algebroid~\cite{Uchino2002}. For later reference, we also note that \eqref{eq:CourantD2d} together with the bracket homomorphism property \eqref{eq:anchorbracket} imply
\begin{align}\label{eq:rho0cD}
\langle \cD f,\cD g\rangle_E = 0
\end{align}
for all functions $f,g\in C^\infty(\cQ)$, or equivalently that the anchor map $\rho$ vanishes identically on the image of the generalized exterior derivative $\cD$.
\end{remark}

\begin{example}\label{ex:standardCourant}
The most common example is the {standard Courant algebroid}, which features prominently in generalized geometry. It is an extension of the tangent Lie algebroid $(T\cQ,[\,\cdot\,,\,\cdot\,]_{T\cQ},\unit_{T\cQ})$ by cotangent vectors and is based on the {generalized tangent bundle}
\begin{align*}
E = \IT \cQ = T\cQ\oplus T^*\cQ
\end{align*}
over a manifold $\cQ$, with the three natural operations
\begin{align*}
\langle X+\alpha,Y+\beta\rangle_{\IT\cQ} &= \iota_X\beta + \iota_Y\alpha \ , \\[4pt]
\rho(X+\alpha) &= X \ , \\[4pt]
[X+\alpha,Y+\beta]_\dorf &= [X,Y]_{T\cQ} + \pounds_X\beta-\iota_Y\,\de\alpha \ ,
\end{align*}
where the sections of $E=\IT\cQ$ are comprised of vector fields $X,Y\in\mathsf{\Gamma}(T\cQ)$ and $1$-forms $\alpha,\beta\in\mathsf{\Gamma}(T^*\cQ)$. In this example, $\cD=\de$.
\end{example}

It is a celebrated result, due to Roytenberg~\cite{Roytenberg,Roytenberg2002}, and independently \v{S}evera~\cite{Severa-letters}, that a symplectic Lie $2$-algebroid is the same thing as a Courant algebroid.

\begin{theorem}\label{thm:QP2Courant}
There is a one-to-one correspondence between symplectic Lie $2$-algebroids and Courant algebroids.
\end{theorem}

\begin{proof}
Let $(\cM,Q,\omega)$ be a symplectic dg-manifold of degree~$2$. Choose local Darboux coordinates $(x^i,\zeta^a,\xi_i)$ with degrees $(0,1,2)$ in which the graded symplectic structure is given by 
\begin{align}\label{eq:omegaCourant}
\omega = \de \xi_i\wedge\de x^i + \tfrac12\,\eta_{ab}\,\de\zeta^a\wedge\de \zeta^b \ ,
\end{align}
where $\eta_{ab}$ is a constant metric on the degree~$1$ subspace of $\cM$. The most general degree~$3$ function $\gamma$ on $\cM$ has the form
\begin{align}\label{eq:gammaCourant}
\gamma = \rho_a^i(x) \,\xi_i\,\zeta^a - \tfrac1{3!}\,T_{abc}(x) \, \zeta^a\,\zeta^b\,\zeta^c \ ,
\end{align}
where $\rho_a^i$ and $T_{abc}$ are degree~$0$ functions on the body $\cQ$ of $\cM$.  We use Theorem~\ref{thm:QPLeibniz} to construct a Leibniz-Loday algebroid $(E,[\,\cdot\,,\,\cdot\,]_\dorf,\rho)$, with operations defined on degree~$1$ functions $e$ which are identified as local sections of a vector bundle $E$ over $\cQ$. Then the Poisson bracket
\begin{align*}
\langle e_1,e_2\rangle_E = \{e_1,e_2\}
\end{align*}
defines a fibrewise symmetric pairing $\langle\,\cdot\,,\,\cdot\,\rangle_E$ on $E$ with local coordinate expression $\eta_{ab}$. It satisfies the Courant algebroid axioms \eqref{eq:CourantD1} and \eqref{eq:CourantD2} as a consequence of the derivation rule and the Jacobi identity for the Poisson bracket $\{\,\cdot\,,\,\cdot\,\}$ induced by the symplectic structure $\omega$.

Conversely, given a Courant algebroid $(E,[\,\cdot\,,\,\cdot\,]_\dorf,\langle\,\cdot\,,\,\cdot\,\rangle_E,\rho)$ on $\cQ$, we define a symplectic dg-manifold $(\cM,Q,\omega)$ of degree~$2$ by the symplectic submanifold of $T^*[2]E[1]$ corresponding to the isometric embedding $E\hookrightarrow E\oplus E^*$ with respect to the Courant algebroid pairing and the canonical dual pairing. Then $x^i$ are local coordinates on $\cQ$, $\xi_i$ are local fibre coordinates of the shifted cotangent bundle $T^*[2]\cQ$ corresponding to the holonomic vector fields $\frac\partial{\partial x^i}$, and $\zeta^a$ are local fibre coordinates of the shifted vector bundle $E[1]$ corresponding to a choice of basis $e^a$ of sections of the dual bundle $E^*$. With $e_a$ the basis of sections of $E$ dual to $e^a$, the structure functions in \eqref{eq:omegaCourant} and \eqref{eq:gammaCourant} are given by
\begin{align*}
\langle e_a, e_b\rangle_E = \eta_{ab} \ , \quad \rho(e_a) = \rho_a^i\,\frac\partial{\partial x^i} \qquad \mbox{and} \qquad \langle[e_a,e_b]_\dorf,e_c\rangle_E = T_{abc} \ ,
\end{align*}
and the Courant algebroid axioms imply the Maurer-Cartan equation $\{\gamma,\gamma\}=0$. The differential $Q=\{\gamma,\,\cdot\,\}$ sends a function $f\in C^\infty(\cQ)$ to $\rho^{\rm t}\de f$.
\end{proof}

We now choose the Liouville $1$-form 
\begin{align}\label{eq:LiouvilleCourant}
\vartheta=\xi_i\,\de x^i + \tfrac12\,\zeta^a\,\eta_{ab}\,\de\zeta^b
\end{align}
on $\cM$ and apply the AKSZ construction with the Hamiltonian \eqref{eq:gammaCourant}. Let $\Sigma_3$ be an oriented compact $3$-manifold, and choose a grading-preserving connection to fix an isomorphism $\cM\simeq E[1]\oplus T^*[2]\cQ$. The AKSZ action functional is defined on the space of degree~$0$ maps $\hat X:T[1]\Sigma_3\to \cM$, which are given by a smooth map $X:\Sigma_3\to\cQ$, a section $A\in\mathsf{\Gamma}(T^*\Sigma_3\otimes X^*E)$, and a section $F\in\mathsf{\Gamma}(\midwedge^2T^*\Sigma_3\otimes X^*T^*\cQ)$. The action functional then reads
\begin{align}\label{eq:AKSZ3daction}
S(X,A,F) = \int_{\Sigma_3} \, \langle F,\de X\rangle + \frac12\,\langle A,\de A\rangle_E - \langle F,(\rho\circ X)A\rangle + \frac1{3!} \, \langle A ,[A,A]_\dorf\rangle_E \ ,
\end{align}
where we view the fields as differential forms $A\in\Omega^1(\Sigma_3,X^*E)$, $\de X\in\Omega^1(\Sigma_3,X^*T\cQ)$ and $F\in\Omega^2(\Sigma_3,X^*T^*\cQ)$, and the pairings are taken in the pullback bundles together with the exterior products of differential forms. This is the action functional of the \emph{Courant sigma-model}~\cite{Ikeda:2002wh,Hofman:2002rv,Hofman:2002jz,Roytenberg:2006qz}, which is a canonical three-dimensional topological field theory associated to any Courant algebroid; its BV quantization thus gives a quantization of Courant algebroids, though as yet this has not been fully achieved. This AKSZ sigma-model is a vast generalization of three-dimensional Chern-Simons gauge theory: For the special case when $\cQ$ is a point, a Courant algebroid is just a quadratic Lie algebra, and \eqref{eq:AKSZ3daction} is the classical Chern-Simons functional on the $3$-manifold $\Sigma_3$.

\begin{remark}
Gauge invariance of the action functional \eqref{eq:AKSZ3daction} under local BRST transformations is equivalent to the axioms and properties of a Courant algebroid structure on the vector bundle $E\to\cQ$. In particular, when $E=\IT\cQ=T\cQ\oplus T^*\cQ$ is the generalized tangent bundle (with general anchor map $\rho:\IT\cQ\to T\cQ$ and Dorfman bracket), the tensor
\begin{align*}
T(e_1,e_2,e_3) := \langle [e_1,e_2]_\dorf,e_3\rangle_{\IT \cQ} \ ,
\end{align*}
for $e_1,e_2,e_3\in\mathsf{\Gamma}(\IT\cQ)$, encodes the \emph{fluxes} of supergravity and the axioms give their Bianchi identities.
\end{remark}

\medskip

\subsection{Gauge Algebras and Courant Brackets} ~\\[5pt]
\label{subsec:Courantbracket}
The infinitesimal symmetries of symplectic Lie $n$-algebroids are linked to gauge symmetries of the corresponding AKSZ field theories. Here we focus on those which are given by inner derivations, in a suitable sense, which we shall generally call `gauge transformations'. 

For the first two rungs of the AKSZ ladder these symmetries are essentially the same and are easy to describe. For $n=0$, an infinitesimal symmetry of a symplectic manifold $(\cQ,\omega)$ corresponds to a symplectic vector field $X\in\mathsf{\Gamma}(T\cQ)$, that is, $\pounds_X\omega=0$; they close a Lie algebra under the Lie bracket of vector fields by virtue of the Cartan structure equation
\begin{align*}
[\pounds_X,\pounds_Y]_\circ = \pounds_{[X,Y]_{T\cQ}}
\end{align*}
for the Lie derivatives along $X,Y\in\mathsf{\Gamma}(T\cQ)$. The Hamiltonian vector fields $X_f$ by definition satisfy $\iota_{X_f}\omega=\de f$, for functions $f\in C^\infty(\cQ)$, and form a natural subalgebra which is isomorphic to the Lie algebra of smooth functions on $\cQ$ with the corresponding Poisson bracket; these are called (infinitesimal) canonical transformations and we think of them as infinitesimal `gauge symmetries', generated by the action of functions on $\cQ$ through the Poisson bracket. 
For $n=1$, an infinitesimal symmetry of a Poisson manifold $(\cQ,\pi)$ similarly corresponds to a Poisson vector field $X\in\mathsf{\Gamma}(T\cQ)$, that is, $\pounds_X\pi=0$. Again the Hamiltonian vector fields $X_f:=\pi^\sharp\de f$ form a Lie algebra isomorphic to the Poisson algebra: $[X_f,X_g]_{T\cQ} = X_{\{f,g\}_\pi}$ for $f,g\in C^\infty(\cQ)$.

For $n=2$, the situation is more involved. Let $(E,[\,\cdot\,,\,\cdot\,]_\dorf,\langle\,\cdot\,,\,\cdot\,\rangle_E,\rho)$ be a Courant algebroid on a manifold $\cQ$. By definition, the adjoint action of $\mathsf{\Gamma}(E)$ on itself by the Dorfman bracket
\begin{align*}
\pounds_e^\dorf := [e,\,\cdot\,]_\dorf \ ,
\end{align*}
for $e\in\mathsf{\Gamma}(E)$,
is a first order differential operator whose symbol is the vector field $\rho(e)$ on $\cQ$. It acts as an inner derivation of $(E,[\,\cdot\,,\,\cdot\,]_\dorf,\langle\,\cdot\,,\,\cdot\,\rangle_E)$, that is, it is an infinitesimal symmetry of the Courant algebroid. The operator $\pounds_e^\dorf:\mathsf{\Gamma}(E)\to\mathsf{\Gamma}(E)$ is called a \emph{generalized Lie derivative} on the Courant algebroid; the reason for the terminology is best motivated by recalling the standard Courant algebroid of Example~\ref{ex:standardCourant} where it coincides with the generalized Lie derivative of generalized geometry.

The collection of generalized Lie derivatives for all sections of $E$ is a vector space which we will think of as the (infinitesimal) `gauge transformations' of the Courant algebroid. They should close a Lie algebra with respect to the commutator bracket on endomorphisms of $\mathsf{\Gamma}(E)$; we refer to this Lie algebra as the \emph{gauge algebra} of the Courant algebroid. From the Leibniz identity \eqref{eq:Leibnizid} we directly obtain
\begin{align}\label{eq:Dorfmanclosure}
\pounds^\dorf_{e_1}\circ\pounds^\dorf_{e_2} = \pounds^\dorf_{[e_1,e_2]_\dorf} + \pounds^\dorf_{e_2}\circ\pounds^\dorf_{e_1} \ ,
\end{align}
for $e_1,e_2\in\mathsf{\Gamma}(E)$, which shows that the gauge closure can be expressed in terms of the Dorfman bracket on sections of $E$. However, the Dorfman bracket is not skew-symmetric so it is not the natural bracket operation to use for this algebra. To write a manifestly skew-symmetric closure relation, we subtract from \eqref{eq:Dorfmanclosure} the corresponding identity with $e_1$ and $e_2$ interchanged, and after rearrangement we obtain the gauge algebra
\begin{align}\label{eq:Courantclosure}
\big[\pounds_{e_1}^\dorf,\pounds_{e_2}^\dorf\big]_\circ = \pounds^\dorf_{[e_1,e_2]_\cour} \ ,
\end{align}
where
\begin{align*}
[e_1,e_2]_\cour := \tfrac12\,\big([e_1,e_2]_\dorf - [e_2,e_1]_\dorf \big)
\end{align*}
is the skew-symmetrization of the Dorfman bracket. The skew-symmetric bracket on sections $[\,\cdot\,,\,\cdot\,]_\cour$ is called a \emph{Courant bracket}.

In contrast to the Dorfman bracket, the Courant bracket is neither an algebroid bracket nor a Lie bracket, as it violates both the anchored derivation property \eqref{eq:anchorLeibniz} and the Jacobi identity (which is equivalent to the Leibniz identity \eqref{eq:Leibnizid} for a skew-symmetric bracket). Nevertheless, it can be used to completely characterize the Courant algebroid axioms~\cite{Roytenberg}.

\begin{proposition}\label{prop:CourantalgrboidC}
Let $(E,[\,\cdot\,,\,\cdot\,]_\dorf,\langle\,\cdot\,,\,\cdot\,\rangle_E,\rho)$ be a Courant algebroid over a manifold $\cQ$. Then the compatibility conditions on the Dorfman bracket can be equivalently expressed in terms of the Courant bracket as
\begin{align*}
[e_1,f\,e_2]_\cour &= f\,[e_1,e_2]_\cour + \big(\rho(e_1)\cdot f\big) \, e_2 - \tfrac12\,\langle e_1,e_2\rangle_E \, \cD f \ , \\[4pt]
{\sf Jac}_\cour(e_1,e_2,e_3) &= \cD\,{\sf Nij}_\cour(e_1,e_2,e_3) \ , \\[4pt]
\rho(e_1)\cdot\langle e_2,e_3\rangle_E &= \big\langle[e_1,e_2]_{\cour} + \tfrac12\,\cD\langle e_1,e_2\rangle_E ,e_3\big\rangle_E + \big\langle e_2,[e_1,e_3]_{\cour} + \tfrac12\,\cD\langle e_1,e_3\rangle_E\big\rangle_E \ ,
\end{align*}
for all $e_1,e_2,e_3\in\mathsf{\Gamma}(E)$ and $f\in C^\infty(\cQ)$, where 
\begin{align*}
{\sf Jac}_\cour(e_1,e_2,e_3) := [[e_1,e_2]_\cour,e_3]_\cour + [[e_3,e_1]_\cour,e_2]_\cour + [[e_2,e_3]_\cour,e_1]_\cour
\end{align*}
is the {Jacobiator} of the Courant bracket, and
\begin{align*}
{\sf Nij}_\cour(e_1,e_2,e_3) := \tfrac1{3!} \, \big( \langle[e_1,e_2]_\cour,e_3\rangle_E + \langle[e_3,e_1]_\cour,e_2\rangle_E + \langle[e_2,e_3]_\cour,e_1\rangle_E \big)
\end{align*}
is the {Nijenhuis tensor} of the Courant bracket.
\end{proposition}

\begin{proof}
The three conditions on the Courant bracket easily follow from using \eqref{eq:Dorfsym} to express its deviation from the Dorfman bracket as
\begin{align}\label{eq:DorfCourant}
[e_1,e_2]_\dorf = [e_1,e_2]_\cour + \tfrac12\,\cD\langle e_1,e_2\rangle_E \ ,
\end{align}
and substituting this into the anchored derivation property \eqref{eq:anchorLeibniz}, the Leibniz identity \eqref{eq:Leibnizid}, and the metric compatibility condition \eqref{eq:CourantD1}, respectively.
\end{proof}

\medskip

\subsection{Flat $L_\infty$-Algebras} ~\\[5pt]
\label{subsec:flatLinfty}
The formulation of the gauge algebra \eqref{eq:Courantclosure} in terms of the Courant bracket still leaves open one puzzle: the violation of the Jacobi identity by the Courant bracket from Proposition~\ref{prop:CourantalgrboidC} appears to be in contradiction with the vanishing Jacobiator ${\sf Jac}_\circ=0$ of the commutator bracket. This in fact poses no problem as one can explicitly check $[\cD\,{\sf Nij}_\cour(e_1,e_2,e_3),e]_\dorf=0$ for all sections $e,e_1,e_2,e_3\in\mathsf{\Gamma}(E)$. However, the violation of the Jacobi identity itself in Proposition~\ref{prop:CourantalgrboidC}, which is controlled by the generalized exterior derivative of the Nijenhuis tensor of the Courant bracket, suggests a more natural formulation of the gauge algebra and its closure in the language of strong homotopy Lie algebras, or $L_\infty$-algebras, which are homotopy coherent weakenings of the axioms of a Lie algebra. They generally underlie the gauge structure and dynamics of classical perturbative field theories. This formulation is dual to the BV formalism and is naturally tailored to control field theories with open gauge algebras and reducible symmetries, like the Courant sigma-model, as well as violations of the Jacobi identities. 

We begin by recalling the definitions~\cite{Lada:1994mn}.

\begin{definition}\label{def:Linfty}
Let $R$ be a commutative ring. An \emph{$L_\infty$-algebra} over $R$ is a free graded $R$-module $L=\bigoplus_{k\in\IZ}\, L_k$ with a degree~$1$ derivation
\begin{align*}
Q:\midodot^\bullet L^*[1] \longrightarrow \midodot^\bullet L^*[1] 
\end{align*}
which is a differential, $Q^2=0$, making the symmetric algebra $\midodot^\bullet L^*[1]$ into a commutative dg-algebra over $R$, called the \emph{Chevalley-Eilenberg algebra} of the $L_\infty$-algebra~$L$. 

An \emph{$L_\infty$-morphism} from an $L_\infty$-algebra $(L,Q)$ to an $L_\infty$-algebra $(L',Q')$ is an algebra homomorphism $\Psi:\midodot^\bullet L^*[1] \longrightarrow \midodot^\bullet L'^*[1] $ of degree~$0$ which intertwines the derivations: $\Psi\circ Q = Q'\circ\Psi$.
\end{definition}

By virtue of the Leibniz rule, the derivation is determined entirely by its action on $L^*[1]$, hence we may view it as an $R$-linear map $Q:L^*[1]\to \midodot^\bullet L^*[1]$. Let $Q_m:L^*[1]\to\midodot^mL^*[1]$ be the homogeneous components of $Q$, for $m=0,1,2,\dots$. Let $s:L^*\to L^*[1]$ be the suspension map of degree~$1$; this is the tautological isomorphism which identifies $v\in L_k^*$ with $v\in L^*[1]_{k-1}:=L^*_k$. Taking the dual gives a sequence of maps $Q_m^*:(\midodot^mL^*[1])^*\to (L^*[1])^*$ which, after composing with the suspension, we can consider as maps
\begin{align*}
\ell_m := s^{-1}\circ Q_m^* \circ s^{\otimes m} : \midwedge^mL \longrightarrow L
\end{align*}
for $m=0,1,2,\dots$, which are called the \emph{$m$-brackets} of the $L_\infty$-algebra $L$; they are multilinear graded skew-symmetric maps of degree~$2-m$ which satisfy a sequence of higher homotopy Jacobi identities among them, encoded in the condition $Q^2=0$. 

A \emph{flat} $L_\infty$-algebra is an $L_\infty$-algebra with $\ell_0=0$. In this case $\ell_1$ is a differential and $\ell_2$ is a cochain map obeying the Jacobi identity up to exact terms; hence the cohomology of the cochain complex $(L,\ell_1)$ of a flat $L_\infty$-algebra is a graded Lie algebra. A \emph{curved} $L_\infty$-algebra is an $L_\infty$-algebra with $\ell_0\neq0$. Differential graded Lie algebras can be regarded as $L_\infty$-algebras with differential $\ell_1$, bracket $\ell_2$ and $\ell_m=0$ for all~$m>2$.

On general grounds, any dg-manifold $(\cM,Q)$ is naturally described as a (local) $L_\infty$-algebra which completely captures its algebraic structure: in this case $L$ is the graded vector space of polynomial functions on $\cM$~\cite{Voronov2005}. In these applications, we take $R=\IR$ and work in an appropriate category of topological vector spaces with the natural morphisms, tensor products, and so on, though we do not indicate this explicitly in the notation. In particular, $L^*:={\sf Hom}_{\IR}(L,\IR)$ means the continuous dual to $L$, and $\midodot^\bullet L^*$ means the completed symmetric algebra defined using the continuous product and the completed projective tensor product. 

Here we are interested in symplectic Lie $n$-algebroids $(\cM,Q,\omega)$, which have associated flat $n$-term $L_\infty$-algebras, comprising functions on $\cM$ of degrees $0,1,\dots,n-1$, whose $m$-brackets $\ell_m$ can be computed from derived  brackets with the Hamiltonian $\gamma$ of the symplectic dg-structure~\cite{Ritter:2015ffa}. For $n=0$ all brackets are identically zero, giving the trivial $L_\infty$-algebra on a symplectic manifold $(\cQ,\omega)$, while for $n=1$ we obtain only one non-zero bracket $\ell_2=\{\,\cdot\,,\,\cdot\,\}_\pi$ on $L=L_0=C^\infty(\cQ)$, which recovers the Lie algebra of Poisson brackets on a Poisson manifold $(\cQ,\pi)$. For $n=2$, we recover the $L_\infty$-algebra of a Courant algebroid, originally due to Roytenberg and Weinstein~\cite{Roytenberg:1998vn}.

\begin{theorem}\label{thm:LinftyCourant}
Let $(E,[\,\cdot\,,\,\cdot]_\dorf,\langle\,\cdot\,,\,\cdot\,\rangle_E,\rho)$ be a Courant algebroid on a manifold $\cQ$. Then there is a flat $2$-term $L_\infty$-algebra on $L=L_{-1}\oplus L_0$ with
\begin{align*}
L_{-1} = C^\infty(\cQ) \qquad \mbox{and} \qquad L_0 = \mathsf{\Gamma}(E) \ ,
\end{align*}
whose non-zero brackets are given by
\begin{align*}
\ell_1(f) &= \cD f \ , \\[4pt]
\ell_2(e_1,e_2) = [e_1,e_2]_\cour \quad &  \ \, , \, \ \quad
\ell_2(e_1,f) = \tfrac12\,\langle e_1,\cD f \rangle_E \ , \\[4pt]
\ell_3(e_1,e_2,e_3) &= -{\sf Nij}_\cour(e_1,e_2,e_3) \ ,
\end{align*}
for all $f\in C^\infty(\cQ)$ and $e_1,e_2,e_3\in\mathsf{\Gamma}(E)$.
\end{theorem}

\begin{remark}\label{rem:LinftyCourant}
In the context of gauge algebras, the brackets from $L_{-1}$ in Theorem~\ref{thm:LinftyCourant} represent non-trivial ``higher'' gauge symmetries, that is, gauge symmetries among the gauge transformations in $L_0$ themselves; in other words, the gauge symmetries of a Courant algebroid are \emph{reducible}. This happens as well in the BV formulation of the Courant sigma-model, which has an open gauge algebra of reducible symmetries. The corresponding dual $L_\infty$-algebra formalism involves infinitely many brackets. An explicit $L_\infty$-morphism from the Courant algebroid $L_\infty$-algebra of Theorem~\ref{thm:LinftyCourant} to the gauge $L_\infty$-algebra of the Courant sigma-model is constructed by Grewcoe and Jonke in~\cite{Grewcoe:2020ren}; this requires extending the $2$-term cochain complex of Theorem~\ref{thm:LinftyCourant} by the degree~$1$ subspace $L_1=\mathsf{\Gamma}(T\cQ)$ and the anchor map $\ell_1|_{L_0}=\rho$ in order to accomodate the field-dependent gauge algebra of the Courant sigma-model.
\end{remark}

\section{Metric Algebroids}
\label{sec:metricalg}
In this section we will introduce and study a weakening of the notion of Courant algebroid from Section~\ref{subsec:Courant}, which is the natural generalization for the algebroids underlying doubled geometry that we consider later on~\cite{Vaisman2012}.

\medskip

\subsection{Metric Algebroids and Pre-Courant Algebroids} 

\begin{definition} \label{malg}
A \emph{metric algebroid} over a manifold $\cQ$ is an algebroid $(E,\llbracket\,\cdot\,,\,\cdot\,\rrbracket_\dorf,\rho)$ with a fibrewise non-degenerate pairing $\langle\,\cdot\,,\,\cdot\,\rangle_E\in\mathsf{\Gamma}(\midodot^2E^*)$ which is preserved by the bracket operation:
\begin{align}
\rho(e)\cdot\langle e_1, e_2\rangle_E &= \langle\llbracket e,e_1\rrbracket_\dorf, e_2\rangle_E + \langle e_1, \llbracket e, e_2\rrbracket_\dorf\rangle_E \ , \label{eq:metric1} \\[4pt]
\langle\llbracket e, e\rrbracket_\dorf, e_1\rangle_E &= \tfrac{1}{2}\, \rho(e_1)\cdot\langle e, e\rangle_E \ , \label{eq:metric2}
\end{align}
for all $e,e_1, e_2 \in \mathsf{\Gamma}(E)$. The bracket of sections $\llbracket\,\cdot\,,\,\cdot\,\rrbracket_\dorf$ is called a \emph{D-bracket}.

A \emph{metric algebroid morphism} from a metric algebroid $(E,\llbracket\,\cdot\,,\,\cdot\,\rrbracket_\dorf,\langle\,\cdot\,,\,\cdot\,\rangle_E,\rho)$ to a metric algebroid $(E',\llbracket\,\cdot\,,\,\cdot\,\rrbracket'_\dorf,\langle\,\cdot\,,\,\cdot\,\rangle_{E'},\rho')$ over the same manifold is an algebroid morphism $\psi$ which is an isometry, that is, $\langle\,\cdot\,,\,\cdot\,\rangle_{E'}\circ(\psi\times\psi)=\langle\,\cdot\,,\,\cdot\,\rangle_E$.
\end{definition}

A metric algebroid $(E,\llbracket\,\cdot\,,\,\cdot\,\rrbracket_\dorf,\langle\,\cdot\,,\,\cdot\,\rangle_E,\rho)$ is called \emph{regular} if its anchor map $\rho:E\to T\cQ$ has constant rank, and \emph{transitive} if $\rho$ is surjective. A \emph{split metric algebroid} is a metric algebroid whose underlying vector bundle $E\to \cQ$ is the Whitney sum $E=A\oplus A^*$ of a vector bundle $A\to \cQ$ and its dual $A^*\to \cQ$.

\begin{example}\label{ex:Courantmetric}
A Courant algebroid is precisely a metric algebroid which is also a Leibniz-Loday algebroid. 
\end{example} 

\begin{remark}\label{rem:Courantmetric}
As can be anticipated from Example~\ref{ex:Courantmetric}, metric algebroids share some features in common with Courant algebroids. In particular, the anchored derivation property \eqref{eq:anchorLeibniz} again follows from the axiom \eqref{eq:metric1}, and the discussion of Remark~\ref{rem:Courantcond} applies verbatum to a metric algebroid to show that the symmetric part of the D-bracket $\llbracket\,\cdot\,,\,\cdot\,\rrbracket_\dorf$ can be written in terms of the generalized exterior derivative $\cD:C^\infty(\cQ)\to \mathsf{\Gamma}(E)$ and the metric $\langle\,\cdot\,,\,\cdot\,\rangle_E$ analogously to~\eqref{eq:Dorfsym}:
\begin{align}\label{eq:Courantmetric}
\llbracket e_1,e_2\rrbracket_\dorf + \llbracket e_2,e_1\rrbracket_\dorf = \cD\langle e_1,e_2\rangle_E \ ,
\end{align}
for all $e_1,e_2\in\mathsf{\Gamma}(E)$. For later use, we also note that \eqref{eq:Courantmetric} together with the anchored derivation property \eqref{eq:anchorLeibniz} imply the left derivation property
\begin{align}\label{eq:leftderivation}
\llbracket f\,e_1,e_2\rrbracket_\dorf = f\,\llbracket e_1,e_2\rrbracket_\dorf - \big(\rho(e_2)\cdot f\big)\,e_1 + \cD f \, \braket{e_1,e_2}_E \ ,
\end{align}
for all $f\in C^\infty(\cQ)$ and $e_1,e_2\in\mathsf{\Gamma}(E)$.
\end{remark}

Despite the similarities noted in Remark~\ref{rem:Courantmetric}, the failure of the Leibniz identity \eqref{eq:Leibnizid} for a generic metric algebroid means that its anchor map $\rho$ is not a bracket morphism in general. On the other hand, one can impose the homomorphism property \eqref{eq:anchorbracket} independently, and arrive at an important class of (non-Courant) metric algebroids which resemble Courant algebroids in the closest possible way~\cite{Vaisman2005,Hansen:2009zd,Liu:2012qub}.

\begin{definition}\label{def:preCourant}
A metric algebroid $(E,\llbracket\,\cdot\,,\,\cdot\,\rrbracket_\dorf,\langle\,\cdot\,,\,\cdot\,\rangle_E,\rho)$ over a manifold $\cQ$ is a \emph{pre-Courant algebroid} if its anchor map $\rho:E\to T\cQ$ is a bracket morphism:
$$
\rho(\llbracket e_1, e_2\rrbracket_\dorf)= [\rho(e_1), \rho(e_2)]_{T\cQ} \ ,
$$
for all $e_1 , e_2 \in \mathsf{\Gamma}(E)$. 
\end{definition}

\begin{example}\label{ex:exactmetric}
Let $(E,\llbracket\,\cdot\,,\,\cdot\,\rrbracket_\dorf,\langle\,\cdot\,,\,\cdot\,\rangle_E,\rho)$ be a regular pre-Courant algebroid over a manifold~$\cQ$. Let $\rho^*:T^*\cQ\to E$ be the map defined by $\langle\rho^*(\alpha),e\rangle_E = \langle \rho^{\rm t}(\alpha),e\rangle$ for $\alpha\in\Omega^1(\cQ)$ and $e\in\mathsf{\Gamma}(E)$. Then $\Im(\rho^*)$ is a coisotropic subbundle of $E$ and $\mathsf{\Gamma}(\Im(\rho^*))$ is an abelian ideal of $(\mathsf{\Gamma}(E), \llbracket\,\cdot \, , \, \cdot\,\rrbracket_\dorf ).$ It follows that $\rho \circ \rho^* =0,$ and there is a chain complex of bundle morphisms
\be\label{eq:metricexact}
T^*\cQ \xlongrightarrow{\rho^*} E \xlongrightarrow{\rho} T\cQ \ .
\ee
The pre-Courant algebroid is said to be \emph{exact} if \eqref{eq:metricexact} is a short exact sequence.  

An isotropic splitting $s: T\cQ \rightarrow E$ of an exact pre-Courant algebroid defines an isomorphism $E \simeq \IT \cQ=T\cQ \oplus T^*\cQ$ to the generalized tangent bundle of $\cQ$, viewed as a split metric algebroid, as well as a $3$-form $H\in\Omega^3(\cQ)$ by
\begin{align*}
H(X,Y,Z)= \big\langle\llbracket s(X),s(Y)\rrbracket_\dorf, s(Z) \big\rangle_E \ ,
\end{align*}
for $X, Y,Z \in \mathsf{\Gamma}(T\cQ)$. The D-bracket $\llbracket\,\cdot\,,\,\cdot\,\rrbracket_\dorf$ maps to the bracket on the splitting $E\simeq\IT \cQ$ given by
\be\nonumber
 \llbracket X+ \alpha, Y + \beta \rrbracket_{\dorf} = [X,Y]_{T\cQ} + \pounds_X \beta  -\iota_Y \, \de \alpha + \iota_X\iota_Y H  \ , 
\ee
for all $X, Y \in \mathsf{\Gamma}(T\cQ)$ and $\alpha,\beta \in \mathsf{\Gamma}(T^*\cQ)$. This is the Dorfman bracket of the standard Courant algebroid from Example~\ref{ex:standardCourant}, now `twisted' by the $3$-form $H$. However, since $\llbracket\,\cdot \, , \, \cdot\,\rrbracket_\dorf$ violates the Leibniz identity \eqref{eq:Leibnizid}, the $3$-form $H$ is not closed and so does not represent any class in ${\sf H}^3(\cQ,\IR).$ In other words, there is no extension of the \v{S}evera classification of exact Courant algebroids~\cite{Severa-letters,gualtieri:tesi} to exact pre-Courant algebroids. 
\end{example} 

\begin{example}\label{ex:metricalgeeta}
Let $(\cQ,\eta)$ be a pseudo-Riemannian manifold, and let $\nabla^\LC$ denote the Levi-Civita connection of $\eta$. Define a bracket operation $\llbracket\,\cdot\,,\,\cdot\,\rrbracket^\eta_\dorf:\mathsf{\Gamma}(T\cQ)\times\mathsf{\Gamma}(T\cQ)\to\mathsf{\Gamma}(T\cQ)$ by
\begin{align*}
\eta(\llbracket X,Y\rrbracket^\eta_\dorf,Z) = \eta(\nabla^\LC_XY-\nabla_Y^\LC X,Z)  + \eta(\nabla_Z^\LC X,Y)
\end{align*}
for vector fields $X,Y,Z\in\mathsf{\Gamma}(T\cQ)$. Then $(T\cQ,\llbracket\,\cdot\,,\,\cdot\,\rrbracket^\eta_\dorf,\eta,\unit_{T\cQ})$ is a metric algebroid~\cite{Svoboda:2020msh} which is not a pre-Courant algebroid. This D-bracket can be twisted by any $3$-form $H\in\Omega^3(\cQ)$, similarly to Example~\ref{ex:exactmetric}.
\end{example}

\begin{definition}\label{def:Dstructure}
An \emph{almost D-structure} on a metric algebroid $(E, \llbracket\,\cdot \, , \, \cdot\,\rrbracket_\dorf,\langle\,\cdot\,,\,\cdot\,\rangle_E,\rho)$ is an isotropic vector subbundle $L \subset E.$ It is a \emph{D-structure} if $L$ is also involutive with respect to the D-bracket $\llbracket \, \cdot \, , \, \cdot\,\rrbracket_\dorf,$ that is,
$\llbracket\mathsf{\Gamma}(L), \mathsf{\Gamma}(L)\rrbracket_\dorf \subseteq \mathsf{\Gamma}(L).$ 
\end{definition}

\begin{example}\label{ex:splitexactDstructure}
Let $(E,\llbracket\,\cdot\,,\,\cdot\,\rrbracket_\dorf,\langle\,\cdot\,,\,\cdot\,\rangle_E,\rho)$ be a split exact pre-Courant algebroid over a manifold~$\cQ$. Then $T^*\cQ\subset E$ is a D-structure.
\end{example}

Given that Courant algebroids correspond to symplectic Lie $2$-algebroids (cf. Theorem~\ref{thm:QP2Courant}), and in view of the possibility of developing an AKSZ-type sigma-model formulation for quantization of a metric algebroid, it is natural to wonder what objects metric algebroids correspond to in graded geometry. In the remainder of this section we develop this correspondence in some detail. Whereas our main result (Theorem~\ref{thm:1-1metric}) should not be surprising to experts, here we follow a more contemporary approach to the geometrization of degree~$2$ manifolds, see e.g.~\cite{delCarpio,Jotz2018}, and hence offer a new geometric perspective on the correspondence.

\medskip

\subsection{Involutive Double Vector Bundles}~\\[5pt]
\label{subsec:invvecbun}
In the proof of Proposition~\ref{prop:LiealgebroidQ} we saw that degree~$1$ manifolds correspond geometrically to vector bundles: if $\cM=(\cQ,\cA)$ is a $1$-graded manifold, then $\cM\simeq E[1]$ for a vector bundle $E\to\cQ$. Let us now recall some well-known facts about the geometrization of degree~$2$ manifolds. In particular, we discuss the implications of Batchelor\rq{}s Theorem for degree 2 manifolds, see~\cite{Jotz2018}. 

\begin{theorem}
Any {degree $n$ manifold} is (non-canonically) isomorphic to a {split degree~$n$ manifold}.
\end{theorem}

Let us spell this out explicitly for a degree 2 manifold $\cM=(\cQ, \cA)$ over a manifold $\cQ,$ where $\cA$ is its sheaf of functions. The subsheaves $\cA^1$ and $ \cA^2$ of $\cA$, consisting of functions of degree $1$ and $2$, respectively, are locally free finitely-generated $C^{\infty}(\cQ)$-modules. Hence there exist vector bundles $E\to \cQ$ and $ \bar{F}\to\cQ$ such that $\cA^1 \simeq \mathsf{\Gamma}(E^*)$ and $\cA^2 \simeq \mathsf{\Gamma}(\bar{F}^*).$ Abusing notation slightly, the subalgebra of $\cA$ generated by $C^{\infty}(\cQ) \oplus \mathsf{\Gamma}(E^*)$ is isomorphic to $\mathsf{\Gamma}(\midwedge^{\bullet} E^*).$ Thus 
$$
\mathsf{\Gamma}(\midwedge^{\bullet} E^*) \cap \mathsf{\Gamma}(\bar{F}^*) \simeq \mathsf{\Gamma}(\midwedge^2 E^*)
$$ 
is a proper $C^{\infty}(\cQ)$-submodule of $\mathsf{\Gamma}(\bar{F}^*).$ Since, as a sheaf of functions, $\mathsf{\Gamma}(\bar{F}^*) / \mathsf{\Gamma}(\midwedge^2 E^*)$ is again a locally free finitely-generated $C^{\infty}(\cQ)$-module, there is a vector bundle $F$ over $\cQ$ such that 
$$
\mathsf{\Gamma}(\bar{F}^*)\, \big/\, \mathsf{\Gamma}(\midwedge^2 E^*) \simeq \mathsf{\Gamma}(F^*) \ . 
$$
This gives a short exact sequence of $C^{\infty}(\cQ)$-modules 
\be\nonumber
0 \longrightarrow \mathsf{\Gamma}(\midwedge^2 E^*) \longrightarrow  \mathsf{\Gamma}(\bar{F}^*) \longrightarrow  \mathsf{\Gamma}(F^*) \longrightarrow 0 \ , 
\ee
which yields a short exact sequence of the underlying vector bundles over $\cQ$:
\be \label{sesdeg2}
0 \longrightarrow \midwedge^2 E^* \xlongrightarrow{i} \bar{F}^* \xlongrightarrow{p} F^* \longrightarrow 0 \  .
\ee
A choice of splitting of either of these sequences gives an isomorphism $\cM \simeq E[1]\oplus F[2].$

\begin{remark}
This construction also yields a one-to-one correspondence (up to isomorphisms) between degree~2 manifolds and pairs of vector bundles $(E,\bar{F})$ with a surjective vector bundle morphism $\bar{p} \colon \bar{F} \rightarrow \midwedge^2 E.$ 
In other words, degree~2 manifolds are in one-to-one correspondence with \textit{involutive sequences} (see \cite{delCarpio}), i.e. short exact sequences of the form
\be \label{involutivesequence}
0 \longrightarrow F  \xlongrightarrow{} \bar{F} \xlongrightarrow{\bar{p}} \midwedge^2 E \longrightarrow 0 \  ,
\ee
where $F=\ker(\bar p).$ This in turn aids in understanding the correspondence between degree~2 manifolds and involutive double vector bundles, and hence the correspondence with metric double vector bundles~\cite{Jotz2018, delCarpio}. For background and details on double vector bundles that we use in the following, see~\cite{Mackenzie}.
\end{remark}

\begin{remark}
Let $\cM=(\cQ, \cA)$ be a degree 2 manifold. Consider the corresponding involutive double vector bundle $(D; E, E; \cQ),$ with core bundle $F \rightarrow \cQ,$ given by the commutative diagram of vector bundles 
\be \label{involutiveDVB}
\begin{tikzcd}
D \arrow{r}{q_1} \arrow{d}[swap]{q_2} & E \arrow{d}{q^E}  \\
E \arrow{r}[swap]{q^E} & \cQ   
\end{tikzcd}
\ee
endowed with a double vector bundle morphism $\cI:D\to D$ such that $\cI^2 =\unit_D,$ $q_1 \circ \cI = q_2,$ $q_2 \circ \cI= q_1$ and core morphism $-\unit_F \colon F \rightarrow F,$ where $F=\ker(q_1)\cap\ker(q_2)$. The map $D\to E\times_\cQ E$ to the fibred product, with respect to the horizontal and vertical base projections, is a surjective submersion whose kernel is the core bundle $F\to\cQ$. Its linear approximation is given by the linear sequence 
\be \label{linearsequence}
0 \longrightarrow {\mathsf{Hom}}(E,F) \xlongrightarrow{} \hat{E} \xlongrightarrow{} E \longrightarrow 0 \ , 
\ee
where $\mathsf{\Gamma}(\hat{E}) \simeq \mathsf{\Gamma}_{\mathtt{lin}}(D)$, the \emph{linear} sections of $D$, that is, the sections of $q_1$ which are bundle morphisms from the vertical bundle $q^E$ to $q_2$ covering sections of the horizontal bundle $q^E$, and $\mathsf{\Gamma}(F) \simeq \mathsf{\Gamma}_{\mathtt{core}}(D),$ the \emph{core} sections of $D$. One shows that the degree $-1$ and $-2$ vector fields on $\cM$ are given by $\frX_{-1}(\cM) \simeq \mathsf{\Gamma}(\hat{E})$ and $\frX_{-2}(\cM) \simeq \mathsf{\Gamma}(F).$

Following~\cite{delCarpio}, we define a tensor $W \in \mathsf{\Gamma}(\midodot^2 \hat{E}^* \otimes \bar{F})$ as follows. Choose any splitting of \eqref{linearsequence}. Then any section $\hat e \in \mathsf{\Gamma}(\hat{E})$ can be written as
\be \nonumber
\hat e = \tau + e  \qquad \mbox{with} \quad \tau \in \mathsf{\Gamma}\big({\mathsf{Hom}}(E,F)\big) \qquad \mbox{and} \qquad e \in \mathsf{\Gamma}(E) \ .
\ee
We set
\be\nonumber
W(\hat e_1, \hat e_2) \coloneqq \tau_1(e_2) + \tau_2 (e_1) \ .
\ee
It is proven in \cite{delCarpio} that this definition does not depend on the choice of splitting. Notice that $W(\hat e_1, \hat e_2) \in \mathsf{\Gamma}(F)$ and we identify $F$ with its image in $\bar{F}$ from \eqref{involutivesequence}. One further shows
\be\nonumber
W(\hat e_1, \hat e_2)=[\hat e_1, \hat e_2] \ \in \ \frX_{-2}(\cM) \simeq  \mathsf{\Gamma}(F)
\ee
where we regard $\hat e_1$ and $ \hat e_2$ as degree $-1$ vector fields on $\cM.$
\end{remark}

\begin{remark}
Let $D_1$ be the horizontal vector bundle $q_1 \colon D \rightarrow E$ in \eqref{involutiveDVB}. Then the space of fibrewise linear functions $C^{\infty}_{\mathtt{lin}}(D_1)$ is endowed with a vector bundle structure such that its dual $\hat{F} \coloneqq C^{\infty}_{\mathtt{lin}}(D_1)^*$ fits into the short exact sequence
\be\nonumber
0 \longrightarrow F \xlongrightarrow{} \hat{F} \xlongrightarrow{} E \otimes E \longrightarrow 0 \ .
\ee
It straightforwardly follows that 
\be \nonumber
\hat{F} \simeq \midodot^2 E \oplus \bar{F} \ .
\ee
A similar construction holds for the vertical vector bundle of \eqref{involutiveDVB}.
\end{remark}

We can now discuss the higher analogue of Proposition~\ref{prop:LiealgebroidQ} in the degree~$2$ case, within the weakened setting appropriate for our later considerations of metric algebroids. For this, we note that the subsheaf $\cA^3$ of degree $3$ functions is similarly isomorphic to $\mathsf{\Gamma}(\bar{D}^*)$ for some vector bundle $\bar{D}\rightarrow \cQ.$ This induces a short exact sequence of vector bundles given by
\be \label{sesfun3}
0 \longrightarrow \midwedge^3 E^* \xlongrightarrow{i} \bar{D}^* \xlongrightarrow{} E^* \otimes F^* \longrightarrow 0 \  .
\ee
The following result, proven in \cite{delCarpio}, provides the characterization of degree $3$ functions on degree $2$ manifolds.

\begin{theorem} \label{classdegree3}
Let $\cM=(\cQ, \cA)$ be a degree $2$ manifold. Then there is a one-to-one correspondence between degree $3$ functions $\gamma \in \cA^3$ and pairs of vector bundle morphisms 
\be \nonumber 
\boldsymbol{\gamma}_1 \colon F \longrightarrow E^*  \qquad \mbox{and} \qquad \boldsymbol{\gamma}_2 \colon \hat{E} \longrightarrow \bar{F}^*
\ee
satisfying
\begin{enumerate}
\item  $  \braket{\boldsymbol{\gamma}_1(\phi), \hat e}=\braket{\boldsymbol{\gamma}_2(\hat e), \phi},$ for all $ \phi \in F $ and $\hat e \in \hat{E}\,;$
~\\[1pt]
\item For all $\tau \in {\mathsf{Hom}}(E,F)$,
\be \label{prop2}
\boldsymbol{\gamma}_2 \circ \tau = (\boldsymbol{\gamma}_1 \circ \tau)^* - \boldsymbol{\gamma}_1 \circ \tau \ \in \ \midwedge^2 E^* \ ;
\ee
and
\item The symmetric part of $\boldsymbol{\gamma}_2$ is controlled by $W$:
\be \label{prop3}
\braket{\boldsymbol{\gamma}_2(\hat e_1), \hat e_2}+ \braket{\boldsymbol{\gamma}_2(\hat e_2), \hat e_1}= \boldsymbol{\gamma}_1 \big(W(\hat e_1, \hat e_2)\big) \ ,
\ee
for all $\hat e_1, \hat e_2 \in \hat{E}.$
\end{enumerate}
\end{theorem}

\begin{remark}
The morphisms in Theorem~\ref{classdegree3} can be defined as follows. Choosing a splitting of the short exact sequence \eqref{sesfun3}, any degree $3$ function $\gamma \in \cA^3$ can be written as 
\be \nonumber
\gamma=\gamma_1 + \gamma_2
\ee
where $\gamma_1 \in \mathsf{\Gamma}(E^* \otimes F^*)$ and $\gamma_2 \in \mathsf{\Gamma}(\midwedge^3 E^*).$ Then we define
\be \label{morphgamma1}
\boldsymbol{\gamma}_1(\phi)\coloneqq \braket{\phi, \gamma_1} \ \in \ \mathsf{\Gamma}(E^*) \ ,
\ee
where $\phi \in \mathsf{\Gamma}(F)$ and here $\braket{\,\cdot\,,\, \cdot\,}$ is the duality pairing between $F$ and the $F^*$-component of $E^* \otimes F^*.$ Given a splitting of the linear sequence \eqref{linearsequence}, for any $\hat e= e + \tau \in \mathsf{\Gamma}(\hat{E})$ with $e \in \mathsf{\Gamma}(E)$ and $\tau \in \mathsf{\Gamma}(E^* \otimes F),$ we set
\be \label{morphgamma2}
\boldsymbol{\gamma}_2(\hat e) \coloneqq \bigl( (\boldsymbol{\gamma}_1 \circ \tau)^* - \boldsymbol{\gamma}_1 \circ \tau - \iota_e \gamma_2 \bigr) + \boldsymbol{\gamma}_1^*(e) \ \in \ \mathsf{\Gamma}(\midwedge^2 E^* \oplus F^*) \simeq \mathsf{\Gamma}(\bar{F}^*) \ .
\ee
Notice that only the definition of $\boldsymbol{\gamma}_2$ depends on the choice of splitting of \eqref{linearsequence}, but it  is shown in \cite{delCarpio} that it is well-defined under changes of splitting. It is straightforward to see that the pair $(\boldsymbol{\gamma}_1, \boldsymbol{\gamma}_2)$ defined in \eqref{morphgamma1} and \eqref{morphgamma2} satisfy all three properties (1)--(3) of Theorem~\ref{classdegree3}.
\end{remark}

\begin{remark}
The pair of vector bundle morphisms $(\boldsymbol{\gamma}_1, \boldsymbol{\gamma}_2)$ corresponding to $\gamma \in \cA^3$ can also be characterized in terms of the identifications $\frX_{-1}(\cM) \simeq \mathsf{\Gamma}(\hat{E})$ and $\frX_{-2}(\cM) \simeq \mathsf{\Gamma}(F)$ as
\be\nonumber
\boldsymbol{\gamma}_1(\phi)= -\phi \cdot\gamma \qquad \mbox{and} \qquad \boldsymbol{\gamma}_2(\hat e)= -\hat e\cdot\gamma \ ,
\ee
where on the right-hand sides $\phi \in \mathsf{\Gamma}(F)$ acts as a degree $-2$ derivation and $\hat e \in  \mathsf{\Gamma}(\hat{E})$ as a degree $-1$ derivation on $\gamma$.
\end{remark}

\medskip

\subsection{Poisson Structures and VB-Algebroids}~\\[5pt]
\label{subsec:invLiealg}
To give further structure to the involutive sequences corresponding to degree~$2$ manifolds, we shall now require that our graded manifolds are endowed with a Poisson structure.

\begin{definition}
A \textit{degree} $n$ \textit{Poisson manifold} is a degree $n$ manifold $\cM=(\cQ, \cA)$ together with a degree $-n$ Poisson structure,  i.e. an $\R$-bilinear map $\{\,\cdot\,,\, \cdot\,\} \colon \cA \times \cA \rightarrow \cA$ satisfying 
\begin{align}
|\{f, g\}|&=|f|+|g|-n \ , \\[4pt]
\{f,g\}&=(-1)^{(|f|+n)\,(|g|+n)} \, \{g,f\} \ , \label{gcomm} \\[4pt]
\{f, g\,h\}&=\{f, g\}\,h + (-1)^{(|f|+n)\,|g|} \, g\,\{f,h\} \ , \label{gder} \\[4pt]
\{f, \{g,h\}\}&=\{\{f,g\},h\}+(-1)^{(|f|+n)\,(|g|+n)}\,\{g,\{f,h\}\} \ , \label{glei}
\end{align}
for all homogeneous functions $f, \, g, \, h \in \cA$ of degree $|f|$, $|g|$ and $|h|$, respectively.
\end{definition} 

Let $\cdo(E) \rightarrow \cQ$ denote the Lie algebroid of covariant differential operators on a vector bundle $E\to\cQ,$ i.e. its sections are differential operators acting on $\mathsf{\Gamma}(E).$

\begin{theorem} \label{algebroid2Poisson}
There is a one-to-one correspondence between degree $-2$ Poisson brackets on a degree 2 manifold $\cM=(\cQ,\cA)$ and the following structures on its involutive sequence of vector bundles \eqref{sesdeg2}:
\begin{enumerate}
\item A symmetric bilinear pairing $\braket{\,\cdot\,,\, \cdot\,}_{E^*}$ on $E^*;$
\item A Lie algebroid $(\bar{F}^*, [\,\cdot\,,\, \cdot\,]_{\bar F^*}, \bar\sfa);$ and
\item A Lie algebroid action $\Psi \colon \bar{F}^* \rightarrow \cdo(E^*)$ preserving $\braket{\,\cdot\,,\, \cdot\,}_{E^*}$ such that
\be \label{fstar1}
[\zeta, \varepsilon_1 \wedge \varepsilon_2]_{\bar F^*}=\big(\Psi(\zeta)\cdot\varepsilon_1\big)\wedge \varepsilon_2 + \varepsilon_1 \wedge \big(\Psi(\zeta)\cdot\varepsilon_2\big) \ ,
\ee
and
\be \label{fstar2}
\Psi(\varepsilon_1 \wedge \varepsilon_2)\cdot\varepsilon=\braket{\varepsilon_2, \varepsilon}_{E^*}\,\varepsilon_1 - \braket{\varepsilon_1, \varepsilon}_{E^*}\,\varepsilon_2 \ ,
\ee
for all $\varepsilon, \varepsilon_1, \varepsilon_2 \in \mathsf{\Gamma}(E^*)$ and $\zeta \, \in \mathsf{\Gamma}(\bar{F}^*).$ 
\end{enumerate}
\end{theorem}

\begin{proof}
Let $\{\,\cdot\,,\,\cdot\,\}$ be a degree $-2$ Poisson bracket on $\cM$, and define
\be\nonumber
\braket{\varepsilon_1, \varepsilon_2}_{E^*} \coloneqq \{\varepsilon_1, \varepsilon_2 \} \ , 
\ee
for all $\varepsilon_1, \varepsilon_2 \in  \mathsf{\Gamma}(E^*).$ Then $\braket{\,\cdot\,,\, \cdot\,}_{E^*}$ is symmetric and bilinear by the graded skew-symmetry \eqref{gcomm} and the graded derivation rule \eqref{gder}. This gives~(1).

For any degree 2 function $\zeta \in \mathsf{\Gamma}(\bar{F}^*)$, the graded derivation rule \eqref{gder} implies that $\{ \zeta, \,\cdot\, \} \in \mathsf{Der}(C^{\infty}(\cQ))$ is a derivation of the algebra of functions $C^\infty(\cQ)$, and we can define a vector bundle morphism $\bar{\sfa} \colon \bar{F}^* \rightarrow T\cQ$ in terms of the corresponding morphism of $C^{\infty}(\cQ)$-modules $\bar{\sfa}\colon \mathsf{\Gamma}(\bar{F}^*) \rightarrow \mathsf{\Gamma}(T\cQ)$ by
\be\nonumber
\bar{\sfa}(\zeta)\cdot f\coloneqq \{ \zeta, f\} \ ,
\ee
where $\bar{\sfa}(g\,\zeta)\cdot f= g \, \bar{\sfa}(\zeta)\cdot f$ for all $f,g\in C^\infty(\cQ)$ also because of \eqref{gder}. The vector space $\mathsf{\Gamma}(\bar{F}^*)$ can be endowed with the bracket 
\be\nonumber
[\zeta_1, \zeta_2]_{\bar F^*} \coloneqq \{ \zeta_1, \zeta_2 \} 
\ee
for any degree 2 functions $\zeta_1  , \zeta_2  \in \mathsf{\Gamma}(\bar{F}^*),$ because $\{ \cA^2 , \cA^2\} \subseteq \cA^2.$ This is a skew-symmetric bracket because of \eqref{gcomm} and it satisfies the (ungraded) Jacobi identity because the Poisson bracket does on degree 2 functions, i.e. $[\,\cdot\,,\, \cdot\,]_{\bar F^*}$ is a Lie bracket. Lastly, $\bar\sfa \colon \mathsf{\Gamma}(\bar{F}^*) \rightarrow \mathsf{\Gamma}(T\cQ) $ is a bracket homomorphism:
\begin{align}
\bar{\sfa}([\zeta_1, \zeta_2]_{\bar F^*})\cdot f =\{ \{\zeta_1, \zeta_2 \}, f\} = \{\zeta_1, \{ \zeta_2 , f\} \}- \{ \zeta_2, \{ \zeta_1 , f\}\} = 
[\bar{\sfa}(\zeta_1), \bar{\sfa}(\zeta_2)]_{T\cQ}\cdot f  \ , \nonumber
\end{align}
for all $\zeta_1,  \zeta_2\in \mathsf{\Gamma}(\bar{F}^*)$ and $f\in C^\infty(\cQ)$, where in the second equality we used the graded Jacobi identity \eqref{glei}. Thus $(\bar{F}^*, [\,\cdot\,,\, \cdot\,]_{\bar F^*},  \bar{\sfa})$ is a Lie algebroid over $\cQ$. This gives~(2).

Define the vector bundle morphism $\Psi \colon \bar{F}^* \rightarrow \cdo(E^*)$ by 
\be \nonumber
\Psi(\zeta)\cdot\varepsilon \coloneqq \{ \zeta, \varepsilon\}
\ee
for all $\zeta \in \mathsf{\Gamma}(\bar{F}^*)$ and $\varepsilon \in \mathsf{\Gamma}(E^*),$ where we used $\{\cA^2, \cA^1 \} \subseteq \cA^1.$ The graded derivation rule \eqref{gder} implies that $\Psi$ takes values in $\cdo(E^*)$ with symbols given by $\bar\sfa$:
\be \nonumber
\Psi(\zeta)\cdot(f\, \varepsilon) = \{ \zeta, f \,\varepsilon\}= f\,\{\zeta, \varepsilon \}+ \{ \zeta, f \}\, \varepsilon= f\,\big( \Psi(\zeta)\cdot\varepsilon\big) + \big(\bar{\sfa}(\zeta)\cdot f\big)\,\varepsilon \ .
\ee
Furthermore, $\Psi$ is a Lie algebroid morphism:
\begin{align}
\Psi([\zeta_1, \zeta_2]_{\bar F^*})\cdot\varepsilon = \Psi(\zeta_1)\cdot\big(\Psi(\zeta_2)\cdot\varepsilon\big)-\Psi(\zeta_2)\cdot\big(\Psi(\zeta_1)\cdot\varepsilon\big) = [\Psi(\zeta_1), \Psi(\zeta_2)]_\circ\cdot\varepsilon \ , \nonumber
\end{align}
as a result of the graded Jacobi identity \eqref{glei}. The graded Jacobi identity also
gives 
\be \nonumber
\bar\sfa(\zeta)\cdot\braket{\varepsilon_1, \varepsilon_2}_{E^*} = \braket{\Psi(\zeta)\cdot\varepsilon_1,\varepsilon_2}_{E^*} + \braket{\varepsilon_1, \Psi(\zeta)\cdot\varepsilon_2}_{E^*} \ ,
\ee
and so the Lie algebroid action $\Psi$ preserves the symmetric bilinear pairing $\braket{\,\cdot\,,\, \cdot\,}_{E^*}.$ Lastly, the equations \eqref{fstar1} and \eqref{fstar2} are immediate consequences of the graded derivation rule \eqref{gder}. This gives~(3).

Conversely, given (1)--(3), we define the Poisson brackets $\{\cA^2, \cA^2\}$ and $\{\cA^2, \cA^0\}$ from the Lie algebroid on $\bar{F}^*$, $\{\cA^2, \cA^1\}$ from the $\braket{\,\cdot\,,\, \cdot\,}_{E^*}$-preserving Lie algebroid action of $\bar{F}^*$ on $E^*,$ $\{\cA^1, \cA^1\}$ from the symmetric bilinear form $\braket{\,\cdot\,,\, \cdot\,}_{E^*}$ on $E^*$, and we set $\{\cA^1, \cA^0\}=0.$ These brackets are then extended to arbitrary functions by the graded derivation rule. 
\end{proof}

\begin{corollary} \label{Cor:algebroid2Poisson}
Under the correspondence of Theorem~\ref{algebroid2Poisson}, the involutive sequence
\be \nonumber
0 \longrightarrow \midwedge^2 E^* \xlongrightarrow{i} \bar{F}^* \xlongrightarrow{p} F^* \longrightarrow 0 
\ee
is a short exact sequence of Lie algebroids.
\end{corollary}

\begin{proof}
The equations \eqref{fstar1} and \eqref{fstar2} give
\be \nonumber
[\varepsilon_1 \wedge \varepsilon_2 , \varepsilon_3\wedge \varepsilon_4]_{\bar F^*} = \braket{\varepsilon_2, \varepsilon_3}_{E^*}\,\varepsilon_1 \wedge \varepsilon_2 - \braket{\varepsilon_1, \varepsilon_3}_{E^*}\,\varepsilon_2 \wedge \varepsilon_4 + \braket{\varepsilon_2, \varepsilon_4}_{E^*}\,\varepsilon_3 \wedge \varepsilon_1 - \braket{\varepsilon_1, \varepsilon_4}_{E^*}\,\varepsilon_3 \wedge \varepsilon_2  \ ,
\ee
hence $\midwedge^2 E^*$ is a Lie subalgebroid of $\bar{F}^*.$  Thus $\midwedge^2 E^*$ is a Lie algebroid ideal of $\bar{F}^*$ because of the involutivity of $\midwedge^2 E^*$ and \eqref{fstar1}. The restriction of the anchor map $\bar\sfa$ to $\midwedge^2 E^*$ vanishes:
\be \nonumber
\bar\sfa(\varepsilon_1 \wedge \varepsilon_2)\cdot f = \{\varepsilon_1 \wedge \varepsilon_2, f \}=\{ f, \varepsilon_1\} \wedge \varepsilon_2 + \varepsilon_1 \wedge \{f, \varepsilon_2\}=0
\ee
because $\{ f, \varepsilon_1\}=\{f,\varepsilon_2\}=0$ for degree reasons. Hence $F^*$ can be endowed with a Lie algebroid given by the bracket
\be \nonumber
[p(\zeta_1), p(\zeta_2)]_{F^*} := p([\zeta_1, \zeta_2]_{\bar F^*})
\ee
and anchor $\sfa \colon F^* \rightarrow T\cQ $ through which $\bar\sfa:\bar F^*\to T\cQ$ factors:
\be \nonumber
\bar{\sfa}= \sfa \circ p \ .
\ee
Therefore \eqref{sesdeg2} is a short exact sequence of Lie algebroids.
\end{proof}

\begin{remark} \label{VBalgebroid}
Recall~\cite{Mackenzie} that a VB-algebroid is a Lie algebroid object in the category of vector bundles, and that a double vector bundle is precisely a VB-algebroid with trivial Lie algebroid structures. Theorem \ref{algebroid2Poisson} implies that the horizontal dual to \eqref{involutiveDVB}, that is, the double vector bundle 
\be \nonumber
\begin{tikzcd}
D_{E}^* \arrow{r}{q_E} \arrow{d}[swap]{q_{F^*}} & E \arrow{d}{q^E}  \\
F^* \arrow{r}[swap]{q^{F^*}} & \cQ   
\end{tikzcd}
\ee
with core bundle $E^*\to \cQ$, can be endowed with a VB-algebroid structure as follows:
The anchor map $\sfa_D \colon D_{E}^* \rightarrow TE $ is defined by
\begin{align}
\sfa_D (\varepsilon_1)\cdot f  \coloneqq 0 \ , \
\sfa_D (\varepsilon_1)\cdot\varepsilon_2 \coloneqq \braket{\varepsilon_1, \varepsilon_2}_{E^*} \ , \
\sfa_D (\zeta)\cdot f \coloneqq \bar\sfa(\zeta)\cdot f \ , \
\sfa_D (\zeta)\cdot\varepsilon_1 \coloneqq \Psi(\zeta)\cdot\varepsilon_1 \ , \nonumber
\end{align}
for all $f \in C^{\infty}(\cQ),$ $\varepsilon_1, \varepsilon_2  \in \mathsf{\Gamma}(E^*)$ and $\zeta  \in \mathsf{\Gamma}(\bar{F}^*);$ here we identify $\mathsf{\Gamma}(E^*) \simeq C^{\infty}_{\mathtt{lin}}(E).$
The VB-algebroid bracket is given by
\begin{align}
[\varepsilon_1, \varepsilon_2]_D  \coloneqq 0 \ , \quad
[\zeta_1, \varepsilon_1]_D \coloneqq \Psi(\zeta_1)\cdot\varepsilon_1 \qquad \mbox{and} \qquad
[\zeta_1, \zeta_2]_D \coloneqq [\zeta_1, \zeta_2]_{\bar F^*} \ , \nonumber
\end{align}
for all $\varepsilon_1,\varepsilon_2 \in \mathsf{\Gamma}(E^*)$ and $\zeta_1, \zeta_2 \in \mathsf{\Gamma}(\bar{F}^*);$ here we identify $\mathsf{\Gamma}(E^*) \simeq \mathsf{\Gamma}_{\mathtt{core}}(D_E^*)$ and extend the VB-algebroid bracket to any section of
$\hat{F}^* \simeq \midodot^2 E^* \otimes \bar{F}^*$ by the derivation rule.
For a general statement about the one-to-one correspondence between degree~2 Poisson manifolds and metric VB-algebroids, see \cite{delCarpio, Jotz2018}.
\end{remark}

\begin{remark}
From Remark \ref{VBalgebroid} it follows that the double vector bundle $(D; E, E; \cQ)$ given by \eqref{involutiveDVB} is endowed with a \emph{double linear Poisson structure}, i.e. a Poisson structure which is linear with respect to both vector bundle structures. The bundle $D$ is further endowed with a Lie algebroid differential $\de_D$ induced by the VB-algebroid structure on $D_E^*.$ There is an isomorphism $\Phi:D\to T^*E$ to the double vector bundle
\be 
\begin{tikzcd}
T^*E \arrow{r}{q_E} \arrow{d}[swap]{q_{E^*}} & E \arrow{d}{q^E}  \\
E^* \arrow{r}[swap]{q^{E^*}} & \cQ   
\end{tikzcd}
\ee
with core bundle $T^*\cQ.$ Its linear sequence is given by 
\be \label{lincotangent}
0 \longrightarrow E \otimes T^*\cQ \xlongrightarrow{} J^1(E) \xlongrightarrow{} E \longrightarrow 0 \ ,
\ee
where $J^1(E)$ is the first order jet bundle of $E$.
In particular, the double linear Poisson structure on $(D; E, E; \cQ)$ is the pullback by the isomorphism $\Phi \colon D \rightarrow T^*E$ of the canonical Poisson structure on $T^*E.$
Hence the vector space of linear sections $\mathsf{\Gamma}_{\mathtt{lin}}(D) \simeq \mathsf{\Gamma}(\hat{E})$ is spanned by sections of the form $\de_D \varepsilon,$ with $\varepsilon  \in \mathsf{\Gamma}(E^*)$, while the vector space of core sections $\mathsf{\Gamma}_{\mathtt{core}}(D) \simeq \mathsf{\Gamma}(F)$ is spanned by sections of the form $\de f,$ with $f  \in C^{\infty}(\cQ).$
\end{remark}

We are now ready to look at the interplay between degree~$3$ functions and degree~$-2$ Poisson structures.

\begin{proposition} \label{deg3Poisson}
Let $\gamma \in \cA^3$ be a degree 3 function on a degree 2 Poisson manifold $(\cM, \{\, \cdot\,,\, \cdot\,\})$, and let $(\boldsymbol{\gamma}_1, \boldsymbol{\gamma}_2)$ be its characteristic pair of vector bundle morphisms. Then
\be
\boldsymbol{\gamma}_1(\de f)= \{\gamma, f\}  \qquad \mbox{and} \qquad \boldsymbol{\gamma}_2(\de_D \varepsilon)= \{\gamma, \varepsilon\} \ , 
\ee 
for all $f\in C^{\infty}(\cQ)$ and $\varepsilon \in \mathsf{\Gamma}(E^*).$
\end{proposition}

\begin{proof}
Choosing a splitting of the short exact sequence \eqref{sesfun3}, we can write $\gamma \in \cA^3$ as 
\be
\gamma= \varepsilon_0 \otimes \zeta + \varepsilon_1 \wedge \varepsilon_2 \wedge \varepsilon_3 \ ,
\ee  
where $\varepsilon_i\in\mathsf{\Gamma}(E^*)$ and $\zeta\in\mathsf{\Gamma}(F^*)$. 
Then
\be \nonumber
\{\gamma, f\}= \{ \varepsilon_0 \otimes \zeta, f\}= \varepsilon_0 \, \{\zeta, f\}= \varepsilon_0 \, \braket{\zeta, \de f} = \boldsymbol{\gamma}_1(\de f) \ ,
\ee 
where we identify $\de f$ with its corresponding section of $F$, and here $\braket{\,\cdot\,,\, \cdot\,}$ is the duality pairing between $F$ and $F^*.$ The last equality follows from \eqref{morphgamma1}.

The second equality is obtained from calculating
\begin{align}
\{\gamma, \varepsilon\} 
&= \braket{\varepsilon_0, \varepsilon}_{E^*}\, \zeta - \varepsilon_0 \, \wedge \{\zeta, \varepsilon\}+ \braket{\varepsilon, \varepsilon_1}_{E^*}\, \varepsilon_2 \wedge \varepsilon_3 - \braket{\varepsilon, \varepsilon_2}_{E^*}\, \varepsilon_1 \wedge \varepsilon_3 +\braket{\varepsilon, \varepsilon_3}_{E^*}\, \varepsilon_1 \wedge \varepsilon_2 \ . \nonumber  
\end{align}
On the other hand, by identifying $\de_D \varepsilon$ with its corresponding section in $\mathsf{\Gamma}(\hat{E})$ and using \eqref{morphgamma2}, we have
\be \nonumber
\boldsymbol{\gamma}_2(\de_D \varepsilon)= \braket{\varepsilon_0, \varepsilon}_{E^*}\, \zeta - \varepsilon_0 \wedge \braket{\zeta, \de_D \varepsilon} + \braket{\varepsilon, \varepsilon_1}_{E^*}\, \varepsilon_2 \wedge \varepsilon_3 - \braket{\varepsilon, \varepsilon_2}_{E^*}\, \varepsilon_1 \wedge \varepsilon_3 +\braket{\varepsilon, \varepsilon_3}_{E^*}\, \varepsilon_1 \wedge \varepsilon_2 \ .
\ee
The second equality now follows from
\be \nonumber
\{\zeta, \varepsilon\}= \braket{\zeta, \de_D \varepsilon} \ ,
\ee
where here $\braket{\,\cdot\,,\, \cdot\,}$ is the duality pairing between $F$ and $F^*$ for the component of $\de_D \varepsilon$ in $\mathsf{\Gamma}(E^* \otimes F^*)$.
\end{proof}

\medskip

\subsection{Symplectic 2-Algebroids and Symplectic Almost Lie 2-Algebroids}~\\[5pt]
\label{subsec:symplectic2alg}
In order to establish a correspondence with pseudo-Euclidean vector bundles, we require our graded manifold to be endowed with a symplectic structure.

\begin{definition} \label{gradedsymplectic}
A \emph{symplectic degree $n$ manifold} is a degree $n$ manifold $\cM=(\cQ, \cA)$ with degree $-n$ Poisson structure $\{\, \cdot\,,\, \cdot\, \}$ such that, for all $x \in \cQ,$ there is an open subset $U \subset \cQ$ containing $x$ and local coordinates $(x^i, w^\alpha)$, where $x^i$ are degree~$0$ coordinates and $w^\alpha$ are coordinates of degrees $1,\dots,n$, such that the matrix of the degree~$0$ components of the Poisson bracket at $x$, $\big\{ (x^i, w^{\alpha}), (x^j, w^{\beta})\big\}^0(x),$ is non-degenerate. 
\end{definition}

We are now ready to discuss the correspondence between symplectic degree 2 manifolds and pseudo-Euclidean vector bundles together with their gauge symmetries, through a result due to Roytenberg \cite{Roytenberg2002}.

\begin{proposition} \label{atiyahsymplectic}
There is a one-to-one correspondence between symplectic degree $2$ manifolds $(\cM, \{\,\cdot\,,\, \cdot\,\})$ and pseudo-Euclidean vector bundles $(E^*, \braket{\,\cdot\,,\, \cdot\,}_{E^*}).$ The associated Lie algebroid $(F^*, [\,\cdot\,,\, \cdot\,]_{F^*}, \sfa)$ is isomorphic to the tangent Lie algebroid $(T\cQ, [\,\cdot\,,\, \cdot\,]_{T\cQ},\unit_{T\cQ}),$ and its involutive sequence of Lie algebroids
\be \nonumber
0 \longrightarrow \midwedge^2 E^* \xlongrightarrow{i} \bar{F}^* \xlongrightarrow{p} F^* \longrightarrow 0 
\ee
is isomorphic (as Lie algebroids) to the Atiyah sequence of $(E^*, \braket{\,\cdot\,,\, \cdot\,}_{E^*})$:
\be \nonumber
0 \longrightarrow \frso(E^*) \xlongrightarrow{} \IA(E^*, \braket{\,\cdot\,,\, \cdot\,}_{E^*}) \longrightarrow T\cQ \longrightarrow 0 \ .
\ee
\end{proposition}

\begin{proof}
Let $(x^i, \varepsilon^a, \zeta^\mu)$ be local coordinates on $\cM=(\cQ, \cA)$ of degrees $(0,1,2)$. Since the Poisson bracket $\{\, \cdot\,,\, \cdot\,\}$ is symplectic, it follows that
\be \label{detsymplectic}
\det\big(\{ \zeta^\mu, x^i \}(x)\big) \neq 0 \qquad \mbox{and} \qquad \det\big(\{ \varepsilon^a, \varepsilon^b \}(x)\big) \neq 0 \ ,
\ee 
for all $x \in \cQ.$ The second condition in \eqref{detsymplectic} implies that the symmetric bilinear pairing $\braket{\,\cdot\,,\, \cdot\,}_{E^*}$ on $E^*$ constructed in Theorem~\ref{algebroid2Poisson} is non-degenerate, so it endows $E^*$ with the structure of a pseudo-Euclidean vector bundle over $\cQ$. The first condition in \eqref{detsymplectic} shows that the anchor map $\sfa \colon F^* \rightarrow T\cQ$ defined in Corollary~\ref{Cor:algebroid2Poisson} is an isomorphism. From Theorem~\ref{algebroid2Poisson} and Corollary~\ref{Cor:algebroid2Poisson} there is a morphism of Lie algebroid sequences given by
\be\nonumber
\begin{tikzcd}
\midwedge^2 E^* \arrow{r}{i} \arrow{d} & \bar{F}^* \arrow{r}{p} \arrow{d}{\Psi} & F^* \arrow{d}{\sfa} \\
\frso(E^*) \arrow{r} & \IA(E^*, \braket{\,\cdot\,,\, \cdot\,}_{E^*})  \arrow{r} & T\cQ 
\end{tikzcd}
\ee
From the non-degeneracy of $\braket{\,\cdot\,,\, \cdot\,}_{E^*}$ it follows that $\midwedge^2  E^* \simeq \frso(E^*).$ Since this is a commutative diagram and $\sfa$ is an isomorphism, it follows that $\Psi \colon \bar{F}^* \rightarrow \IA(E^*, \braket{\,\cdot\,,\, \cdot\,}_{E^*}) $ is an isomorphism as well. 
\end{proof}

\begin{remark}
Since a symplectic degree 2 manifold is associated with a pseudo-Euclidean vector bundle $(E^*, \braket{\,\cdot\,,\, \cdot\,}_{E^*}),$ the constructions here and in the following can be made directly on the vector bundle $E,$ as discussed in \cite{Roytenberg2002}, because of the isomorphism $E\simeq E^*$ induced by the pseudo-Euclidean metric.
\end{remark}

We are finally ready to discuss the correspondence with metric algebroids. For this, we introduce the appropriate weakening of the notion of symplectic Lie $n$-algebroid from Definition~\ref{def:symplnalg}.

\begin{definition}
A \emph{symplectic $n$-algebroid} is a symplectic degree~$n$ manifold $(\cM,\{\,\cdot\,,\,\cdot\,\})$ endowed with a degree~$n+1$ function $\gamma\in\cA^{n+1}$, or equivalently (by Remark~\ref{rem:sympln}) a degree~$1$ symplectic vector field $Q\in\frX_1(\cM)$.
\end{definition}

\begin{example}\label{ex:symplLien}
A symplectic Lie $n$-algebroid is precisely a symplectic $n$-algebroid which is also a dg-manifold.
\end{example}

\begin{remark}
Symplectic $n$-algebroids are called `symplectic nearly Lie $n$-algebroids' (for $n=2$) in~\cite{Bruce2019}. In~\cite{Samann2018} they are refered to as `symplectic pre-N$Q$-manifolds of degree~$n$', while in~\cite{Heller2016} they are called `pre-QP-manifolds'.
\end{remark}

The main result of this section, inspired by \cite{delCarpio}, is the following weakening of Theorem~\ref{thm:QP2Courant}.

\begin{theorem}\label{thm:1-1metric}
There is a one-to-one correspondence between symplectic $2$-algebroids and metric algebroids.
\end{theorem}

\begin{proof}
Let $(\cM, \{\,\cdot\,,\, \cdot\,\}, \gamma)$ be a symplectic $2$-algebroid, and consider its involutive sequence of Lie algebroids. By~Proposition \ref{atiyahsymplectic}, the vector bundle $E^*\rightarrow \cQ$ is endowed with a fibrewise pseudo-Euclidean metric $\braket{\,\cdot\,,\, \cdot\,}_{E^*}$. 
Define 
\be\nonumber
\rho(\varepsilon)\cdot f \coloneqq - \{\{\gamma, \varepsilon\}, f\} \ ,
\ee
for all $\varepsilon \in \mathsf{\Gamma}(E^*)$ and $f \in C^{\infty}(\cQ).$ By the derivation property of the Poisson bracket, this defines a map
\be \nonumber
\rho \colon   \mathsf{\Gamma}(E^*) \longrightarrow {\mathsf{Der}}\big(C^\infty(\cQ)\big) \ ,
\ee
which is a morphism of $C^{\infty}(\cQ)$-modules and thus induces a vector bundle morphism
\be \nonumber
\rho \colon E^* \longrightarrow T\cQ \ .
\ee
The D-bracket on $E^*$ is given by
\be \nonumber
\llbracket\varepsilon_1, \varepsilon_2\rrbracket_\dorf \coloneqq -\{\{\gamma, \varepsilon_1\},\varepsilon_2\}
\ee
for all $\varepsilon_1,  \varepsilon_2 \in \mathsf{\Gamma}(E^*).$
The compatibility conditions \eqref{eq:metric1} and \eqref{eq:metric2} of Definition~\ref{malg} follow straightforwardly from the graded Jacobi identity for the graded Poisson bracket.
As a further check, the anchored derivation property \eqref{eq:anchorLeibniz} for $\llbracket\,\cdot\,,\, \cdot\,\rrbracket_\dorf$ follows from the graded derivation property of the Poisson bracket. Thus $(E^*,\llbracket\,\cdot\,,\,\cdot\,\rrbracket_\dorf,\langle\,\cdot\,,\,\cdot\,\rangle_{E^*},\rho)$ is a metric algebroid on $\cQ$.

For the converse statement we have to work a bit harder. Let $(E^*, \llbracket\,\cdot\,,\,\cdot\,\rrbracket_\dorf,\braket{\,\cdot\,,\, \cdot\,}_{E^*}, \rho)$ be a metric algebroid over $\cQ$. By Proposition~\ref{atiyahsymplectic}, the underlying pseudo-Euclidean vector bundle corresponds to a symplectic degree~$2$ manifold $(\cM,\{\,\cdot\,,\,\cdot\,\})$. Define the pair of vector bundle morphisms $(\boldsymbol{\gamma}_1, \boldsymbol{\gamma}_2)$ by the dual pairings
\begin{align}
\braket{\boldsymbol{\gamma}_1(\de f), \de_D \varepsilon} & \coloneqq \rho(\varepsilon)\cdot f \ , \label{g1} \\[4pt]
\braket{\boldsymbol{\gamma}_2(\de_D \varepsilon), \de f} & \coloneqq \rho(\varepsilon)\cdot f \ , \label{g2} \\[4pt]
\braket{\boldsymbol{\gamma}_2(\de_D \varepsilon_1 ), \de_D \varepsilon_2} & \coloneqq \llbracket\varepsilon_1, \varepsilon_2\rrbracket_{\dorf} \ . \label{g3}
\end{align}
We shall prove that the pair $(\boldsymbol{\gamma}_1, \boldsymbol{\gamma}_2)$ satisfies the properties (1)--(3) of Theorem~\ref{classdegree3}, and hence characterizes a degree~3 function $\gamma$ on $\cM$. Property~(1) of Theorem~\ref{classdegree3} follows immediately from the definitions \eqref{g1} and \eqref{g2}.

For property~(2), we compute
\begin{align}
\braket{(\varepsilon_1 \otimes \de f)^* \circ \boldsymbol{\gamma}_1^* - \boldsymbol{\gamma}_1 \circ (\varepsilon_1 \otimes \de f), \de_D \varepsilon_2} & = - \braket{\boldsymbol{\gamma}_1 (\de f),\de_D \varepsilon_2} \, \varepsilon_1 + \boldsymbol{\gamma}_1(\de f)\,\braket{\varepsilon_1, \varepsilon_2}_{E^*}  \nonumber \\[4pt]
&= - \big(\rho(\varepsilon_2)\cdot f\big)\,\varepsilon_1 + \cD f\,\braket{\varepsilon_1, \varepsilon_2}_{E^*} \ , \nonumber
\end{align}
where we use the usual definition of the generalized exterior derivative $\cD$ in the metric algebroid given by
\be \nonumber
\braket{\cD f, \varepsilon}_{E^*} \coloneqq \rho(\varepsilon)\cdot f \ .
\ee
On the other hand, from the Leibniz rule for the Lie algebroid differential $\de_D$ and the left derivation property \eqref{eq:leftderivation} of the D-bracket it follows that
\begin{align}
\braket{\boldsymbol{\gamma}_2(\varepsilon_1 \otimes \de f ), \de_D \varepsilon_2} &= \braket{\boldsymbol{\gamma}_2(\de_D (f \,  \varepsilon_1)- f \, \de_D \varepsilon_1), \de_D \varepsilon_2} \nonumber \\[4pt]
&=\llbracket f \, \varepsilon_1, \varepsilon_2\rrbracket_\dorf -f\, \llbracket\varepsilon_1, \varepsilon_2 \rrbracket_{ \dorf}\nonumber \\[4pt]
&= - \big(\rho(\varepsilon_2)\cdot f\big)\,\varepsilon_1 + \cD f\,\braket{\varepsilon_1, \varepsilon_2}_{E^*} \ , \nonumber
\end{align}
and hence \eqref{prop2} follows.

Finally, for property~(3) of Theorem~\ref{classdegree3}, we use \eqref{eq:Courantmetric} to get
\be \nonumber
\braket{\boldsymbol{\gamma}_2(\de_D \varepsilon_1), \de_D \varepsilon_2}+\braket{\boldsymbol{\gamma}_2(\de_D \varepsilon_2), \de_D \varepsilon_1}= \llbracket\varepsilon_1, \varepsilon_2\rrbracket_\dorf + \llbracket\varepsilon_2, \varepsilon_1\rrbracket_\dorf = \cD\braket{\varepsilon_1, \varepsilon_2}_{E^*} \ .
\ee
On the other hand, by choosing a splitting of the linear sequence \eqref{lincotangent} we can decompose $\de_D \varepsilon$ as $\de_D \varepsilon= (\de_D\varepsilon- \nabla \varepsilon) + \nabla \varepsilon$ where $\nabla\varepsilon\in\mathsf{\Gamma}(E^*\otimes T^*\cQ)$ defines a metric connection on $(E^*,\langle\,\cdot\,,\,\cdot\,\rangle_{E^*})$. Then we compute
\be \nonumber
W(\de_D \varepsilon_1, \de_D \varepsilon_2)= \braket{\nabla \varepsilon_1, \varepsilon_2} + \braket{\nabla \varepsilon_2, \varepsilon_1}=  \de \braket{\varepsilon_1, \varepsilon_2}_{E^*} \ ,
\ee
which gives
\be \nonumber
\boldsymbol{\gamma}_1\big(W(\de_D \varepsilon_1, \de_D \varepsilon_2)\big)= \cD\braket{\varepsilon_1, \varepsilon_2}_{E^*} \ ,
\ee
and therefore \eqref{prop3} follows.
\end{proof}

\begin{remark}\label{rem:boldgamma2}
By the definition of the vector bundle morphism $\boldsymbol{\gamma}_2$ in \eqref{g2} and \eqref{g3}, we have
\be\nonumber
\braket{\boldsymbol{\gamma}_2(\de_D \varepsilon_1), \de_D (f \, \varepsilon_2)}= \llbracket\varepsilon_1 , f \, \varepsilon_2\rrbracket_\dorf = f\,\llbracket\varepsilon_1, \varepsilon_2\rrbracket_\dorf + \big(\rho(\varepsilon_1)\cdot f\big)\,\varepsilon_2 \ .
\ee
This is consistent with the Leibniz rule for the Lie algebroid differential $\de_D$:
\begin{align}
\braket{\boldsymbol{\gamma}_2(\de_D \varepsilon_1), \de_D( f \, \varepsilon_2)}=\braket{\boldsymbol{\gamma}_2(\de_D \varepsilon_1),f\, \de_D \varepsilon_2 + \varepsilon_2 \otimes \de f} = f\,\llbracket\varepsilon_1, \varepsilon_2\rrbracket_\dorf + \big(\rho(\varepsilon_1)\cdot f\big)\,\varepsilon_2 \ . \nonumber
\end{align}

From Proposition \ref{deg3Poisson} it also follows that the structure maps of the metric algebroid can be written in terms of $\boldsymbol{\gamma}_2$ using the Poisson bracket and the Lie algebroid differential as 
\be\nonumber
\rho(\varepsilon)\cdot f = -\{\boldsymbol{\gamma}_2 (\de_D \varepsilon), f\} \qquad \mbox{and} \qquad \llbracket\varepsilon_1, \varepsilon_2\rrbracket_\dorf = -\{\boldsymbol{\gamma}_2 (\de_D \varepsilon_1), \varepsilon_2 \} \ .
\ee
\end{remark}

\begin{remark}\label{rem:Leibhom}
Introduce the map
\be \nonumber
\mathsf{Leib}_\dorf \colon \mathsf{\Gamma}(E^*)\times \mathsf{\Gamma}(E^*)\times \mathsf{\Gamma}(E^*)\longrightarrow \mathsf{\Gamma}(E^*)
\ee
which measures the failure of the Leibniz identity \eqref{eq:Leibnizid} for the D-bracket $\llbracket\,\cdot\,,\,\cdot\,\rrbracket_\dorf$ (and coincides with the Jacobiator for a skew-symmetric bracket); it is defined by
\begin{align}\label{eq:LeibC}
{\sf Leib}_\dorf(\varepsilon_1,\varepsilon_2,\varepsilon_3) := \llbracket \varepsilon_1,\llbracket \varepsilon_2,\varepsilon_3\rrbracket_\dorf\rrbracket_\dorf - \llbracket \llbracket \varepsilon_1,\varepsilon_2\rrbracket_\dorf,\varepsilon_3\rrbracket_\dorf - \llbracket \varepsilon_2,\llbracket \varepsilon_1,\varepsilon_3\rrbracket_\dorf\rrbracket_\dorf \ ,
\end{align}
for $\varepsilon_1,\varepsilon_2,\varepsilon_3\in\mathsf{\Gamma}(E^*)$. Similarly, we introduce the map 
\begin{align*}
{\sf hom}_\rho:\mathsf{\Gamma}(E^*)\times\mathsf{\Gamma}(E^*)\longrightarrow\mathsf{\Gamma}(T\cQ)
\end{align*}
which measures the failure of the anchor map $\rho:E^*\to T\cQ$ from being a bracket morphism to the Lie bracket of vector fields; it is defined by
\begin{align}\label{eq:homrho}
{\sf hom}_\rho(\varepsilon_1,\varepsilon_2) := \rho\big(\llbracket \varepsilon_1,\varepsilon_2\rrbracket_\dorf\big) - \big[\rho(\varepsilon_1),\rho(\varepsilon_2)\big]_{T\cQ} \ .
\end{align}
The map \eqref{eq:LeibC} is given in terms of third order higher derived brackets generated by $\{\gamma,\gamma\}$ as
\be\nonumber
{\mathsf{Leib}_\dorf}(\varepsilon_1, \varepsilon_2, \varepsilon_3)=- \tfrac12\, \{\{\{\{ \gamma, \gamma\}, \varepsilon_1\}, \varepsilon_2 \}, \varepsilon_3\} 
\ee
on the corresponding symplectic $2$-algebroid $(\cM,\{\,\cdot\,,\,\cdot\,\},\gamma)$, while \eqref{eq:homrho} is given by 
\be\nonumber
{\sf hom}_\rho(\varepsilon_1,\varepsilon_2)\cdot f = \tfrac12\, \{\{\{\{ \gamma, \gamma\}, f \}, \varepsilon_1 \}, \varepsilon_2\} \ ,
 \ee
for all $f \in C^\infty(\cQ).$

As noted in \cite{Bruce2019} (see also~\cite{Carow-Watamura:2020xij}), the maps \eqref{eq:LeibC} and \eqref{eq:homrho} are related by
\be\nonumber
\mathsf{Leib}_\dorf(\varepsilon_1, \varepsilon_2, f\, \varepsilon_3) - f\, \mathsf{Leib}_\dorf(\varepsilon_1, \varepsilon_2, \varepsilon_3)= - \bigl( {\sf hom}_\rho(\varepsilon_1,\varepsilon_2)\cdot f \bigr)\, \varepsilon_3 \ ,
\ee
for all $f \in C^\infty(\cQ).$ In other words, the lack of tensoriality of $\mathsf{Leib}_\dorf$ in its third entry measures the failure of the anchor map from being a bracket homomorphism. Similarly
\begin{align*}
{\sf hom}_\rho(f\,\varepsilon_1,\varepsilon_2) - f \, {\sf hom}_\rho(\varepsilon_1,\varepsilon_2) = \langle\varepsilon_1,\varepsilon_2\rangle_{E^*} \, \rho(\cD f) \ ,
\end{align*}
so the lack of tensoriality of ${\sf hom}_\rho$ in its first entry measures the violation of the condition $\rho\circ\cD=0$, or equivalently of \eqref{eq:rho0cD}.
\end{remark}

Remark~\ref{rem:Leibhom} and Definition~\ref{def:preCourant} motivate the following notion~\cite{Bruce2019}.

\begin{definition}\label{def:almostLie2}
A symplectic $2$-algebroid $(\cM,\{\,\cdot\,,\,\cdot\,\},\gamma)$ is a \emph{symplectic almost Lie 2-algebroid} if
\be\nonumber
 \{\{ \gamma, \gamma\}, f \}=0 \ ,
 \ee
 for all $f \in C^\infty(\cQ),$ or in other words if the square $Q^2$ of the corresponding Hamiltonian vector field $Q= \{\gamma, \,\cdot\,\}$ preserves the sheaf of functions of degree 0.
\end{definition}

\begin{corollary}\label{cor:preCourant}
There is a one-to-one correspondence between symplectic almost Lie $2$-algebroids and pre-Courant algebroids.
\end{corollary}

\section{Para-Hermitian Geometry}
\label{sec:parahermgeom}
The natural home for metric algebroids, and in particular split exact pre-Courant algebroids, is provided by para-Hermitian geometry. This provides a precise mathematical framework for a global notion of `doubled geometry' in string theory, as originally suggested by Vaisman~\cite{Vaisman2012,Vaisman2013}, and further developed by~\cite{Freidel2017,Svoboda2018,Freidel2019,SzMar,Mori2019,Marotta:2019eqc,Mori:2020yih}. In this setting, double field theory is formulated on an almost para-Hermitian manifold, as we shall discuss in Section~\ref{sec:applications}. 

\medskip

\subsection{Para-Hermitian Vector Bundles} \label{sec:paravector} ~\\[5pt]
\label{subsec:PHvectorbundles}
We start with an overview of the main ideas and their relation to some of the concepts introduced in Sections~\ref{sec:AKSZ} and~\ref{sec:metricalg}.

\begin{definition} \label{parahermvector}
Let $E\rightarrow \cQ$ be a vector bundle of even rank $2d$ over a manifold $\cQ$. A
\emph{para-complex structure} on $E$ is a vector bundle automorphism
$K \in {\sf Aut}(E)$ covering the identity such that
$K^2=\unit_E$, $K\neq\pm\,\mathds{1}_E$, and the $\pm\,1$-eigenbundles of $K$ have equal rank~$d$. The pair $(E,K)$ is a \emph{para-complex vector bundle}. 

If in addition $E$ admits a fibrewise metric $\braket{\,\cdot\,,\,\cdot\,}_E \in \mathsf{\Gamma}(\midodot^2 E^*)$ of split signature $(d,d)$ such that 
$$
\braket{K(e_1),K(e_2)}_E=-\braket{e_1,e_2}_E \ , 
$$
for all $e_1, e_2 \in \mathsf{\Gamma}(E)$, then the pair $(K, \braket{\,\cdot\,,\,\cdot\,}_E)$ is
a \emph{para-Hermitian structure} on $E$ and the triple $(E,K,\braket{\,\cdot\,,\,\cdot\,}_E)$ is a \emph{para-Hermitian vector bundle}. 

A \emph{para-Hermitian bundle morphism} from a para-Hermitian vector bundle $(E,K,\braket{\,\cdot\,,\,\cdot\,}_E)$ to a para-Hermitian vector bundle $(E',K',\braket{\,\cdot\,,\,\cdot\,}_{E'})$ over the same manifold is an isometry $\psi:(E,\braket{\,\cdot\,,\,\cdot\,}_E)\to (E',\braket{\,\cdot\,,\,\cdot\,}_{E'}) $ covering the identity which intertwines the para-complex structures: $\psi\circ K = K' \circ\psi$.
\end{definition}

The $\pm\,1$-eigenbundles $L_\pm$ of $K$ split the vector bundle $E$ into a Whitney sum
$$
E= L_+ \oplus L_- \ ,
$$ 
such that $L_\pm$ are maximally
isotropic with respect to the fibrewise metric $\braket{\,\cdot\,,\,\cdot\,}_E.$

\begin{remark}\label{rem:maxisosplit}
Let $E\to \cQ$ be a vector bundle of rank $2d$ endowed with a split
signature metric $\braket{\,\cdot\,,\,\cdot\,}_E$, and $L$ a maximally isotropic subbundle
of $E.$ Then the short exact sequence
\be
0\longrightarrow L\longrightarrow E\longrightarrow
E/L\longrightarrow 0 
\label{eq:EL-}\ee
always admits a maximally isotropic splitting. This determines a para-Hermitian structure on $E.$ All maximally isotropic splittings of the short exact sequence \eqref{eq:EL-} give isomorphic para-Hermitian structures on $E$.
\end{remark}

The compatibility condition between $\braket{\,\cdot\,,\,\cdot\,}_E$ and $K$ in Definition~\ref{parahermvector} is equivalent to 
$$
\braket{K(e_1), e_2}_E = -\braket{e_1, K(e_2)}_E \ , 
$$ 
for all $e_1,e_2 \in \mathsf{\Gamma}(E)$. A para-Hermitian vector
bundle $E$ is therefore endowed with a non-degenerate \emph{fundamental $2$-form} $\omega\in \mathsf{\Gamma}(\midwedge^2 E^*)$ given by
$$
\omega(e_1,e_2)= \braket{K(e_1), e_2}_E \ , 
$$ 
for all $e_1, e_2 \in \mathsf{\Gamma}(E)$. The eigenbundles $L_\pm \subset E$ are also maximally isotropic with respect to $\omega.$

\begin{example} \label{ex:gentanbun}
Let $E=A\oplus A^*$ be the Whitney sum of a vector bundle $A$ and its dual $A^*$ over a manifold $\cQ$. 
It is naturally endowed with the fibrewise split signature metric
$$
\braket{a+a^{\textrm{\tiny$\vee$}},b+b^\dual}_{A\oplus A^*} = \langle\,a,b^\dual\rangle + \langle b,a^\dual\rangle \ , 
$$
for all $a,b\in \mathsf{\Gamma}(A)$ and $a^\dual,b^\dual\in\mathsf{\Gamma}(A^*)$, where $\langle\,\cdot\,,\,\cdot\,\rangle$ is the canonical dual pairing between sections of $A$ and sections of $A^*$. The natural para-complex structure $K$ on $E$ is given by
$$
K(a+a^\dual)= a-a^\dual \ ,
$$ 
so that $A$ and
$A^*$ are the respective $\pm\,1$-eigenbundles. Then $\braket{\,\cdot\,,\,\cdot\,}_{A\oplus A^*}$ and
$K$ are compatible in the sense of Definition~\ref{parahermvector},
and the subbundles $A$ and $A^*$ are maximally isotropic with respect to $\braket{\,\cdot\,,\,\cdot\,}_{A\oplus A^*}.$ Thus we obtain a fundamental $2$-form
$$
\omega(a+a^\dual,b+b^\dual)= \langle\,a,b^\dual\rangle - \langle b,a^\dual\rangle \ , 
$$ 
which is the additional natural non-degenerate pairing that can be defined in this case. 

This construction applies to any split metric algebroid $(E,\llbracket\,\cdot\,,\,\cdot\,\rrbracket_\dorf,\braket{\,\cdot\,,\,\cdot\,}_E,\rho)$. A special instance is the generalized tangent bundle $E=\IT \cQ$ of Example~\ref{ex:standardCourant}, for which $A=T\cQ$.
\end{example}

\begin{example}\label{ex:exactmetricparaherm}
Let $(E,\llbracket\,\cdot\,,\,\cdot\,\rrbracket_\dorf,\braket{\,\cdot\,,\,\cdot\,}_E,\rho)$ be an exact pre-Courant algebroid on $\cQ$ specified by the short exact sequence \eqref{eq:metricexact} from Example~\ref{ex:exactmetric}, 
with fibrewise metric $\braket{\,\cdot\,,\,\cdot\,}_E$ and anchor map
$\rho: E \rightarrow T\cQ.$ 
From the definition of $\rho^*$ and exactness of the 
sequence \eqref{eq:metricexact}, it follows that the subbundle
$\Im(\rho^*) \subset E,$ which is isomorphic to $T^*\cQ,$ is maximally isotropic with respect to $\braket{\,\cdot\,,\,\cdot\,}_E.$
A choice of
isotropic splitting $s:T\cQ\to E$ of \eqref{eq:metricexact} gives a Whitney sum decomposition
$$
E= \Im(s) \oplus \Im(\rho^*) \ ,
$$ 
with an associated para-complex structure defined by
$$
K_s\big(s(X)+ \rho^*(\alpha)\big)=s(X)- \rho^*(\alpha) \ , 
$$ 
for all $ X \in \mathsf{\Gamma}(T\cQ)$ and $\alpha \in
\mathsf{\Gamma}(T^*\cQ).$  The para-complex structure $K_s$ is
compatible with the metric $\braket{\,\cdot\,,\,\cdot\,}_E,$ and in this way $E$ is endowed with a
para-Hermitian structure. This para-Hermitian structure is isomorphic to the para-Hermitian structure on the
generalized tangent bundle $\mathbb{T}\cQ$ from
Example~\ref{ex:gentanbun} with $A=T\cQ$. 
\end{example}

\medskip

\subsection{Generalized Metrics and Born Geometry} \label{genmesubs} ~\\[5pt]
We shall now introduce a notion of generalized metric associated to a para-Hermitian structure on a vector bundle.

\begin{definition} 
Let $E\rightarrow \cQ$ be a vector bundle endowed with a fibrewise pseudo-Euclidean
metric $\braket{\,\cdot\,,\,\cdot\,}_E$. A \emph{generalized} (\emph{Euclidean}) \emph{metric} on $E$ is an automorphism $I \in {\sf Aut}(E)$ such that
$I^2=\unit_E$, $I\neq\pm\,\mathds{1}_E$, and
\be \nonumber
\cH (e_1,e_2):=\braket{I(e_1),e_2}_E \ ,
\ee  
for all $e_1,e_2 \in \mathsf{\Gamma}(E)$, defines a fibrewise Euclidean metric $\cH$ on $E$. 
\end{definition}

A generalized metric determines a decomposition 
$$
E=V_+\oplus V_-
$$
into the $\pm\,1$-eigenbundles of $I$, such that the subbundle $V_+\subset E$ is maximally positive-definite with respect to the metric $\braket{\,\cdot\,,\,\cdot\,}_E$ and $V_-$ is the orthogonal complement of $V_+$ with
respect to $\braket{\,\cdot\,,\,\cdot\,}_E.$ Any generalized metric induces an isomorphism $\cH^\flat \in {\sf Hom}(E, E^*)$ which satisfies the compatibility condition
$$ 
\braket{\cH^\flat(e_1), \cH^\flat(e_2)}_E^{-1}=\braket{e_1,e_2}_E
$$
with the fibrewise metric $\braket{\,\cdot\,,\,\cdot\,}_E$.

This definition takes the following concrete form, proven in~\cite{Marotta:2019eqc}, when the metric $\braket{\,\cdot\,,\,\cdot\,}_E$ is part of a para-Hermitian structure.

\begin{proposition} \label{gbparaherm}
Let $(E,K, \braket{\,\cdot\,,\,\cdot\,}_E)$ be a para-Hermitian vector bundle over a manifold $\cQ$. A generalized metric $I \in {\sf Aut}(E)$ defines a unique pair $(g_+, b_+)$,
where $g_+ \in \mathsf{\Gamma}(\midodot^2 L_+^*)$ is a fibrewise Euclidean metric on the
$+1$-eigenbundle $L_+\subset E$ and $b_+\in \mathsf{\Gamma}(\midwedge^2 L_+^*)$ is a $2$-form on $L_+$.
Conversely, any such pair $(g_+,b_+)$ uniquely defines a generalized metric on $(E,K,\braket{\,\cdot\,,\,\cdot\,}_E)$.
\end{proposition}

Since the eigenbundles $L_\pm$ are both maximally isotropic with respect to $\braket{\,\cdot\,,\,\cdot\,}_E$, and $V_+$ is maximally positive-definite, it follows that $L_\pm\cap V_+=0$ and $L_\pm\cap V_-=0$. 
The pair $(g_+,b_+)$ induces a fibrewise metric $g_-$ on $L_-$ by
\begin{align*}
g_-(e_-,e_-') = g_+^{-1}\big(\braket{e_-,\,\cdot\,}_E,\braket{e_-',\,\cdot\,}_E\big) 
\end{align*}
and a skew-symmetric vector bundle map $B_+\in{\sf Hom}(L_+,L_-)$ by
\begin{align*}
\braket{B_+(e_+),e_+'}_E = b_+(e_+,e_+') = -\braket{e_+,B_+(e_+')}_E \ ,
\end{align*}
for all $e_\pm,e_\pm'\in\mathsf{\Gamma}(L_\pm)$. In the splitting $E=L_+\oplus L_-$ associated with the para-complex structure $K$, the Euclidean metric $\cH$ then assumes the matrix form
\begin{align*}
\cH = \bigg( \begin{matrix}
g_+ + B_+^{\rm t} \, g_-\,B_+ & -B_+^{\rm t} \, g_- \\
-g_-\,B_+ & g_-
\end{matrix} \bigg) \ ,
\end{align*}
where $B_+^{\rm t}:L_-^*\to L_+^*$ is the transpose map.

\begin{example}\label{ex:genmetricgengeom}
Let $E=\mathbb{T}\cQ=T\cQ\oplus
T^*\cQ$ be the generalized tangent bundle over a manifold $\cQ$. A generalized metric $I \in {\sf Aut}(\IT M)$ is equivalent to
a Riemannian metric $g$ and a $2$-form $b$ on $\cQ$.  This is a special case of the notion of generalized metric in generalized geometry~\cite{Jurco2016, gualtieri:tesi}, and $\cH$ assumes the standard form
\begin{align*}
\cH = \bigg(\begin{matrix}
g + b\,g^{-1}\,b &-b\,g^{-1} \\
-g^{-1} \, b & g^{-1}
\end{matrix}\bigg)
\end{align*}
with respect to the splitting $\mathbb{T}\cQ=T\cQ\oplus
T^*\cQ$.
\end{example}

\begin{definition} \label{compagenmetr}
A \emph{compatible generalized metric} on a para-Hermitian bundle $(E,K,\braket{\,\cdot\,,\,\cdot\,}_E)$ is a generalized metric $\mathcal{H}_0$ on $E$ which is compatible with the fundamental $2$-form $\omega$:
$$ 
\omega^{-1}\big(\cH_0^\flat(e_1), \cH_0^\flat(e_2)\big)=-\omega(e_1,e_2) \ ,
$$
for all $e_1, e_2  \in \mathsf{\Gamma}(E)$.
The triple $(K,\braket{\,\cdot\,,\,\cdot\,}_E,\mathcal{H}_0)$ is a \emph{Born geometry} on
$E$ and the quadruple $(E,K,\braket{\,\cdot\,,\,\cdot\,}_E,\cH_0)$ is a \emph{Born vector bundle}.
\end{definition}

A Born geometry is a particular type of generalized metric which can be concretely characterized as follows~\cite{Marotta:2019eqc}.

\begin{proposition}\label{prop:cHg+}
A Born geometry on a para-Hermitian vector bundle $(E,K,\braket{\,\cdot\,,\,\cdot\,}_E)$ is a generalized metric $\cH_0$ specified solely by a fibrewise metric $g_+$ on the eigenbundle $L_+.$ 
\end{proposition}

In other words, the compatible Euclidean metric $\cH_0$ can be regarded
as a choice of a metric on the subbundle $L_+$ in the splitting associated with $K$, where in matrix notation it reads
\be \label{eq:diaghermmetric}
\mathcal{H}_0 = \bigg( \begin{matrix}
g_+ & 0 \\ 0 & g_-
\end{matrix} \bigg) \ .
\ee

\medskip

\subsection{$B$-Transformations}\label{sec:Btransformations} ~\\[5pt]
To classify the distinct splittings of an exact pre-Courant algebroid, as well as to relate generic generalized metrics to the compatible generalized metrics of a Born geometry, we introduce the notion of a $B$-transformation for a
para-Hermitian vector bundle $(E,K,\braket{\,\cdot\,,\,\cdot\,}_E)$. Let us fix the splitting $E=L_+ \oplus L_-$ induced by the para-complex structure $K.$ Then any section $e\in\mathsf{\Gamma}(E)$ decomposes as $e=e_++e_-$ with $e_\pm\in\mathsf{\Gamma}(L_\pm)$, and $K\in {\sf
  Aut}(E)$ can be written as $K=
\unit_{L_+}- \unit_{L_-}.$

\theoremstyle{definition}
 \begin{definition}\label{def:B+}
Let $(E, K, \braket{\,\cdot\,,\,\cdot\,}_E)$ be a para-Hermitian vector bundle on a manifold~$\cQ$. A
$B_+$-\emph{transformation} is an isometry $e^{B_+}: E \rightarrow
E$ of $\braket{\,\cdot\,,\,\cdot\,}_E$ covering the identity which is given in matrix notation by
\begin{align*} \label{btra}
e^{B_+}= 
\bigg(\begin{matrix}
\mathds{1}_{L_+} & 0 \\
B_+ & \mathds{1}_{L_-}
\end{matrix}\bigg)
: E \longrightarrow E
\end{align*}
in the chosen splitting induced by $K$, where $B_+ : L_+
\rightarrow L_-$ is a skew-symmetric map:
$$
\braket{B_+(e_1),e_2}_E=- \braket{e_1, B_+(e_2)}_E \ ,
$$ 
for all $ e_1, e_2 \in \mathsf{\Gamma}(E)$.
\end{definition}

A $B_+$-transformation induces another 
para-complex structure from the para-Hermitian vector bundle
$(E,K,\braket{\,\cdot\,,\,\cdot\,}_E)$ given by the pullback $K_{B_+}=K- 2\,B_+$, which can be cast in the form 
\be \nonumber
K_{B_+}=e^{-B_+}\circ K\circ e^{B_+}=
\begin{pmatrix}
\unit_{L_+} & 0 \\
-2\,B_+ & -\unit_{L_-}
\end{pmatrix} \ .
\ee
Then $K_{B_+}^2= \unit_E,$ since $B_+(K(e))=-K(B_+(e))$ and
$B_+(B_+(e))=0,$ for all $e\in \mathsf{\Gamma}(E),$ and $K_{B_+}$ satisfies the compatibility condition $\braket{K_{B_+}(e_1),K_{B_+}(e_2)}_E=-\braket{e_1,e_2}_E$ with $\braket{\,\cdot\,,\,\cdot\,}_E$ because of the skew-symmetry property of $B_+.$
Thus $(K_{B_+}, \braket{\,\cdot\,,\,\cdot\,}_E)$ is a para-Hermitian structure on $E$. Only the $-1$-eigenbundle of the original splitting $E=L_+\oplus L_-$ is preserved by a
$B_+$-transformation, while the $+1$-eigenbundle changes.

To understand how the fundamental $2$-form $\omega$ changes under a $B_+$-transformation, we note that the endomorphism $B_+$ defines a $2$-form $b_+\in\mathsf{\Gamma}(\midwedge^2L_+^*)$ by 
$$
b_+(e_1,e_2)= \braket{B_+(e_1),e_2}_E \ ,
$$
for all $e_1,e_2\in \mathsf{\Gamma}(E).$ The fundamental 2-form $\omega_{B_+}$ of $(K_{B_+}, \braket{\,\cdot\,,\,\cdot\,}_E)$ is obtained by computing $\omega_{B_+}(e_1,e_2)= \braket{K_{B_+}(e_1),e_2}_E,$ which gives
\begin{align*}
\omega_{B_+}= \omega - 2\,b_+ \ .
\end{align*}
Thus a $B_+$-transformation does not generally preserve the closure or non-closure of the fundamental $2$-form.

\begin{remark}\label{rem:B-transf}
A completely analogous discussion applies to a \emph{$B_-$-transformation}, defined by a skew-symmetric map
$B_-: L_- \rightarrow L_+$ similarly to Definition~\ref{def:B+}, by interchanging the roles of the eigenbundles $L_+$ and $L_-$.
\end{remark}

\begin{remark}\label{rem:B+}
If $(E,\braket{\,\cdot\,,\,\cdot\,}_E,L)$ is an even rank vector bundle
endowed with a split signature metric and a choice of maximally
isotropic subbundle, see Remark~\ref{rem:maxisosplit}, then the
maximally isotropic splittings of the short exact sequence
\eqref{eq:EL-} are mapped into each other via $B_+$-transformations which preserve $L$.
\end{remark}

\begin{example}\label{ex:exactmetricB}
Recall from Example~\ref{ex:exactmetricparaherm} that every splitting of an exact pre-Courant algebroid $(E,\llbracket\,\cdot\,,\,\cdot\,\rrbracket_\dorf,\braket{\,\cdot\,,\,\cdot\,}_E,\rho)$ is associated with a para-Hermitian structure on $E$. By Remark~\ref{rem:B+}, any two splittings of an exact pre-Courant algebroid are related by a $B_+$-transformation. Each distinct isotropic
splitting of \eqref{eq:metricexact} is associated with a different $3$-form
$H\in\Omega^3(\cQ)$. A $B_+$-transformation of an exact pre-Courant algebroid is generated by a $2$-form $b\in\Omega^2(\cQ)$, which preserves the D-bracket $\llbracket\,\cdot\,,\,\cdot\,\rrbracket_\dorf$ if $b$ is a closed $2$-form. When $b$ is not closed the corresponding D-bracket maps to the Dorfman bracket twisted by $H+\de b$. 
\end{example}

Let us finally discuss the $B_+$-transformation of a
compatible generalized metric of a Born geometry. A compatible
generalized metric $\cH_0$ of a para-Hermitian structure $(K,
\braket{\,\cdot\,,\,\cdot\,}_E)$ on $E$ transforms under a $B_+$-transformation to the compatible
generalized metric $\cH_{B_+}$ of the pullback
para-Hermitian structure $(K_{B_+}, \braket{\,\cdot\,,\,\cdot\,}_E)$ on $E$ given by
\begin{align*}
\cH_{B_+}^\flat = \big(e^{-B_+}\big)^{\rm t}\circ\cH_0^\flat\circ e^{-B_+} \ .
\end{align*}
Recalling that
$\cH_0$ takes the diagonal form \eqref{eq:diaghermmetric}, we then have~\cite{Marotta:2019eqc}

\begin{proposition}
A generalized metric $I \in {\sf Aut}(E)$ on a para-Hermitian
vector bundle $(E, K, \braket{\,\cdot\,,\,\cdot\,}_E)$ corresponds to a choice of a Born geometry
$(K,\braket{\,\cdot\,,\,\cdot\,}_E, \cH_0)$ and a $B_+$-transformation.
\end{proposition}

\medskip

\subsection{Almost Para-Hermitian Manifolds} ~\\[5pt]
The special case where $E=TM$ is the tangent bundle of a manifold $M$ in Definition~\ref{parahermvector} is particularly important because it
allows one to formulate conditions for the integrability of the
eigenbundles $L_\pm$, and hence on the possibility that $M$ is a foliated
manifold. 

\begin{definition}\label{def:parahermmfd}
An \emph{almost para-Hermitian manifold} $(M,K,\eta)$ is a manifold $M$ of even dimension $2d$ with a para-Hermitian structure $(K,\eta)$ on its tangent bundle $TM$.
\end{definition}

The para-complex structure $K\in {\sf Aut}(TM)$ is equivalent to the splitting of the tangent bundle
\begin{align*}
TM = L_+\oplus L_-
\end{align*}
of the manifold $M$ into the Whitney sum of two distributions $L_\pm$ of the same constant rank~$d$, identified as the $\pm\,1$-eigenbundles of $K$. If the eigenbundles $L_\pm$ are both integrable, that is, $[\mathsf{\Gamma}(L_\pm),\mathsf{\Gamma}(L_\pm)]_{TM}\subseteq\mathsf{\Gamma}(L_\pm)$, then $(M,K,\eta)$ is a \emph{para-Hermitian manifold}; in this instance, by Frobenius' Theorem, $M$ admits two regular foliations $\cF_{\pm}$, such that $L_\pm=T\cF_\pm$. However, the integrability conditions for $L_+$ and $L_-$ are independent of each other~\cite{Freidel2017,Svoboda:2020msh}: one of them may be integrable while the other may not. This is the situation that most commonly occurs in examples, and in this case $M$ admits only one foliation.

The fundamental $2$-form $\omega$ of an almost para-Hermitian manifold $(M,K,\eta)$ defines an almost symplectic structure on $M$; the $3$-form $\de\omega$ describes the `generalized fluxes' of double field theory on $M$~\cite{Marotta:2019eqc}. If $\omega$ is symplectic, that is, $\de\omega=0$, then $(M,K,\eta)$ is an \emph{almost para-K\"ahler manifold}. In this case, since the subbundles $L_\pm$ are maximally isotropic with respect to $\omega$, they are Lagrangian subbundles of the tangent bundle $TM$; if one of them is integrable, then $M$ admits a Lagrangian foliation with respect to the symplectic structure~$\omega$.

\begin{example}\label{ex:cotangent}
Let $M$ be the total space of the cotangent bundle $\pi:T^*\cQ\to\cQ$ of a  $d$-dimensional manifold $\cQ$, endowed with the canonical symplectic $2$-form $\omega_0$. Then almost para-K\"ahler structures on $M$ correspond to isotropic splittings of the short exact sequence of vector bundles
\begin{align*}
0\longrightarrow \mathsf{Ker}(\pi_*)\xrightarrow{ \ i \ } TM \xrightarrow{\pi_*}\pi^*(T\cQ)\longrightarrow0 
\end{align*}
with respect to $\omega_0$, where $i$ is the inclusion of the vertical vector subbundle $\mathsf{Ker}(\pi_*)$ into $TM$. Such a splitting $s:\pi^*(T\cQ)\to TM$ defines an Ehresmann connection on $M$ with
\begin{align*}
TM = \Im(s) \oplus \mathsf{Ker}(\pi_*)
\end{align*}
the Whitney sum decomposition of the tangent bundle $TM$ into the $\pm\,1$-eigenbundles of an almost para-complex structure $K_s\in{\sf Aut}(TM)$ which is compatible with $\omega_0$. The leaves of the canonical foliation $\cF$ of the cotangent bundle $M=T^*\cQ$ are the fibres $\cF_q=\pi^{-1}(q)$ over $q\in \cQ$, which are diffeomorphic to $\IR^d$. There is a vector bundle isomorphism $\ker(\pi_*)\simeq T\cF$, and the quotient by the action of the foliation is $M/\cF\simeq\cQ$.

If the base $\cQ$ is a Riemannian manifold with metric $g$, then the horizontal lift of $g$, that is, the pullback $g_+=\pi^*g$, gives a fibrewise Euclidean metric on $\Im(s)$. This defines a Born geometry with compatible generalized metric $\cH_0$ on $(M,K_s,\omega_0)$ given by \eqref{eq:diaghermmetric}. Since any manifold $\cQ$ admits a Riemannian metric, one can always define a Born geometry on $M$ of this type. Similarly, given any $2$-form $b\in\Omega^2(\cQ)$, its horizontal lift $b_+=\pi^*b$ defines a $B_+$-transformation of the almost para-K\"ahler manifold $(M,K_s,\omega_0)$.
\end{example}

\medskip

\subsection{The Canonical Metric Algebroid} ~\\[5pt]
\label{sec:canonicalDbracket}
On any almost para-Hermitian manifold $(M,K,\eta)$, one can define the canonical D-bracket which makes the tangent bundle $TM$ into a metric algebroid over $M$ on which both eigenbundles $L_\pm$ of $K$ are D-structures~\cite{Freidel2017,Svoboda2018,Freidel2019,Svoboda:2020msh}. For this, we first need the following preliminary notion.

\begin{definition}\label{def:canconn}
Let $(M,K,\eta)$ be an almost para-Hermitian manifold. Let $\nabla^\LC$ be the Levi-Civita connection of $\eta$, and let $P_\pm:TM\to TM$ be the compositions of the inclusion maps $L_\pm\hookrightarrow TM$ with the projections $TM\to L_\pm$ to the $\pm\,1$-eigenbundles of $K$. The \emph{canonical connection} on $(M,K,\eta)$ is
\begin{align*}
\nabla^{\tt can} = P_+\circ\nabla^\LC\circ P_+ + P_-\circ\nabla^\LC\circ P_- \ .
\end{align*}
\end{definition}

The canonical connection is a \emph{para-Hermitian connection}: $\nabla^{\tt can}K=\nabla^{\tt can}\eta=0$; in particular, it preserves the eigenbundles $L_\pm$ of $K$. It coincides with the Levi-Civita connection, $\nabla^{\tt can}=\nabla^\LC$, if and only if $(M,K,\eta)$ is an almost para-K\"ahler manifold~\cite{Svoboda2018}. By a construction similar to Example~\ref{ex:metricalgeeta}, we then arrive at one of our central concepts.

\begin{definition}\label{thm:canDbracket}
Let $(M,K,\eta)$ be an almost para-Hermitian manifold, and let $\nabla^{\tt can}$ be its canonical connection. The \emph{canonical D-bracket} $\llbracket\,\cdot\,,\,\cdot\,\rrbracket_\dorf^K$ is defined by
\begin{align*}
\eta\big(\llbracket X,Y\rrbracket_\dorf^K,Z\big) = \eta(\nabla_X^{\tt can}Y - \nabla_Y^{\tt can}X,Z) + \eta(\nabla_Z^{\tt can}X,Y) \ ,
\end{align*}
for vector fields $X,Y,Z\in\mathsf{\Gamma}(TM)$. The metric algebroid $(TM,\llbracket\,\cdot\,,\,\cdot\,\rrbracket_\dorf^K,\eta,\unit_{TM})$ is the \emph{canonical metric algebroid} over $(M,K,\eta)$.
\end{definition}

The canonical D-bracket is compatible with the almost para-complex structure $K$, that is, both of its eigenbundles $L_\pm$ are D-structures on the metric algebroid $(TM,\llbracket\,\cdot\,,\,\cdot\,\rrbracket_\dorf^K,\eta,\unit_{TM})$:
\begin{align*}
\llbracket\mathsf{\Gamma}(L_\pm),\mathsf{\Gamma}(L_\pm)\rrbracket_\dorf^K \ \subseteq \ \mathsf{\Gamma}(L_\pm) \ .
\end{align*}
It is `canonical' because it is the projection of the Lie bracket of vector fields~\cite{Freidel2017}:
\begin{align*}
\llbracket P_\pm(X),P_\pm(Y)\rrbracket_\dorf^K = P_\pm\big([P_\pm(X),P_\pm(Y)]_{TM}\big) \ ,
\end{align*}
for all $X,Y\in\mathsf{\Gamma}(TM)$. Given $(M,K,\eta)$, the bracket $\llbracket\,\cdot\,,\,\cdot\,\rrbracket_\dorf^K$ is the unique D-bracket on $(TM,\eta,\unit_{TM})$ which is compatible with $K$ and related to the Lie bracket $[\,\cdot\,,\,\cdot\,]_{TM}$ in this way~\cite{Svoboda:2020msh}. In the canonical metric algebroid, the generalized exterior derivative is given by $\cD^{\tt can}=\eta^{-1}{}^\sharp\circ\de$.

\begin{remark}
Any D-bracket on $(TM, K, \eta)$ can be obtained by choosing a para-complex connection corresponding to a splitting of the short exact Atiyah sequence from Proposition~\ref{atiyahsymplectic}:
\be \label{eq:Atiyahtangent}
0 \longrightarrow {\mathfrak{so}}(TM) \longrightarrow \IA(TM, \eta) \longrightarrow TM \longrightarrow 0 \ ,
\ee
where $\IA(TM, \eta) \subset J^1(TM)$ is a vector subbundle of the first order jet bundle of the tangent bundle $TM$. In other words, the canonical connection on $(M,K,\eta)$ is just a change of splitting of the short exact sequence \eqref{eq:Atiyahtangent} when the splitting associated with the Levi-Civita connection $\nabla^\LC$ is chosen, leading to the canonical D-bracket. Any other change of splitting which yields a para-complex connection induces a D-bracket on $TM.$ By using the metric $\eta$ to identify $\mathfrak{so}(TM)\simeq\midwedge^2TM$, the different splittings of \eqref{eq:Atiyahtangent} correspond to bundle morphisms in ${\sf Hom}(\midwedge^2TM,TM)$, and hence the different D-brackets can be shown to correspond to $3$-forms $T\in\Omega^3(M)$, similarly to~\cite[Proposition~2.3]{Vaisman2012}.
\end{remark}

\begin{remark}\label{rem:DHtwisted}
Similarly to the D-bracket of a split exact pre-Courant algebroid (cf. Examples~\ref{ex:exactmetric} and~\ref{ex:exactmetricB}), the canonical D-bracket can be twisted by any closed $3$-form $H\in\Omega^3(M)$. In particular, a $B_+$-transformation maps the canonical D-bracket $\llbracket\,\cdot\,,\,\cdot\,\rrbracket_\dorf^K$ of $(M,K,\eta)$ to the canonical D-bracket of $(M,K_{B_+},\eta)$, which is equal to $\llbracket\,\cdot\,,\,\cdot\,\rrbracket_\dorf^K$ twisted by the $3$-form $\de b$~\cite{Svoboda2018,Svoboda:2020msh}:
\begin{align*}
\eta\big(\llbracket X,Y\rrbracket_\dorf^{K_{B_+}},Z\big) = \eta\big(\llbracket X,Y\rrbracket_\dorf^K,Z\big) - \de b(X,Y,Z) \ ,
\end{align*}
for all $X,Y,Z\in\mathsf{\Gamma}(TM)$.

Under the correspondence of Theorem~\ref{thm:1-1metric}, the canonical metric algebroid on $(M,K,\eta)$ can be identified with a symplectic $2$-algebroid $(T[1]M\oplus T^*[2]M,\{\,\cdot\,,\,\cdot\,\},\gamma)$, where the Poisson bracket $\{\,\cdot\,,\,\cdot\,\}$ coincides with the metric $\eta$ on degree~$1$ functions $\mathsf{\Gamma}(TM)$. In this language a $B_+$-transformation is a particular example of a canonical transformation of the symplectic $2$-algebroid, which twists the Hamiltonian $\gamma$ (see e.g.~\cite{Heller2016,Kokenyesi:2018xgj}).
\end{remark}

\begin{remark}\label{rem:Dfluxes}
The canonical D-bracket also gives a notion of relative weak integrability of almost para-Hermitian structures~\cite{Svoboda2018}. If $(K,\eta)$ and $(K',\eta)$ are almost para-Hermitian structures on the same even-dimensional manifold $M$, with respective eigenbundles $L_\pm$ and $L_\pm'$, then $K'$ is said to be \emph{D-integrable} with respect to $K$ if $\llbracket\mathsf{\Gamma}(L_\pm'),\mathsf{\Gamma}(L_\pm')]_\dorf^K\subseteq\mathsf{\Gamma}(L_\pm')$. The lack of D-integrability is then measured by the \emph{fluxes}
\begin{align*}
T(X,Y,Z) := \eta\big(\llbracket X,Y\rrbracket_\dorf^K - \llbracket X,Y\rrbracket_\dorf^{K'},Z\big) \ ,
\end{align*}
for $X,Y,Z\in\mathsf{\Gamma}(TM)$. The $3$-forms $T\in\Omega^3(M)$ reproduce the standard generalized fluxes of double field theory~\cite{SzMar}.
\end{remark}

\section{The Metric Algebroids of Doubled Geometry}
\label{sec:doubledgeom}
In this section we will make precise some notions of doubled geometry, and in particular its algebroid structures, which are relevant to a global description of double field theory.

\medskip

\subsection{Metric Algebroids from the Large Courant Algebroid} ~\\[5pt]
\label{sec:metricfromlarge}
Let $(M,K,\eta)$ be an almost para-Hermitian manifold of dimension $2d$. We have seen that there are two natural metric algebroids that can be defined over $M$ in this case. Firstly, there is the canonical metric algebroid $(TM,\llbracket\,\cdot\,,\,\cdot\,\rrbracket_\dorf^K,\eta,\unit_{TM})$ defined on the tangent bundle $TM$ of rank $2d$ in Definition~\ref{thm:canDbracket}. Secondly, there is the standard Courant algebroid $(\IT M,[\,\cdot\,,\,\cdot\,]_{\dorf},\langle\,\cdot\,,\,\cdot\,\rangle_{\IT M},\rho)$ defined on the generalized tangent bundle $\IT M=TM\oplus T^*M$ of rank $4d$ in Example~\ref{ex:standardCourant}; we call this latter metric algebroid the \emph{large Courant algebroid} on $M$ to distinguish it from other smaller rank Courant algebroids that will appear on $M$ momentarily. Both the canonical D-bracket on $TM$ and the Dorfman bracket on $\IT M$ can be twisted by a closed $3$-form $H\in\Omega^3(M)$, and both classes of metric algebroids are classified by $B_+$-transformations which shift the twisting $3$-form to $H+\de b$ for $2$-forms $b\in\Omega^2(M)$ (cf.~Remark~\ref{rem:DHtwisted} and Example~\ref{ex:exactmetricB}). We shall now discuss the precise relation between these two metric algebroids, following~\cite{Jonke2018,Hu:2019zro}.

For this, let us start from a more general setting. Let $(M,\eta)$ be any pseudo-Riemannian manifold. On the para-Hermitian vector bundle $\IT M\to M$ (cf.\ Example~\ref{ex:gentanbun}) we can define a generalized split signature metric by the involution
\begin{align*}
\cI_0 = \bigg( \begin{matrix}
0 & \eta^{-1}{}^\sharp \\
\eta^\flat & 0
\end{matrix} \bigg)
\end{align*}
with respect to the splitting $\IT M = TM\oplus T^*M$,
and its $B_+$-transformations $\cI_{B_+}$ given by $\cI_{B_+}=e^{B_+}\circ\cI_0\circ e^{-B_+}$. This defines another para-complex structure on $\IT M$ which splits the generalized tangent bundle into its $\pm\,1$-eigenbundles $C_\pm$ of equal rank:
\begin{align*}
\IT M = C_+\oplus C_- \qquad \mbox{with} \quad C_\pm = \mathsf{Graph}(b\pm\eta) \ ,
\end{align*}
where $\mathsf{Graph}(b\pm\eta) = \{X+b^\flat(X)\pm\eta^\flat(X)\, |\, X\in\mathsf{\Gamma}(TM)\}\subset\IT M$. We denote by
\begin{align*}
\sfp_\pm : \IT M\longrightarrow C_\pm
\end{align*}
the projections to the subbundles $C_\pm$ of $\IT M$, and by $\sfi_\pm:C_\pm\hookrightarrow \IT M$ the inclusion maps. 

By this process of doubling, splitting and projecting via the generalized tangent bundle $\IT M$, we get a pair of metric algebroids on any pseudo-Riemannian manifold.

\begin{proposition}\label{prop:DFTfromLarge}
Let $(\IT M,[\,\cdot\,,\,\cdot\,]_{\dorf},\langle\,\cdot\,,\,\cdot\,\rangle_{\IT M},\rho)$ be the $H$-twisted standard Courant algebroid over a pseudo-Riemannian manifold $(M,\eta)$, and let $C_\pm\subset \IT M$ be the eigenbundles associated to the generalized metric \smash{$\cI_{B_+}$}. Define structure maps on sections of the vector bundles $C_\pm\to M$ by
\begin{align*}
\llbracket \xi_\pm,\xi_\pm' \rrbracket_\dorf^{\pm} &:= \sfp_\pm\big([\sfi_\pm(\xi_\pm),\sfi_\pm(\xi_\pm')]_{\dorf}\big) \ , \\[4pt]
\langle\xi_\pm,\xi_\pm'\rangle_{C_\pm} &:= \pm\,\tfrac12\,\langle\sfi_\pm(\xi_\pm),\sfi_\pm(\xi_\pm')\rangle_{\IT M} \ , \\[4pt]
\rho_\pm &:= \rho\circ\sfi_\pm \ ,
\end{align*}
for $\xi_\pm,\xi_\pm'\in\mathsf{\Gamma}(C_\pm)$. Then $(C_\pm,\llbracket\,\cdot\,,\,\cdot\,\rrbracket_\dorf^{\pm},\langle\,\cdot\,,\,\cdot\,\rangle_{C_\pm},\rho_\pm)$ are metric algebroids over $M$.
\end{proposition}

The construction of Example~\ref{ex:metricalgeeta} then yields

\begin{proposition}\label{prop:DFTcanonical}
Let $(M,\eta)$ be a pseudo-Riemannian manifold. Then the metric algebroids $(C_\pm,\llbracket\,\cdot\,,\,\cdot\,\rrbracket_\dorf^{\pm},\langle\,\cdot\,,\,\cdot\,\rangle_{C_\pm},\rho_\pm)$ of Proposition~\ref{prop:DFTfromLarge} are isomorphic to the metric algebroid $(TM,\llbracket\,\cdot\,,\,\cdot\,\rrbracket_\dorf^\eta,\eta,\unit_{TM})$ with D-bracket twisted by the $3$-forms $\pm\,(H+\de b)$.
\end{proposition}
\begin{proof}
The anchor maps $\rho_\pm:C_\pm\to TM$ defined in Proposition~\ref{prop:DFTfromLarge} are bundle isomorphisms: the inverse maps $\rho_\pm^{-1}:TM\to C_\pm$ send $X\in\mathsf{\Gamma}(TM)$ to $X + (b^\flat\pm\eta^\flat)X$. One then checks explicitly  
\begin{align*}
\eta(X,Y)=\big\langle\rho_\pm^{-1}(X),\rho_\pm^{-1}(Y)\big\rangle_{C_\pm} 
\end{align*}
along with
 \begin{align*} 
\big\langle\llbracket\rho_\pm^{-1}(X),\rho_\pm^{-1}(Y)\rrbracket_\dorf^{\pm},\rho_\pm^{-1}(Z)\big\rangle_{C_\pm} = \eta\big(\llbracket X,Y\rrbracket_\dorf^\eta,Z\big) \pm\tfrac12\,(H+\de b)(X,Y,Z) \ ,
\end{align*}
for all $X,Y,Z\in\mathsf{\Gamma}(TM)$,
and that the anchor maps of the two metric algebroids are compatible, which is satisfied trivially as $\rho_\pm\circ\rho_\pm^{-1}=\unit_{TM}$.
\end{proof}

Let us finally specialize Proposition~\ref{prop:DFTcanonical}, with $b=0$ and $H=0$, to the case that the metric $\eta$ is part of an almost para-K\"ahler structure $(K,\eta)$. We then immediately arrive at

\begin{proposition}\label{cor:DFTfromLarge}
Let $(\IT M,[\,\cdot\,,\,\cdot\,]_{\dorf},\langle\,\cdot\,,\,\cdot\,\rangle_{\IT M},\rho)$ be the large Courant algebroid over an almost para-K\"ahler manifold $(M,K,\eta)$. Then the anchor map $\rho_+:C_+\to TM$ defines a metric algebroid isomorphism from the metric algebroid $(C_+,\llbracket\,\cdot\,,\,\cdot\,\rrbracket_\dorf^{+},\langle\,\cdot\,,\,\cdot\,\rangle_{C_+},\rho_+)$ with $b=0$ to the canonical metric algebroid $\big(TM,\llbracket\,\cdot\,,\,\cdot\,\rrbracket_\dorf^{K'},\eta,\unit_{TM}\big)$ corresponding to the para-complex structure $K'=\rho_+\circ(K-K^{\rm t})\circ\rho_+^{-1}$. 
\end{proposition}

\begin{remark}\label{rem:DFTfromLarge}
The construction of Proposition~\ref{cor:DFTfromLarge} was originally given for flat para-K\"ahler manifolds in~\cite{Jonke2018} (see also~\cite{Chatzistavrakidis:2019huz}), using the local model of the cotangent bundle $M=T^*\cQ$ with the canonical symplectic structure $\omega_0$ (cf. Example~\ref{ex:cotangent}), and in this way recovering the well known local expression for the canonical D-bracket $\llbracket\,\cdot\,,\,\cdot\,\rrbracket_\dorf^{+}$ with $b=0$ and $H=0$. It was extended to any generalized para-K\"ahler manifold in~\cite{Hu:2019zro} (see also~\cite{Svoboda:2020msh}), using Proposition~\ref{prop:DFTcanonical} with $b\neq0$ and $H\neq0$, where the D-brackets $\llbracket\,\cdot\,,\,\cdot\,\rrbracket_\dorf^{\pm}$ coincide with the canonical D-brackets of the two para-Hermitian structures $K_\pm$ associated to the generalized para-K\"ahler structure; the relative flux for this pair (cf. Remark~\ref{rem:Dfluxes}) is precisely the twisting $3$-form of the underlying large Courant algebroid: $T=H+\de b$.
\end{remark}

\medskip

\subsection{Doubled Manifolds and DFT Algebroids} ~\\[5pt]
\label{subsec:DFTalgebroid}
The constructions of Section~\ref{sec:metricfromlarge} motivate a generalization to a special class of metric algebroids where the deviation from a Courant algebroid is done in a controlled way. As we discuss below, this is the essence of the `section constraint' in double field theory (DFT for short). For this, we consider a special class of manifolds on which the geometry of double field theory, or \emph{doubled geometry}, is based. We denote by ${\sf O}(d,d)$ the split orthogonal group, whose maximal compact subgroup is ${\sf O}(d)\times{\sf O}(d)$.

\begin{definition}\label{def:doubledmfd}
A \emph{doubled manifold} is a manifold $M$ of even dimension $2d$ endowed with an ${\sf O}(d,d)$-structure, or equivalently a pseudo-Riemannian metric $\eta$ of split signature $(d,d)$.
\end{definition}

\begin{example}
Any almost para-Hermitian manifold $(M,K,\eta)$ is a doubled manifold, with a reduction of the ${\sf O}(d,d)$-structure to an ${\sf O}(d){\times}{\sf O}(d)$-structure. The introduction of a compatible generalized metric $\cH_0$ further reduces this to an ${\sf O}(d)$-structure.
\end{example}

\begin{definition}\label{def:DFTalgebroid}
A \emph{DFT algebroid} over a doubled manifold $(M,\eta)$ is a metric algebroid $(C,\llbracket\,\cdot\,,\,\cdot\,\rrbracket_\dorf,\langle\,\cdot\,,\,\cdot\,\rangle_C, \rho)$ whose anchor map defines an isometric isomorphism $\rho:(C,\langle\,\cdot\,,\,\cdot\,\rangle_C)\to (TM,\eta)$ of pseudo-Euclidean vector bundles over $M$ such that $\rho\circ\rho^*=\eta^{-1}{}^\sharp$.
\end{definition}

Sections of the vector bundle $C\to M$ are called \emph{doubled vectors} or \emph{DFT vectors}. All constructions of Section~\ref{sec:metricfromlarge} (in the split signature case) clearly fit into this general definition. On the other hand, Courant algebroids and pre-Courant algebroids over doubled manifolds are \emph{not} DFT algebroids, because their anchor maps have non-trivial kernels (cf.~Remark~\ref{rem:Courantminaxioms}). The following are noteworthy particular cases of DFT algebroids that will appear later on.

\begin{example}\label{ex:DFTLarge}
The metric algebroid $(C_+,\llbracket\,\cdot\,,\,\cdot\,\rrbracket_\dorf^{+},\langle\,\cdot\,,\,\cdot\,\rangle_{C_+},\rho_+)$ over any almost para-K\"ahler manifold $(M,K,\eta)$, constructed in Proposition~\ref{cor:DFTfromLarge}, is a DFT algebroid.
\end{example}

\begin{example}\label{ex:DFTcanonical}
The canonical metric algebroid $\big(TM,\llbracket\,\cdot\,,\,\cdot\,\rrbracket_\dorf^{K},\eta,\unit_{TM}\big)$ over any almost para-Hermitian manifold $(M,K,\eta)$, constructed in Section~\ref{sec:canonicalDbracket}, is a DFT algebroid.
\end{example}

From Definition~\ref{def:DFTalgebroid} we can describe how the key Courant algebroid properties \eqref{eq:Leibnizid}, \eqref{eq:anchorbracket} and \eqref{eq:rho0cD} are explicitly violated in a DFT algebroid in terms of the underlying geometry of the doubled manifold. Recalling Remark~\ref{rem:Leibhom}, we can write the map \eqref{eq:LeibC} in terms of \eqref{eq:homrho} as
\begin{align*}
& {\sf Leib}_\dorf(c_1,c_2,c_3) \\[4pt]
& \hspace{1cm} = \rho^{-1} \big( [\rho(c_1),{\sf hom}_\rho(c_2,c_3)]_{TM} - [{\sf hom}_\rho(c_1,c_2),\rho(c_3)]_{TM} - [\rho(c_2),{\sf hom}_\rho(c_1,c_3)]_{TM} \\
& \hspace{3cm} + {\sf hom}_\rho(c_1,\llbracket c_2,c_3\rrbracket_\dorf) - {\sf hom}_\rho(\llbracket c_1,c_2\rrbracket_\dorf,c_3) - {\sf hom}_\rho(c_2,\llbracket c_1,c_3\rrbracket_\dorf) \big) \ ,
\end{align*}
for all $c_1,c_2,c_3\in\mathsf{\Gamma}(C)$, which follows from applying the anchor map to \eqref{eq:LeibC} and using the Jacobi identity for the Lie bracket $[\,\cdot\,,\,\cdot\,]_{TM}$ on $\mathsf{\Gamma}(TM)$. This shows that the failure of the anchor map from being a bracket homomorphism completely controls the violation of the Leibniz identity \eqref{eq:Leibnizid} in a DFT algebroid.

\begin{lemma}\label{prop:DFTproperties}
Let $(C,\llbracket\,\cdot\,,\,\cdot\,\rrbracket_\dorf,\langle\,\cdot\,,\,\cdot\,\rangle_C, \rho)$ be a DFT algebroid over a doubled manifold $(M,\eta)$. Then
\begin{align*}
\langle\cD f,\cD g\rangle_C &= \eta^{-1}(\de f,\de g) \ , \\[4pt]
{\sf hom}_\rho(c_1,c_2) \cdot f + {\sf hom}_\rho(c_2,c_1) \cdot f &= \eta^{-1}\big(\de f,\de \langle c_1,c_2\rangle_C\big)  \ , \\[4pt]
\big\langle{\sf Leib}_\dorf(c_1,c_2,c_3) + {\sf Leib}_\dorf(c_2,c_1,c_3) , c_4\big\rangle_C &= - \eta\big({\sf hom}_\rho\big(\rho^{-1}\circ\eta^{-1}{}^\sharp(\de\langle c_1,c_2\rangle_C),c_3\big)\,,\,\rho(c_4)\big) \\
& \quad \, - \eta\big(\big[\eta^{-1}{}^\sharp(\de\langle c_1,c_2\rangle_C),\rho(c_3)\big]_{TM}\,,\,\rho(c_4)\big)   \ ,
\end{align*}
for all $f,g\in C^\infty(M)$ and $c_1,c_2,c_3,c_4\in\mathsf{\Gamma}(C)$.
\end{lemma}
\begin{proof}
The first equality follows immediately from Definition~\ref{def:DFTalgebroid}, which implies that the generalized exterior derivative in a DFT algebroid is given by
\begin{align}\label{eq:cDDFT}
\cD = \rho^{-1}\circ\eta^{-1}{}^\sharp\circ\de \ ,
\end{align}
and using the isometry property
\begin{align}\label{eq:DFTisometry}
\eta\big(\rho(c_1),\rho(c_2)\big) = \langle c_1,c_2\rangle_C \ .
\end{align}

For the second equality, we note that, as in any metric algebroid, the properties \eqref{eq:metric1} and \eqref{eq:Courantmetric} imply
\begin{align*}
{\sf hom}_\rho(c_1,c_2) \cdot f = \langle\llbracket\cD f,c_1\rrbracket_\dorf,c_2\rangle_C \ .
\end{align*} 
The symmetric part of ${\sf hom}_\rho(c_1,c_2)$ may then be written as
\begin{align*}
\big({\sf hom}_\rho(c_1,c_2) + {\sf hom}_\rho(c_2,c_1)\big)\cdot f = (\rho\circ\cD f)\cdot\langle c_1,c_2\rangle_C = \eta^{-1}{}^\sharp(\de f)\cdot\langle c_1,c_2\rangle_C \ ,
\end{align*}
where in the first step we used \eqref{eq:metric1} again and the second step follows from \eqref{eq:cDDFT}.

For the third equality, we use \eqref{eq:Courantmetric} and \eqref{eq:cDDFT} to write
\begin{align*}
{\sf Leib}_\dorf(c_1,c_2,c_3) + {\sf Leib}_\dorf(c_2,c_1,c_3) = -\llbracket\cD\langle c_1,c_2\rangle_C,c_3\rrbracket_\dorf = -\llbracket\rho^{-1}\circ\eta^{-1}{}^\sharp(\de\langle c_1,c_2\rangle_C),c_3\rrbracket_\dorf \ ,
\end{align*}
and the result then follows from the definition \eqref{eq:homrho} and the isometry property \eqref{eq:DFTisometry}.
\end{proof}

\begin{remark}\label{rem:strongconstraint}
The notion of DFT algebroid was introduced in~\cite{Jonke2018}, for the special case of a flat para-K\"ahler manifold in a local formulation based on a cotangent bundle $M=T^*\cQ$, and with the skew-symmetrization of the D-bracket (see also~\cite{Chatzistavrakidis:2019huz}). In that case, one can write local expressions for the Jacobiator and the map ${\sf hom}_\rho$ entirely in terms of the tangent bundle metric $\eta$ and the Schouten-Nijenhuis bracket of multivector fields on $M$ (see also~\cite{Grewcoe:2020gka}). In the double field theory literature, restricting to functions on $M$ and sections of $C$ (equivalently vector fields on $M$) for which the right-hand sides of the identities in Lemma~\ref{prop:DFTproperties} vanish is called imposing the `(strong) section constraint'; in other words, the section constraint is requirement $\langle\cD f,\cD g\rangle_C=\eta^{-1}(\de f,\de g)=0$ for all functions $f,g\in C^\infty(M)$.  Then via a suitable reduction or quotient, a DFT algebroid becomes a Courant algebroid on $M$.
Notice, however, that the well-known reduction procedure for Courant algebroids given in~\cite{Bursztyn2007} cannot be directly extended to metric algebroids, because it relies crucially on the closure of the Leibniz identity for the algebroid bracket. Hence the precise geometric interpretation of the section constraint remains an important open problem.
We shall discuss this point further in Section~\ref{sec:applications} within the setting of a DFT algebroid over a foliated almost para-Hermitian manifold $(M,K,\eta)$, where we shall see that it is possible to make such a reduction in a suitable sense. This will clarify the sense in which a DFT algebroid lies ``in between'' two Courant algebroids, and how doubled geometry is reconciled with generalized geometry.

However, even in this case one needs to exercise caution with these vanishing statements, because in a metric algebroid the maps ${\sf Leib}_\dorf$ and ${\sf hom}_\rho$ do not define tensors (cf. Remark~\ref{rem:Leibhom}).  Whereas in a Courant algebroid the condition ${\sf Leib}_\dorf(e_1,e_2,e_3)=0$ would yield the Bianchi identities for fluxes in supergravity, as discussed in Section~\ref{subsec:Courant}, in a general metric algebroid this condition depends on the choice of a local frame. Instead, it is proposed by~\cite{Carow-Watamura:2020xij} to replace this condition in a DFT algebroid with the tensorial `pre-Bianchi identity'
\begin{align*}
\big\langle{\sf Leib}_\dorf(c_1,c_2,c_3),c_4\big\rangle_C &= \eta\big({\sf hom}_\rho(c_1,c_3),{\sf hom}_\rho(c_2,c_4)\big) - \eta\big({\sf hom}_\rho(c_1,c_2),{\sf hom}_\rho(c_3,c_4)\big) \\
& \quad \, - \eta\big({\sf hom}_\rho(c_1,c_4),{\sf hom}_\rho(c_2,c_3)\big) \ ,
\end{align*}
for all $c_1,c_2,c_3,c_4\in\mathsf{\Gamma}(C)$. This defines a special class of DFT algebroids, including the standard ones of local double field theory that we will discuss in Section~\ref{subsec:localDFT}.
\end{remark}

\medskip

\subsection{C-Brackets and Curved $L_\infty$-Algebras} ~\\[5pt]
Let $(C,\llbracket\,\cdot\,,\,\cdot\,\rrbracket_\dorf,\langle\,\cdot\,,\,\cdot\,\rangle_C, \rho)$ be a DFT algebroid over a doubled manifold $(M,\eta)$. The analogue of the generalized Lie derivative from Section~\ref{subsec:Courantbracket},
\begin{align*}
{\boldsymbol\pounds}_c^\dorf := \llbracket c,\,\cdot\,\rrbracket_\dorf \ ,
\end{align*}
for $c\in\mathsf{\Gamma}(C)$, is now only an infinitesimal symmetry of the split signature pseudo-Euclidean vector bundle $(C,\langle\,\cdot\,,\,\cdot\,\rangle_C)$ over $M$, which is the natural notion of symmetry for a DFT algebroid. By Proposition~\ref{atiyahsymplectic}, these symmetries are encoded by the Atiyah algebroid
\begin{align}\label{eq:AtiyahDFT}
0 \longrightarrow \frso(C) \xlongrightarrow{} \IA(C, \braket{\,\cdot\,,\, \cdot\,}_{C}) \longrightarrow TM \longrightarrow 0 \ ,
\end{align}
which contains both infinitesimal diffeomorphisms of $M$ and orientation-preserving changes of orthonormal frame for  $C$.

In this case, closure of the gauge algebra is obstructed by the violation of the Leibniz identity. For this, we introduce the analogue of the Courant bracket from Section~\ref{subsec:Courantbracket},
\begin{align*}
\llbracket c_1,c_2\rrbracket_\cour := \tfrac12\,\big(\llbracket c_1,c_2\rrbracket_\dorf - \llbracket c_2,c_1\rrbracket_\dorf\big) \ ,
\end{align*}
which is called a \emph{C-bracket} of sections $c_1,c_2\in\mathsf{\Gamma}(C)$. By \eqref{eq:Courantmetric}, it is related to the D-bracket in the same way that the Courant bracket is related to the Dorfman bracket of a Courant algebroid:
\begin{align*}
\llbracket c_1,c_2\rrbracket_\dorf = \llbracket c_1,c_2\rrbracket_\cour + \tfrac12\, \cD\langle c_1,c_2\rangle_C \ .
\end{align*}
One can now equivalently characterize the compatibility conditions on the D-bracket of a DFT algebroid (and more generally any metric algebroid) in terms of the C-bracket by a statement completely analogous to Proposition~\ref{prop:CourantalgrboidC} without the Jacobiator identity; the latter is related to the map ${\sf Leib}_\dorf$ introduced  in Remark~\ref{rem:Leibhom}, and so the violation of the Jacobi identity for the C-bracket is controlled not only by the generalized exterior derivative of the Nijenhuis operator ${\sf Nij}_\cour:\mathsf{\Gamma}(C)\times\mathsf{\Gamma}(C)\times\mathsf{\Gamma}(C) \to C^\infty(M)$, but also by the section constraint $\langle \cD f,\cD g\rangle_C=\eta^{-1}(\de f,\de g)=0$. It is in this alternative formulation using the C-bracket that the notion of DFT algebroid was originally introduced in~\cite{Jonke2018}.

One can then write the commutator bracket of generalized Lie derivatives as
\begin{align*}
\big[\DFTLie_{c_1}^\dorf,\DFTLie_{c_2}^\dorf\big]_\circ(c) = \DFTLie_{\llbracket c_1,c_2\rrbracket_\cour}^\dorf(c) + {\sf Jac}_\cour(c_1,c_2,c) - \cD\,{\sf Nij}_\cour(c_1,c_2,c)  \ ,
\end{align*}
for all $c_1,c_2,c\in\mathsf{\Gamma}(C)$. In other words, the gauge algebra of generalized Lie derivatives only closes on the C-bracket upon imposition of the section constraint $\eta^{-1}(\de f,\de g)=0$ (see Remark~\ref{rem:strongconstraint}). As shown by~\cite{Grewcoe:2020gka}, the natural extension of Theorem~\ref{thm:LinftyCourant} to DFT algebroids involves a curving $\ell_0\neq0$ of the underlying $L_\infty$-algebra (a map of degree~$2$ from the ground ring $\IR$), in order to accomodate the non-vanishing map $\rho\circ\cD=\eta^{-1}{}^\sharp\circ\de\neq0$. This then completely characterizes the DFT algebroid, similarly to the case of Courant algebroids.

\begin{theorem}\label{thm:LinftyDFT}
Let $(C,\llbracket\,\cdot\,,\,\cdot\,\rrbracket_\dorf,\langle\,\cdot\,,\,\cdot\,\rangle_C, \rho)$ be a DFT algebroid over a flat doubled manifold $(M,\eta)$. Then there is a curved $L_\infty$-algebra on $L=L_{-1}\oplus L_0\oplus L_2$ with
\begin{align*}
L_{-1} = C^\infty(M) \ , \quad L_0 = \mathsf{\Gamma}(C) \qquad \mbox{and} \qquad L_2 = {\sf Span}_{\IR}(\eta^{-1}) \ ,
\end{align*}
whose non-zero brackets are given by
\begin{align*}
\ell_0(1) = \eta^{-1} \quad & \ \, , \ \, \quad
\ell_1(f) = \cD f \ , \\[4pt]
\ell_2(c_1,c_2) = \llbracket c_1,c_2\rrbracket_\cour \quad & \ \, , \, \ \quad
\ell_2(c_1,f) = \tfrac12\,\langle c_1,\cD f \rangle_C \ , \\[4pt]
\ell_3(c_1,c_2,c_3) &= -{\sf Nij}_\cour(c_1,c_2,c_3) \ , \\[4pt]
\ell_3(\eta^{-1},f,g) = \tfrac12\,\langle\cD f,\cD g\rangle_C \quad & \ \, , \, \ \quad
\ell_3(\eta^{-1},c_1,f) = \tfrac12\,\llbracket c_1,\cD f\rrbracket_\cour - \tfrac14\,\cD\langle c_1,\cD f\rangle_C \ , \\[4pt]
\ell_4(\eta^{-1},c_1,c_2,c_3) &= {\sf Jac}_\cour(c_1,c_2,c_3) - \cD\,{\sf Nij}_\cour(c_1,c_2,c_3) \ , \\[4pt]
\ell_5(\eta^{-1},c_1,c_2,c_3,c_4) &= \tfrac12\, {\sf Alt}_4\big\langle{\sf Jac}_\cour(c_1,c_2,c_3) - \cD\,{\sf Nij}_\cour(c_1,c_2,c_3),c_4\big\rangle_C \ ,
\end{align*}
for all $f,g\in C^\infty(M)$ and $c_1,c_2,c_3,c_4\in\mathsf{\Gamma}(C)$, where ${\sf Alt}_4$ is the alternatization map of degree~$4$.
\end{theorem}

\medskip

\subsection{AKSZ Construction of Doubled Sigma-Models} ~\\[5pt]
\label{subsec:AKSZDFT}
Using Theorem~\ref{thm:1-1metric}, it is easy to characterize a DFT algebroid $(C,\llbracket\,\cdot\,,\,\cdot\,\rrbracket_\dorf,\langle\,\cdot\,,\,\cdot\,\rangle_C, \rho)$ as a symplectic $2$-algebroid $(C[1]\oplus T^*[2]M,\{\,\cdot\,,\,\cdot\,\},\gamma)$ over a doubled manifold $(M,\eta)$. For example, from Lemma~\ref{prop:DFTproperties} and the derived bracket constructions of Section~\ref{subsec:symplectic2alg}, it can be characterized by the relation
\begin{align}\label{eq:sectionderived}
\{\{\{\gamma,\gamma\},f\},g\} + \{\{\{\gamma,\gamma\},g\},f\} &= 2\,\{\{\gamma,f\},\{\gamma,g\}\} \nonumber \\[4pt]
&= 2\, \langle\boldsymbol{\gamma}_1(\de f),\boldsymbol{\gamma}_1(\de g)\rangle_C \nonumber \\[4pt]
&= 2\,\eta^{-1}(\de f,\de g) \ ,
\end{align}
for all $f,g\in C^\infty(M)$, where we used the graded Jacobi identity for the Poisson bracket along with $\boldsymbol{\gamma}_1(\de f) = \{\gamma,f\} = \cD f$. Similarly, the other two identities of Lemma~\ref{prop:DFTproperties} can be written in terms of derived brackets using Remark~\ref{rem:Leibhom}. However, as $\{\gamma,\gamma\}\neq0$, a DFT algebroid does not correspond to a dg-manifold, and so we cannot apply AKSZ theory directly to write down a topological sigma-model whose BV formalism can be used to quantize the DFT algebroid. Note that all ingredients of the BV formalism, including the antibracket, are present except for the classical master equation.

In analogy to the AKSZ sigma-models of Section~\ref{subsec:Courant}, a three-dimensional topological sigma-model was associated to any DFT algebroid in~\cite{Jonke2018} by pulling back the fields of the Courant sigma-model corresponding to the large Courant algebroid, using the construction of Section~\ref{sec:metricfromlarge}. The graded symplectic geometry viewpoint of this construction was originally presented in~\cite{Samann2018}, and applied in~\cite{Kokenyesi:2018xgj} to AKSZ theory. The idea behind the construction is simple and is based on the way in which we motivated the definition of a DFT algebroid: Start with a Courant algebroid $(E,[\,\cdot\,,\,\cdot\,]_E,\langle\,\cdot\,,\,\cdot\,\rangle_E,\rho)$ of rank $4d$ over the doubled manifold $(M,\eta)$, introduce a generalized split signature metric on $E$, and then restrict the structure maps to obtain metric algebroid structures on the corresponding eigenbundles, as in Proposition~\ref{prop:DFTfromLarge}. This can be rephrased in the language of symplectic $2$-algebroids using Theorems~\ref{thm:QP2Courant} and~\ref{thm:1-1metric}. We will illustrate this in the simplest setting of the large Courant algebroid of Proposition~\ref{cor:DFTfromLarge} over a flat doubled manifold~$(M,\eta)$. 

The symplectic Lie $2$-algebroid corresponding to the large Courant algebroid over $M$ is, according to Example~\ref{ex:standardCourant} and Theorem~\ref{thm:QP2Courant}, based on the degree~$2$ manifold $\cM=T^*[2]T[1]M$ with coordinates $(x^I,\zeta^I,\chi_I,\xi_I)$ of degrees $(0,1,1,2)$. The degree~$2$ symplectic $2$-form \eqref{eq:omegaCourant} in this case is
\begin{align*}
\omega = \de \xi_I\wedge\de x^I + \de\chi_I\wedge\de\zeta^I \ ,
\end{align*}
while the degree~$3$ Hamiltonian \eqref{eq:gammaCourant} becomes
\begin{align*}
\gamma = \sqrt2 \ \xi_I\,\zeta^I \ ,
\end{align*}
which we have rescaled for convenience. 

Using the spit signature metric $\eta$ on the degree~$0$ body $M$, we restrict the tangent space coordinates to the diagonal using the coordinates
\begin{align*}
\tau_\pm^I = \tfrac1{\sqrt2}\,\big(\zeta^I \pm \eta^{IJ}\,\chi_J\big) \ ,
\end{align*}
where $\eta^{-1} = \frac12\,\eta^{IJ}\,\frac\partial{\partial x^I}\odot\frac\partial{\partial x^J}$. The submanifold of $\cM$ defined by the zero locus $\tau_-^I=0$ yields a symplectic $2$-algebroid whose underlying degree~$2$ manifold is $\cM_+:= T[1]M\oplus T^*[2]M$ with coordinates $(x^I,\tau_+^I,\xi_I)$ of degrees $(0,1,2)$, and with the pullbacks of the symplectic structure $\omega$ and Hamiltonian $\gamma$ to $\cM_+$ given by
\begin{align*}
\omega_+ = \de\xi_I\wedge\de x^I + \tfrac12\,\eta_{IJ}\,\de\tau_+^I\wedge \de\tau_+^J \qquad \mbox{and} \qquad \gamma_+ = \xi_I\,\tau_+^I \ .
\end{align*}
This Hamiltonian is not integrable, $\{\gamma_+,\gamma_+\} = \xi_I\,\eta^{IJ}\,\xi_J\neq0$, so $(\cM_+,\gamma_+)$ is not a dg-manifold. 

Under the correspondence of Theorem~\ref{thm:1-1metric}, the symplectic $2$-algebroid $(\cM_+,\omega_+,\gamma_+)$ can be identified with the metric algebroid $(TM,\llbracket\,\cdot\,,\,\cdot\,\rrbracket_\dorf^\eta,\eta,\unit_{TM})$ of Example~\ref{ex:metricalgeeta}. By Proposition~\ref{cor:DFTfromLarge}, this is isomorphic to the DFT algebroid $(C_+,\llbracket\,\cdot\,,\,\cdot\,\rrbracket_\dorf^+,\langle\,\cdot\,,\,\cdot\,\rangle_{C_+},\rho_+)$ on $(M,\eta)$ (for $b=0$), with corresponding degree~$2$ manifold $C_+[1]\oplus T^*[2]M$. This construction can be generalized along the lines of~\cite{Kokenyesi:2018xgj} to arbitrary anchor maps $\rho:\IT M\to TM$ and to arbitrary compatible twists of the Dorfman bracket on the generalized tangent bundle $\IT M$.

We choose the symplectic potential
\begin{align*}
\vartheta_+ = \xi_I \, \de x^I + \tfrac12\,\tau_+^I \,\eta_{IJ}\, \de\tau_+^J \ ,
\end{align*}
which is the pullback to $\cM_+$ of the Liouville $1$-form \eqref{eq:LiouvilleCourant} for the case of the large Courant algebroid.
We now pull back the degree~$0$ maps $\hat X:T[1]\Sigma_3\to\cM$, on which the AKSZ action functional \eqref{eq:AKSZ3daction} for the large Courant sigma-model is defined, to the submanifold $\tau_-^I=0$ to give maps $\hat X_+:T[1]\Sigma_3\to\cM_+$. Composing these with the isomorphism $\cM_+\to C_+[1]\oplus T^*[2]M$, the pullback of \eqref{eq:AKSZ3daction} to this mapping subspace then yields the action functional
\begin{align}\label{eq:DFTaction}
S_+(X,A_+,F) = \int_{\Sigma_3} \, \Big( & \langle F,\de X\rangle + \frac12\,\langle A_+,\de A_+\rangle_{C_+} \nonumber \\ & \quad \, - \langle F,(\rho_+\circ X)A_+\rangle + \frac1{3!} \,\langle A_+,\llbracket A_+,A_+\rrbracket_\dorf^+\rangle_{C_+} \Big) \ ,
\end{align}
where $X:\Sigma_3\to M$ is a smooth map from an oriented compact $3$-manifold $\Sigma_3$, while $A_+\in\mathsf{\Gamma}(T^*\Sigma_3\otimes X^*C_+)$ and $F\in\mathsf{\Gamma}(\midwedge^2T^*\Sigma_3\otimes X^*T^*M)$, with the same conventions as in \eqref{eq:AKSZ3daction}. This defines a canonical topological sigma-model associated to the DFT algebroid $(C_+,\llbracket\,\cdot\,,\,\cdot\,\rrbracket_\dorf^+,\langle\,\cdot\,,\,\cdot\,\rangle_{C_+},\rho_+)$ on the doubled manifold $(M,\eta)$, which we call a \emph{doubled sigma-model}. Although it is similar in form to the Courant sigma-model \eqref{eq:AKSZ3daction}, it is crucially different in many respects; in particular, it does not satisfy the BV master equation, so it is not an AKSZ sigma-model, nor can it be extended to define a BV quantized sigma-model.

\begin{remark}\label{rem:DFTgaugeinv}
Clearly the action functional \eqref{eq:DFTaction} can be written down for any DFT algebroid $(C,\llbracket\,\cdot\,,\,\cdot\,\rrbracket_\dorf,\langle\,\cdot\,,\,\cdot\,\rangle_C,\rho)$ over any doubled manifold $(M,\eta)$, not just the special instance in which we have derived it, though in the general case it cannot be derived from AKSZ theory. The reducible open gauge symmetries of \eqref{eq:DFTaction}, which are encoded through the Atiyah algebroid \eqref{eq:AtiyahDFT}, have been studied in detail by~\cite{Jonke2018,Chatzistavrakidis:2019rpp} by projecting the BRST symmetry of the large Courant sigma-model. It is found that gauge invariance and closure of the gauge algebra imply the analogue of the section constraint together with the axioms for a DFT algebroid. For a flat doubled manifold $(M,\eta)$, these give the Bianchi identities for the fluxes $T\in\Omega^3(M)$ defined by
\begin{align*}
T(X,Y,Z) = \eta(\llbracket X,Y\rrbracket_\dorf,Z)
\end{align*}
for doubled vectors $X,Y,Z\in\mathsf{\Gamma}(TM)$, which agrees with the flux formulation of double field theory~\cite{Geissbuhler2013}. In particular, the doubled sigma-model gives a unified description of geometric and non-geometric fluxes, whilst precluding as classical solutions several physically relevant string backgrounds which do not satisfy the section constraint~\cite{Jonke2018}.
\end{remark}

\begin{remark}\label{rem:DFTAKSZLinfty}
Gauge invariance of \eqref{eq:DFTaction} can also be understood in the AKSZ formalism and the associated (local) curved $L_\infty$-algebra from Theorem~\ref{thm:LinftyDFT}. By \eqref{eq:sectionderived} a natural constraint on the symplectic $2$-algebroid $(C[1]\oplus T^*[2]M,\{\,\cdot\,,\,\cdot\,\},\gamma\}$ corresponding to a DFT algebroid which imposes the section constraint is given by
\begin{align*}
\{\{\{\gamma,\gamma\},f\},g\} = 0 \ ,
\end{align*}
for all degree~$0$ functions $f,g\in C^\infty(M)$. This is a slight weakening of the defining condition of a symplectic almost Lie $2$-algebroid from Definition~\ref{def:almostLie2}, and it coincides with the coordinate-free formulation of the section constraint of double field theory in graded geometry originally presented by~\cite{Deser2015} (see also~\cite{Heller2016}); an alternative derived bracket formulation of a DFT algebroid is found in~\cite{Carow-Watamura:2020xij}. 

On imposing the section constraint, all brackets involving $\eta^{-1}$ in Theorem~\ref{thm:LinftyDFT} vanish and the curving $\ell_0$ may be dropped. The remaining brackets then govern the infinitesimal gauge symmetries of the DFT algebroid with the section constraint, and are formally the same as the flat $L_\infty$-algebra of a Courant algebroid from Theorem~\ref{thm:LinftyCourant}, as formulated originally in the graded geometry framework by ~\cite{Samann2018} as the gauge algebra underlying double field theory (see also~\cite{Hohm:2017pnh}). This is in harmony with the expectation that a DFT algebroid becomes a Courant algebroid when the section constraint is imposed.  We will discuss the section constraint further, as well as explicit solutions of the section constraint, in Section~\ref{sec:applications} below from a classical geometric perspective. The fact that a DFT algebroid can be characterized by a curved $L_\infty$-algebra suggests that it may be possible to formulate it as a dg-manifold~\cite{Grewcoe:2020gka}, though not necessarily one with a compatible symplectic structure; this perspective may allow for an AKSZ-type formulation of the doubled sigma-model as an unconstrained gauge theory which admits a larger set of classical solutions, as well as an extension of BV quantization.
\end{remark}

\section{Algebroids and Double Field Theory}
\label{sec:applications}
In this final section we apply the mathematical framework of this paper to a rigorous study of some kinematical issues in double field theory, including how it reduces to supergravity (in the NS--NS sector) and how T-duality is realized as a manifest symmetry in the doubled geometry formalism.

\medskip

\subsection{Local Double Field Theory}~\\[5pt]
\label{subsec:localDFT}
The standard local treatment of double field theory in the string theory literature~\cite{Siegel1993a,Siegel1993b,HullZw2009,Hull:2009zb,hohmhz,hullzw} is recovered in the case when $(M,K,\eta)$ is a flat para-K\"ahler manifold, and the DFT algebroid is the corresponding canonical metric algebroid $(TM,\llbracket\,\cdot\,,\,\cdot\,\rrbracket_\dorf^K,\eta,\unit_{TM})$~\cite{Vaisman2012}. In this case, the doubled manifold is locally a product of two $d$-dimensional subspaces $M=\cQ\times\tilde\cQ$. We can then write local coordinates $x=(x^I)$, $I=1,\dots,2d$ on $M$ which are adapted to the foliations in the splitting $TM=L_+\oplus L_-$ into integrable distributions $L_+=T\cQ$ and $L_-=T\tilde\cQ$, i.e. $x=(x^I)=(x^i,\tilde x_i)$, $i=1,\dots,d$. With respect to this splitting, the split signature metric has the matrix form
\begin{align*}
\eta = (\eta_{IJ}) = \bigg( \begin{matrix} 0 & \unit \\ \unit & 0 \end{matrix} \bigg) \ ,
\end{align*}
and we write its inverse as $\eta^{-1}=(\eta^{IJ})$ (with the same matrix form). The Levi-Civita connection $\nabla^\LC$ is trivial, and the local expression for the canonical D-bracket on two vector fields $X=X^I\,\frac\partial{\partial x^I}$ and $Y=Y^I\,\frac\partial{\partial x^I}$ is given by
\begin{align*}
\llbracket X,Y\rrbracket_\dorf^K = \Big(X^I\,\frac{\partial Y^J}{\partial x^I} - Y^I\,\frac{\partial X^J}{\partial x^I} + Y^I\, \eta_{IL} \, \frac{\partial X^L}{\partial x^M} \, \eta^{MJ} \Big) \, \frac\partial{\partial x^J} \ .
\end{align*}
The canonical C-bracket
\begin{align*}
\llbracket X,Y\rrbracket_\cour^K = \tfrac12\,\big( \llbracket X,Y\rrbracket_\dorf^K - \llbracket Y,X\rrbracket_\dorf^K \big)
\end{align*}
is thus the standard C-bracket of double field theory~\cite{Hull:2009zb}.

In this sense, standard double field theory is the flat space limit of Born geometry: in the local adapted coordinates, the Born metric is given by
\begin{align*}
\cH_0(g) = \bigg( \begin{matrix}
g & 0 \\ 0 & g^{-1}
\end{matrix} \bigg) \ ,
\end{align*}
for a Riemannian metric $g$. Then $B_+$-transformations by $2$-forms $b$ give the standard generalized metric\begin{align}\label{eq:localgenmetric}
\cH(g,b) = \bigg( \begin{matrix}
g -b\,g^{-1}\,b & b\,g^{-1} \\ -g^{-1}\,b & g^{-1}
\end{matrix} \bigg)
\end{align}
of double field theory~\cite{hullzw}. The D-bracket determines the infinitesimal gauge transformations of $\cH$ via the generalized Lie derivative.

The local form of the section constraint reads
\begin{align*}
\eta(\cD^{\tt can} f,\cD^{\tt can} g) = \eta^{-1}(\de f,\de g) = \frac{\partial f}{\partial x^I} \, \eta^{IJ} \, \frac{\partial g}{\partial x^J} = 0 \ .
\end{align*}
Solutions of this constraint select \emph{polarizations}, which are the $d$-dimensional `physical' null submanifolds of the doubled manifold $(M,\eta)$; these are also called \emph{duality frames}. Double field theory then reduces to supergravity in different duality frames, which are related to one another by T-duality transformations. For example, in the `supergravity frame' the section constraint is solved by choosing para-holomorphic functions
\begin{align*}
\frac{\partial f}{\partial \tilde x_i} = 0 \ ,
\end{align*}
after which the C-bracket reduces to the local form of the standard Courant bracket on the generalized tangent bundle $\IT\cQ$. With this solution of the section constraint, the metric \eqref{eq:localgenmetric} becomes the standard generalized metric of generalized geometry (cf. Example~\ref{ex:genmetricgengeom}).

\medskip

\subsection{Global Aspects of Double Field Theory} ~\\[5pt]
\label{subsec:globalDFT}
Para-Hermitian geometry arises as a framework for doubled geometry when one analyses the implications of the section constraint 
\begin{align*}
\eta^{-1}(\de f,\de g) = 0
\end{align*}
from a global perspective.
This can be solved by picking a maximally $\eta$-isotropic distribution $L_-\subset TM$ of the doubled manifold $(M,\eta)$ which is integrable. Then the section constraint is solved by foliated tensors; on functions this condition reads
\begin{align*}
\pounds_Xf=\iota_X\de f=0
\end{align*}
for $f\in C^\infty(M)$ and $X\in\mathsf{\Gamma}(L_-)$. With $L_-=T\cF$, this polarization selects the physical spacetime as a quotient $\cQ=M/\cF$ by the action on the leaves of the induced foliation $\cF$; foliated tensors are then those fields which are compatible with the surjective submersion from $M$ to the leaf space $\cQ$. The fact that the physical spacetime is a quotient, rather than a subspace, of the doubled manifold $(M,\eta)$ has been appreciated many times before in the double field theory literature, see e.g.~\cite{Hull2009,Vaisman2012,Park:2013mpa,Lee:2015xga}.

To put this into the context of the present paper, by a \emph{polarization} of a doubled manifold $(M,\eta)$ we shall mean a choice of almost para-Hermitian structure $(K,\eta)$ on $M$. A central mathematical problem in understanding how the kinematics of double field theory reduces to supergravity, under imposition of the section constraint, is to understand how doubled geometry reduces to generalized geometry. In generalized geometry~\cite{gualtieri:tesi,Hitchin2011}, a generalized vector is a section $X+\alpha$ of the generalized tangent bundle $\IT\cQ=T\cQ\oplus T^*\cQ$, with $X\in\mathsf{\Gamma}(T\cQ)$ a vector field and $\alpha\in\Omega^1(\cQ)$ a $1$-form on a manifold $\cQ$. In doubled geometry, a generalized vector is simply a vector field $X\in\mathsf{\Gamma}(TM)$ on the doubled manifold $(M,\eta)$. 

For this, we assume that the eigenbundle $L_-$ of the almost
 para-Hermitian manifold is involutive, i.e. it admits integral
 manifolds given by the leaves of a regular foliation $\cF.$ Given the projection map $P_-:TM\to TM$ of Definition~\ref{def:canconn} and the split signature metric $\eta$, define the \emph{$P_-$-projected} canonical D-bracket $\llbracket\,\cdot\,,\,\cdot\,\rrbracket_-$ by the formula
\begin{align*}
\eta(\llbracket X,Y\rrbracket_-,Z) = \eta(\nabla_{P_-(X)}^{\tt can}Y - \nabla_{P_-(Y)}^{\tt can}X,Z) + \eta(\nabla_{P_-(Z)}^{\tt can}X,Y) \ ,
\end{align*}
where $\nabla^{\tt can}$ is the canonical connection of the almost para-Hermitian manifold $(M,K,\eta)$. Then  this defines a Courant algebroid structure on the tangent bundle $TM$ by~\cite[Proposition~3.13]{Svoboda2018} (see also~\cite[Theorem~5.1.3]{Svoboda:2020msh}). In other words, this `projects' the canonical metric algebroid $(TM,\llbracket\,\cdot\,,\,\cdot\,\rrbracket_\dorf^K,\eta,\unit_{TM})$ to the Courant algebroid $(TM,\llbracket\,\cdot\,,\,\cdot\,\rrbracket_-,\eta,P_-)$ on the doubled manifold $(M,\eta)$. By Definition~\ref{def:DFTalgebroid}, this can be straightforwardly generalized to any DFT algebroid over a foliated almost para-Hermitian manifold $(M,K,\eta)$~\cite[Proposition~5.20]{Jonke2018}. 
 
Having established that a DFT algebroid can be projected to a Courant algebroid on $(M,\eta)$ when the section constraint is imposed, let us now examine what becomes of this Courant algebroid on an explicit solution of the section constraints.
We can construct the generalized tangent bundle $\mathbb{T}\cS= T\cS
 \oplus T^*\cS$ on any leaf $\cS$ of the foliation $\cF$. 
  Then there is a morphism from
 $\mathbb{T}\cS$ to $TM$ covering the inclusion
 $\cS\hookrightarrow M$, which is fibrewise bijective and is induced at the level of sections by the split signature metric $\eta$ through
\begin{align}\label{eq:pullbackmorphism}
{\tt p}_-:\mathsf{\Gamma}(\mathbb{T}\cS) \longrightarrow \mathsf{\Gamma}(TM) 
\ , \quad X+\alpha \longmapsto {\tt
  p}_-(X+\alpha)=X+\eta^{-1}{}^\sharp(\alpha) \ . 
\end{align}
By~\cite[Proposition~3.13]{Svoboda2018}, this defines a metric algebroid morphism from the standard Courant algebroid $(\IT\cS,[\,\cdot\,,\,\cdot\,]_\dorf,\langle\,\cdot\,,\,\cdot\,\rangle_{\IT\cS},\rho)$ on $\cS$ (see Example~\ref{ex:standardCourant}) to the Courant algebroid $(TM,\llbracket\,\cdot\,,\,\cdot\,\rrbracket_-,\eta,P_-)$ on $M$, that is,
\begin{align*}
{\tt p}_-\circ [\,\cdot\,,\,\cdot\,]_\dorf =  \llbracket\,\cdot\,,\,\cdot\,\rrbracket_-\circ({\tt p}_-\times{\tt p}_-) \ , \quad \langle\,\cdot\,,\,\cdot\,\rangle_{\IT\cS} = \eta\circ({\tt p}_-\times{\tt p}_-) \qquad \mbox{and} \qquad \rho = P_-\circ{\tt p}_- \ .
\end{align*}
Altogether, this relates the canonical metric algebroid $(TM,\llbracket\,\cdot\,,\,\cdot\,\rrbracket_\dorf^K,\eta,\unit_{TM})$ to the standard Courant algebroid on any leaf $\cS$ of the foliation $\cF$. Again, this construction can be straightforwardly generalized to any DFT algebroid over a foliated almost para-Hermitian manifold $(M,K,\eta)$~\cite[Proposition~5.27]{Jonke2018}. In this sense, doubled geometry recovers generalized geometry.

However, the relation to the generalized geometry of the physical spacetime, i.e. the standard Courant algebroid on $\IT\cQ$, is not so transparent in this framework. 
Let $\cQ=M/\cF$ be the leaf space of the foliation $\cF$ of $M$ defined by $L_-=T\cF$, and denote by $\sfq:M\to\cQ$ the quotient map. Given the splitting $TM=L_+\oplus L_-$ induced by $K,$ the vector bundle morphism $\de \sfq: TM \rightarrow T\cQ,$ covering $\sfq,$ is fibrewise bijective if restricted to $L_+,$ i.e. $\de \sfq \rvert_{L_+}: L_+ \rightarrow T\cQ$ is a fibrewise isomorphism. Hence the $C^\infty(M)$-module $\mathsf{\Gamma}(L_+)$ is isomorphic to the $C^\infty(\cQ)$-module $\mathsf{\Gamma}(T\cQ).$ The metric $\eta$ induces a vector bundle isomorphism $L_-\to L_+^*$ defined by $X\mapsto \eta^\flat(X)$, because $L_\pm$ are maximally isotropic with respect to $\eta$. Making further statements in this direction is part of the general open problem of reducing metric algebroids to Courant algebroids in a suitable sense (see Remark~\ref{rem:strongconstraint}).

\begin{remark}\label{rem:involutivealgebroid}
This solution of the section constraint of double field theory can be interpreted in terms of global objects as follows. Any involutive distribution $L_-=T\cF$ is naturally a Lie algebroid over $M$ with the restriction of the Lie bracket of vector fields $[\,\cdot\,,\,\cdot\,]_{L_-} = [\,\cdot\,,\,\cdot\,]_{TM}\big\rvert_{\mathsf{\Gamma}(L_-)\times\mathsf{\Gamma}(L_-)}$ and the inclusion of the subbundle $\ell_-:L_-\hookrightarrow TM$ as anchor map; this defines a Lie subalgebroid of the tangent Lie algebroid $(TM,[\,\cdot\,,\,\cdot\,]_{TM},\unit_{TM})$. The Lie algebroid $(L_-,[\,\cdot\,,\,\cdot\,]_{L_-},\ell_-)$ is naturally integrated by the holonomy groupoid ${\sf Hol}(\cF)\rightrightarrows M$ of the foliation, which provides a presentation of the leaf space as the quotient $\cQ=M/\cF$~\cite{Mrcun2003}. When $\cQ$ is a manifold, this is a Lie subgroupoid of the pair groupoid $M\times M\rightrightarrows M$ which integrates the tangent Lie algebroid on $M$. In this sense, the holonomy groupoid can be viewed as a smooth replacement for the leaf space.
\end{remark}

\begin{remark}\label{rem:otherstrong}
There are several global treatments of doubled geometry available in the literature which offer complementary interpretations of the section constraint of double field theory. Here we mention a few that are related to the perspectives offered in the present paper:
\begin{itemize}

\item On any para-Hermitian manifold $(M,K,\eta)$, the eigenbundles $L_\pm$ of $K$ naturally define a pair of Lie algebroids on $M$ by Remark~\ref{rem:involutivealgebroid}. Then the section constraint can be interpreted as a compatibility condition on a pair of D-structures $(L_+,L_-)$ in the canonical metric algebroid $(TM,\llbracket\,\cdot\,,\,\cdot\,\rrbracket_\dorf^K,\eta,\unit_{TM})$, which implies that the tangent bundle $TM$ becomes a Courant algebroid on $M$~\cite{Mori2019}. In other words, the canonical metric algebroid is composed of a double of Lie algebroids, analogous to the Drinfel'd double of a pair of Lie algebras (see also~\cite{Mori:2020yih}).

\item A global formulation of doubled geometry based on higher geometry appears in~\cite{Alfonsi:2019ggg} within the framework of double field theory on the total (simplicial) space of a bundle gerbe, regarded as a $\mathsf{U}(1)$-principal $2$-bundle (see also the contribution~\cite{Alfonsi:2021uwh} to this special issue). In this setting the section constraint is interpreted as invariance under the principal $\mathsf{BU}(1)$-action, and para-Hermitian manifolds appear as an atlas for the bundle gerbe. This framework clarifies and makes precise previous patching constructions using finite gauge transformations in double field theory~\cite{Park:2013mpa,Hohm2013,Berman:2014jba,Hull:2014mxa,Howe:2016ggg}.

\item A \emph{rack} is a global group-like object whose infinitesimal counterpart is a Leibniz-Loday algebra. 
A global object integrating a metric algebroid, called a pre-rackoid, has been suggested by~\cite{Ikeda:2020lxz}. This is a weakening of the notion of a rackoid, which is a groupoid-like generalization of a rack, and which is the global structure corresponding to a Leibniz-Loday algebroid that can be used to integrate Courant algebroids. Explicit realizations of pre-rackoids are given in~\cite{Ikeda:2020lxz} for the canonical metric algebroid over any para-Hermitian manifold, which reduce to a rackoid when the section constraint of double field theory is imposed; these pre-rackoids can also be implemented in the corresponding topological doubled sigma-model of Section~\ref{subsec:AKSZDFT}. These structures are relevant to the understanding of finite gauge transformations in double field theory~\cite{Park:2013mpa,Hohm2013,Berman:2014jba,Hull:2014mxa,Howe:2016ggg}.

\item On any foliated flat almost para-Hermitian manifold $(M,K,\eta)$, solutions of the section constraint can be understood~\cite{Grewcoe:2020gka} as an $L_\infty$-morphism from the curved $L_\infty$-algebra of the canonical metric algebroid over $(M,K,\eta)$, given by Theorem~\ref{thm:LinftyDFT}, to the flat $L_\infty$-algebra of the standard Courant algebroid over the leaf space $\cQ$ of the foliation, given by Theorem~\ref{thm:LinftyCourant}. 

\end{itemize}
\end{remark}

\medskip

\subsection{Recovering the Physical Background Fields} ~\\[5pt]
A central problem in understanding the global formulation of the dynamics of double field theory is to investigate the quotient $\cQ=M/\cF$ for a foliated almost para-Hermitian manifold $(M, K, \eta)$ endowed with a generalized metric $\cH.$ We will do this by first recalling a more general result due to Kotov and Strobl~\cite{Kotov2014, Kotov2018}. 

Let $(M,\cH)$ be any Riemannian manifold. Let $( A, [\,\cdot\,,\, \cdot\,]_{ A}, \sfa)$ be a Lie algebroid over $M$ endowed with a linear connection $\nabla$, and define the representation of $ A$ on $TM$ by the flat $ A$-connection 
\be \nonumber
\prescript{\tau}{}{\nabla} \colon \mathsf{\Gamma}( A) \times \mathsf{\Gamma}(TM) \longrightarrow \mathsf{\Gamma}(TM)
\ee
given by 
\be\nonumber
\prescript{\tau}{}{\nabla}_a X \coloneqq [\sfa(a), X]_{TM}  + \sfa(\nabla_X a) \ , 
\ee
for all $a \in \mathsf{\Gamma}( A)$ and $X \in \mathsf{\Gamma}(TM).$

\begin{definition}
The triple $( A, \nabla, \cH)$ is a \emph{Killing Lie algebroid} if 
\be \label{compAlgH}
\prescript{\tau}{}{\nabla}\cH=0 \ .
\ee
\end{definition}

The Killing vector fields for $\cH$ are given by $X= \sfa(a)$ for any covariantly constant section $a \in \mathsf{\Gamma}( A).$ Killing Lie algebroids are related to quotients by

\begin{proposition}\label{prop:Killingsubmersion}
Let $( A, [\,\cdot\,,\, \cdot\,]_{ A}, \sfa)$ be a Lie algebroid over a Riemannian manifold $(M, \cH)$, endowed with a linear conection $\nabla$, whose anchor map $\sfa$ is injective and has constant rank, so that its image $\Im(\sfa)$ defines a regular foliation $\cF$ of $M.$ Then $( A, \nabla, \cH)$ is a Killing Lie algebroid if and only if there is a Riemannian submersion
\be \nonumber
\sfq \colon (M, \cH) \longrightarrow (\cQ, g)
\ee
where $\cQ=M/\cF$ is the leaf space of the foliation $\cF$ and $g$ is a Riemannian metric on $\cQ.$
\end{proposition}

This result allows us to understand under which circumstances the quotient implementing the section constraint exists. For further details and proofs see \cite{Kotov2014, Kotov2018,Marotta:2019eqc}.

\begin{remark} \label{submersGroupoid}
Proposition~\ref{prop:Killingsubmersion} can be interpreted globally from a Lie groupoid perspective~\cite{delHoyo2018}. The existence of a Riemannian submersion $\sfq \colon (M, \cH) \rightarrow (\cQ, g)$ is equivalent to the statement that the submersion groupoid $M \times_\cQ M\rightrightarrows \cQ$ is endowed with a $0$-metric, i.e.~a metric which is invariant under the canonical action of $M \times_\cQ M$ on its base manifold~$\cQ.$

Conversely, for a submersion $\sfq \colon M \rightarrow \cQ$ to be Riemannian it suffices to check for the existence of a $1$-metric on $M \times_\cQ M\rightrightarrows \cQ$, i.e. a metric on the manifold of arrows which is transverse with respect to the source map and for which the inversion map is an isometry, because it induces a $0$-metric. It is further shown in \cite{delHoyo2018} that any $0$-metric on $M \times_\cQ M\rightrightarrows\cQ$ can always be extended to a $1$-metric.
\end{remark}

In order to understand the condition \eqref{compAlgH} let us discuss further the case of a regularly foliated base Riemannian manifold. Choose an orthogonal splitting $ s_\perp$ of the canonical short exact sequence
\be \label{normalsequence}
0 \longrightarrow T\cF \xlongrightarrow{} TM \xlongrightarrow{} \nu(\cF) \longrightarrow 0 
\ee
where $\nu(\cF)$ is the normal bundle of the foliation. Then $TM \simeq \Im( s_\perp)\oplus T\cF$, and with respect to this splitting the Riemannian metric takes the form
\be \label{diagonalform}
\cH= \bigg(\begin{matrix} g_\perp & 0 \\ 0 & g_\parallel \end{matrix} \bigg) \ , 
\ee
where $g_\perp$ is a fibrewise metric on $\Im( s_\perp)$ and $g_\parallel$ is a fibrewise metric on $T\cF.$ Therefore the condition \eqref{compAlgH} is equivalent to~\cite{Kotov2014, Kotov2018}
\be \label{riemfol}
\pounds_{X_\parallel} g_\perp=0 \ ,
\ee
for all $X_{\parallel} \, \in \mathsf{\Gamma}(T\cF),$ which states a further equivalence with the requirement that $\cH$ is a \emph{bundle-like metric} on $M,$ see~\cite{Marotta:2019eqc}. 
Then \eqref{riemfol} makes $(M, \cF, g_\perp)$ into a Riemannian foliation. Clearly, when a Riemannian submersion $\sfq \colon (M, \cH) \rightarrow (\cQ, g)$ exists, then $g_\perp= \sfq^* g .$

\begin{remark}
A Riemannian foliation $(M, \cF, g_\perp)$ induces a $0$-metric on the holonomy groupoid ${\mathsf{Hol}}(\cF) \rightrightarrows M$. Again, any $0$-metric on ${\mathsf{Hol}}(\cF) \rightrightarrows M$ can be extended to a $1$-metric \cite{delHoyo2018}.  Conversely, as in Remark \ref{submersGroupoid}, the existence of a $1$-metric on  ${\mathsf{Hol}}(\cF) \rightrightarrows M$ implies the existence of a Riemannian foliation on~$(M, \cF).$
\end{remark}

For an almost para-Hermitian manifold $(M, K, \eta)$ endowed with a generalized metric $\cH,$ characterized by the pair $(g_+, b_+)$ according to Proposition~\ref{gbparaherm}, we assume the eigenbundle $L_-$ of $K$ is integrable, that is, $L_-= T\cF,$ where $\cF$ is the induced foliation. We further assume that the leaf space $\cQ=M/\cF$ is a manifold. Then the splitting $ s_\perp$ of \eqref{normalsequence} corresponds to the para-Hermitian structure given by the $B_+$-transformation of $K$ induced by the $2$-form $b_+$. Thus $\cH$ takes the diagonal form \eqref{diagonalform} with $g_\perp= g_+,$ that is, $(TM,K_{B_+},\eta,\cH)$ is a Born vector bundle on $M$. In this case the Killing Lie algebroid structure on $L_-=T\cF$ is characterized by the corresponding Bott connection on $TM,$ as discussed in \cite{Marotta:2019eqc}, whereby the Riemannian metric on $M$ is used to construct the corresponding connection on~$TM.$

On the other hand, any $B_+$-transformation preserves the foliation $L_-=T\cF$ and induces a splitting of $TM$ such that the transformed generalized metric $\cH_{B_+}$ has only a different $g_\parallel$ component. In other words, $B_+$-transformations preserve the Riemannian foliation $(M, \cF, g_+).$ Thus when the quotient map $\sfq:M \rightarrow \cQ$ is a Riemannian submersion, it remains the same for all the $B_+$-transformed generalized metrics. Similarly, any diffeomorphism $\phi \in {\mathsf{Diff}}(M)$ preserving the Riemannian foliation $(M, \cF, g_+)$ such that $\phi^* \eta=\eta$ induces a new para-Hermitian structure with a transformed generalized metric $\cH_\phi$, but which preserves the quotient; in other words, $(M, \cH_\phi)$ is still mapped into $(\cQ, g)$ with $g_+= \sfq^*g.$ 

The $B_+$-transformed subbundle $e^{B_+}(L_+)$ is no longer isotropic with respect to the fundamental 2-form $\omega$, and one has
\be \nonumber
\omega\big(e^{B_+}(X_+), e^{B_+}(Y_+)\big)= 2\,b_+(X_+, Y_+) \ ,
\ee
for all $X_+,Y_+\in\mathsf{\Gamma}(L_+)$.
If the 2-form $b_+$ is transversally invariant, i.e. $\pounds_{X_-} b_+ =0$ for all $X_-\in\mathsf{\Gamma}(L_-)$, then  the leaf space admits a 2-form $b \in \Omega^2(\cQ)$ such that $b_+= \sfq^* b.$
In other words, the leaf space $\cQ$ becomes a string target space whose background fields (in the NS--NS sector) are given by the pair~$(g,b).$ 

\begin{remark}\label{rem:genmetricpullback}
Following the treatment of Section~\ref{subsec:globalDFT}, a generalized metric on an almost para-Hermitian manifold can also be related to a generalized metric on a generalized tangent bundle. One shows that the vector bundle morphism \eqref{eq:pullbackmorphism} pulls back a generalized metric on a foliated almost
para-Hermitian manifold, with the foliation associated with the almost
para-complex structure, to a generalized metric on the generalized
tangent bundle $\mathbb{T}\cS$ constructed on any leaf  $\cS$
of the foliation $\cF$. 
\end{remark}

\begin{remark}\label{rem:leafspaceorb}
If we relax the requirement that the leaf space $\cQ=M/\cF$ is a manifold, then these constructions can be used to provide natural geometric realizations of the `non-geometric backgrounds' of string theory, see e.g.~\cite{Dabholkar2002,Hull2005,Shelton2005,Dabholkar2005,Hull2007}.
Following the standard terminology~\cite{Hull2005}, if the foliation defines a singular quotient, then the physical spacetime $\cQ$ is called a \emph{T-fold}; a typical class of examples are the orbifolds that arise from foliations with compact leaves and finite leaf holonomy group~\cite{Mrcun2003}. For a T-fold, the holonomy groupoid ${\mathsf{Hol}}(\cF) \rightrightarrows M$ is no longer a Lie subgroupoid of the pair groupoid $M\times M\rightrightarrows M$. On the other hand, in the non-integrable case, where there is no foliation of $M$ at all and hence no solution of the section constraint, there is no physical spacetime and $M$ is an \emph{essentially doubled space} in the terminology of~\cite{Hull:2019iuy}; see~\cite{Marotta:2019eqc} for further discussion and details, as well as many explicit examples. 
\end{remark}

\medskip

\subsection{Generalized T-Duality} ~\\[5pt]
Double field theory originated as an attempt to extend supergravity, which is described by generalized geometry, into a theory which is manifestly symmetric under the fundamental T-duality symmetry of string theory, that exchanges distinct physical spacetimes and background fields: in doubled geometry T-duality is realized as suitable diffeomorphisms of a doubled manifold. Let us now discuss how this fits into the treatment of the present paper. For this, we introduce a notion of T-duality for almost para-Hermitian manifolds endowed with a generalized metric, starting from the natural notion of symmetries of para-Hermitian vector bundles.

\begin{proposition}\label{prop:Oddvecbun}
Let $\vartheta \in {\sf Aut}(E)$ be an automorphism of a para-Hermitian vector bundle $(E, K, \eta)$ of rank $2d$ which is an isometry of the split signature metric $\eta.$ Then the para-Hermitian structure $(K, \eta)$ is mapped by $\vartheta$ into another para-Hermitian structure $(K_\vartheta,\eta),$ where $K_\vartheta= \vartheta^{-1}\circ K\circ\vartheta$. 
\end{proposition}

\begin{remark}\label{rem:Oddgenmetric}
The automorphisms of Proposition~\ref{prop:Oddvecbun} form a subgroup of ${\sf Aut}(E)$ denoted by ${\sf O}(d,d)(E).$ Any element $\vartheta \in{\sf O}(d,d)(E)$ maps a generalized metric $\cH$ on $(E,K,\eta)$ into another generalized metric $\cH_\vartheta$ on $(E,K_\vartheta,\eta)$.
\end{remark}

In the applications to doubled geometry, we take $E=TM,$ and write ${\sf O}(d,d)(M)$ for ${\sf O}(d,d)(E)$. In this case
the transformations of Proposition~\ref{prop:Oddvecbun} have been identified as \emph{generalized T-dualities} in~\cite{Marotta:2019eqc}, which encompass many known examples, including non-abelian T-duality transformations. They naturally induce changes of polarization $(K,\eta)$ for solutions of the section constraint on a doubled manifold $(M,\eta)$. The doubled geometry viewpoint allows for an interpretation of the usual notion of T-duality by establishing a correspondence between quotients of a doubled manifold with respect to different foliations. 

For this, let $(M,\eta)$ be a foliated doubled manifold endowed with an almost para-Hermitian structure $(K, \eta)$ and a generalized metric $\cH$ such that $(M, \cF, g_+)$ is a Riemannian foliation, where $L_-=T\cF$ is the integrable $-1$-eigenbundle of $K$ and $(g_+, b_+)$ is the pair identifying $\cH$ in the splitting of the tangent bundle $TM$ given by $K$. Then a T-duality transformation is given by an $\eta$-isometric diffeomorphism $\phi$ of $M$ that maps the triple $(K, \eta, \cH)$ into $(K_\phi, \eta, \cH_\phi)$, and the foliation $\cF$ into a different foliation $\cF_\phi.$ We require that $(M, \cF_\phi, {g_+}_\phi)$ be a Riemannian foliation, where $({g_+}_\phi, {b_+}_\phi)$ is the pair identifying the generalized metric $\cH_\phi$ in the splitting given by $K_\phi$. This construction is depicted by the diagram
\begin{center}
\begin{tikzcd}
(M, \cF, \cH) \arrow[r, "\phi "] \arrow[d, "\sfq "'] & (M, \cF_\phi, \cH_\phi) \arrow[d, "\sfq_\phi "] \\
(\cQ, g_+) \arrow[r, dashed, swap, "\ccT "] & (\cQ_\phi, {g_+}_\phi)
\end{tikzcd}
\end{center}
where the dashed arrow (indicatively) defines the T-duality $\ccT$ from the leaf space $\cQ=M/\cF$ to the leaf space $\cQ_\phi=M/\cF_\phi$ via this diagram. 

Here we do not demand that the leaf spaces be endowed with smooth structures. For instance, this construction makes sense when the leaf spaces admit an orbifold structure, see Remark~\ref{rem:leafspaceorb}. Thus it may happen that a T-duality transformation takes a geometric background, with smooth leaf space $\cQ$, to a T-fold. It may also happen that the eigenbundle ${L_-}_\phi$ of $K_\phi$ is not integrable; this corresponds to a generalized T-duality which sends a geometric background to an essentially doubled space. These are the ways in which the prototypical non-geometric backgrounds of string theory arise (see e.g.~\cite{Szabo2018,Plauschinn:2018wbo} for reviews).

Notice that diffeomorphisms $\phi$ which preserve the Riemannian foliation $(M, \cF, g_+)$ give trivial T-duality transformations $\ccT$. We can also extend this constuction beyond diffeomorphisms of $M$ to more general automorphisms $\vartheta\in{\sf O}(d,d)(M)$ of the tangent bundle $TM$. In particular, $B_+$-tranformations preserve the foliation, i.e. the eigenbundle $L_-,$ and so give trivial T-dualities as well.

\begin{remark}
There is a natural equivalence relation on foliated manifolds called `Haudorff Morita equivalence' that preserves regular foliations and induces Morita equivalent holonomy groupoids, see~\cite{Garmendia2018}. Applying this notion to the case at hand, two foliations $\cF$ and $\cF'$ of $M$ are Hausdorff Morita equivalent if there exists a manifold $P$ and two surjective submersions $\pi,\pi':P\to M$ with connected fibres such that $\pi^{-1}\cF=\pi'{}^{-1}\cF'$:
\begin{center}
\begin{tikzcd}
 & P \arrow[dl,"\pi",swap] \arrow[dr,"\pi'"] & \\
 (M,\cF) &  & (M,\cF')
\end{tikzcd}
\end{center}
Then the leaf spaces $\cQ=M/\cF$ and $\cQ'=M/\cF'$ are homeomorphic, and the transverse geometry at corresponding leaves is the same.

The construction above is equivalent to saying that a T-duality transformation is given by two Hausdorff Morita equivalent Riemannian foliations where the equivalence classes are induced by restricting to $\eta$-isometric diffeomorphisms. 
It might be argued that diffeomorphisms which preserve a Riemannian foliation form a subclass of the class of Hausdorff Morita equivalent foliations. Thus a chain of T-duality transformations might be given by different Hausdorff Morita equivalent subclasses of Riemannian foliations inside a Hausdorff Morita equivalence class of foliations for the doubled manifold $(M, \eta).$ Then different Hausdorff Morita equivalence classes correspond to different T-duality chains.
\end{remark}

%\bigskip

\bibliographystyle{ieeetr}
\bibliography{bibprova1}

%\begin{thebibliography}{100}

%\end{thebibliography}

\end{document}